\documentclass{lmcs} 
\pdfoutput=1
\usepackage[utf8]{inputenc}

\usepackage{lastpage}
\lmcsdoi{20}{1}{14}
\lmcsheading{}{\pageref{LastPage}}{}{}%
{Oct.~14,~2022}{Feb.~13,~2024}{}


\usepackage{hyperref}
\theoremstyle{plain}\newtheorem{observation}[thm]{Observation}
\usepackage{graphicx}

\RequirePackage{xspace,xcolor}
\usepackage{xifthen}
\newlength{\breite}
\newcommand{\Math}[1]{\ensuremath{#1}\xspace}
\DeclareMathOperator*{\domain}{{\sf dom}}
\DeclareMathOperator*{\var}{\operatorname{var}}
\DeclareMathOperator*{\range}{{\sf rng}}

\DeclareMathOperator*{\val}{\operatorname{Val}}
\DeclareMathOperator*{\set}{\operatorname{Set}}
\DeclareMathOperator*{\con}{\operatorname{con}}
\DeclareMathOperator*{\ann}{\operatorname{anno}}
\RequirePackage{bussproofs/bussproofs}

\newcommand{\ComplexityClassFont}[1]{\mathsf{#1}}
\newcommand{\class}[1]{\ensuremath{\mathsf{#1}}}
\newcommand{\Pe}{\class{P}}

\newcommand{\PSPACE}{\class{PSPACE}}
\newcommand{\C}{\class{C}}

\newcommand{\ProofSystemsFont}[1]{\mathsf{#1}}

\newcommand{\DefineProofSystem}[2]{\expandafter\def\csname#1\endcsname{\Math{\ComplexityClassFont{#2}}}}

\DeclareMathOperator*{\lev}{\operatorname{lv}}
\DeclareMathOperator*{\ind}{\operatorname{ind}}

\newcommand{\Red}{\Math{\boldsymbol{\forall}\kern .04ex\ProofSystemsFont{red}}}
\newcommand{\Ared}{$\forall$\text{-Red}\xspace}
\newcommand{\red}{\Math{\,\raisebox{.2ex}{$\scriptstyle+$}\,\Red}}
\newcommand{\Frege}[1][]{\Math{\ifthenelse{\isempty{#1}}{\ProofSystemsFont{Frege}}{#1\text{-}\ProofSystemsFont{Frege}}}}

\DefineProofSystem{eFrege}{e\Frege[]}
\newcommand{\FregeRed}[1][]{\Frege[#1]\!\!\red}
\newcommand{\eFregeRed}[1][]{\eFrege\!\!\red}

\newcommand{\Dz}[1]{\ensuremath{\mathcal{D}^{\mbox{\scriptsize \upshape #1}}}\xspace}
\newcommand{\Drrs}{\Dz{rrs}}

\DeclareMathOperator*{\instantiate}{\textsf{inst}}

\newcommand{\ecalculus}{$\forall$\textsf{Exp+Res}\xspace}

\newcommand{\qrc}{\textsf{Q-Res}\xspace}
\newcommand{\qrat}{\textsf{QRAT}\xspace}
\newcommand{\mrc}{\textsf{M-Res}\xspace}

\newcommand{\qurc}{\textsf{QU-Res}\xspace}
\newcommand{\lqrc}{\textsf{LD-Q-Res}\xspace}
\newcommand{\qdrc}{\textsf{Q(\Drrs)-Res}\xspace}
\newcommand{\lqdrc}{\textsf{LD-Q(\Drrs)-Res}\xspace}
\newcommand{\lquprc}{\textsf{LQU}$^+$\textsf{-\hspace{1pt}Res}\xspace}

\newcommand{\irc}{\textsf{IR-calc}\xspace}
\newcommand{\irmc}{\textsf{IRM-calc}\xspace}

\DefineProofSystem{Gfull}{G}

\newcommand{\Select}{\texttt{Select}\xspace}
\newcommand{\Merge}{\texttt{Merge}\xspace}

\DeclareMathOperator*{\rest}{{\sf restrict}}
\newcommand{\restr}[2]{\rest_{{#1}} ({#2})}
\newcommand{\complete}[2]{{#1}\;\circ\;{#2}}
\newcommand{\comprehension}[2]{\ensuremath{\left\{ {#1} \;|\; {#2}\right\}}}

\usepackage{tikz}
\usetikzlibrary{positioning}

\tikzstyle{calcn}=[rectangle%
,rounded corners=1mm%
,thick%
,draw=black,%
,top color=white,bottom color=black!20,%
minimum height=.5cm,minimum width=.5cm,inner sep=1pt%
]
\tikzstyle{legn}=[font=\scriptsize]

\DeclareMathOperator*{\diff}{\operatorname{Dif}}
\DeclareMathOperator*{\equ}{\operatorname{Eq}}

\begin{document}
\title{Towards Uniform Certification in QBF} 

\author[L.~Chew]{Leroy Chew\lmcsorcid{0000-0003-0226-2832}}[a]
\address{TU Wien, Vienna, Austria \url{https://leroychew.wordpress.com/} }
\email{lchew@ac.tuwien.ac.at}

\author[F.~Slivovsky]{Friedrich Slivovsky\lmcsorcid{0000-0003-1784-2346}}[b]
\address{TU Wien, Vienna, Austria \url{http://www.ac.tuwien.ac.at/people/fslivovsky/}}
\email{fslivovsky@ac.tuwien.ac.at}

\thanks{This work was supported by the Vienna Science and Technology Fund (WWTF) under grant ICT19-060.}

\begin{abstract}
  We pioneer a new technique that allows us to prove a multitude of previously open simulations in QBF proof complexity. In particular, we show that extended QBF Frege p-simulates clausal proof systems such as IR-Calculus, IRM-Calculus, Long-Distance Q-Resolution, and Merge Resolution.
  These results are obtained by taking a technique of Beyersdorff et al. (JACM 2020) that turns strategy extraction into simulation and combining it with new local strategy extraction arguments.

This approach leads to simulations that are carried out mainly in propositional logic, with minimal use of the QBF rules. Our proofs therefore provide a new, largely propositional interpretation of the simulated systems.
We argue that these results strengthen the case for uniform certification in QBF solving, since many QBF proof systems now fall into place underneath extended QBF Frege.

\end{abstract}
\maketitle

\section{Introduction}
\label{sec:intro}
The problem of evaluating Quantified Boolean Formulas (QBF), an extension of propositional satisfiability (SAT), is a canonical \PSPACE-complete problem~\cite{SM73,AB09}. 
Many tasks in verification, synthesis and reasoning have succinct QBF encodings~\cite{ShuklaBPS19}, making QBF a natural target logic for automated reasoning.
As such, QBF has seen considerable interest from the SAT community, leading to the development of a variety of QBF solvers (e.g.,~\cite{LonsingB10,JanotaKMC16,rabe2015caqe,JanotaM15,PeitlSS19}). The underlying algorithms are often highly nontrivial, and their implementation can lead to subtle bugs~\cite{BrummayerLB10}. While formal verification of solvers is typically impractical, trust in a solver's output can be established by having it generate a proof trace that can be externally validated.
This is already standard in SAT solving with the \textsf{DRAT} proof system~\cite{WetzlerHH14}, for which even formally verified checkers are available~\cite{Cruz-FilipeHHKS17}. A key requirement for standard proof formats like \textsf{DRAT} is that they \emph{simulate} all current and emerging proof techniques.

Currently, there is no decided-upon checking format for QBF proofs (although there have been some suggestions~\cite{Jus07,HeuleSB17}).
The main challenge of finding such an universal format, is that QBF solvers are so radically different in their proof techniques, that each solver basically works in its own proof system. For instance, solvers based on CDCL and (some) clausal abstraction solvers can generate proofs in Q-resolution (\qrc)~\cite{KBKF95} or long-distance Q-resolution (\lqrc)~\cite{DBLP:journals/fmsd/BalabanovJ12}, while the proof system underlying expansion based solvers combines instantiation of universally quantified variables with resolution (\ecalculus)~\cite{JM15}. Variants of the latter system have been considered: \irc (\textbf{I}nstantiation \textbf{R}esolution) admits instantiation with partial assignments, and \irmc (\textbf{I}nstantiation \textbf{R}esolution \textbf{M}erge) additionally incorporates elements of long-distance Q-resolution~\cite{BCJ19}.

A universal checking format for QBF ought to simulate all of these systems.
A good candidate for such a proof system has been identified in extended QBF Frege (\eFregeRed): 
Beyersdorff et al. showed~\cite{BBCP20} that a lower bound for \eFregeRed would not be possible without a major breakthrough.

In this work, we show that \eFregeRed does indeed p-simulate  \irmc, Merge Resolution (\mrc) and \lquprc (a generalisation of \lqrc), thereby establishing \eFregeRed and any stronger system (e.g., \qrat~\cite{HeuleSB17} or \Gfull~\cite{KP90}) as potential universal checking formats in QBF. 
As corollaries, we obtain (known) simulations of \ecalculus~\cite{KHS17} and \lqrc~\cite{KS19} by \qrat, as well as a (new) simulation of \irc by \qrat, answering a question recently posed by Chede and Shukla~\cite{ChedeS21}.
A simulation structure with many of the known QBF proof systems and our new results is given in Figure~\ref{fig:simstructure}.

\begin{figure}[h]
  \centering
  \includegraphics[scale=0.275]{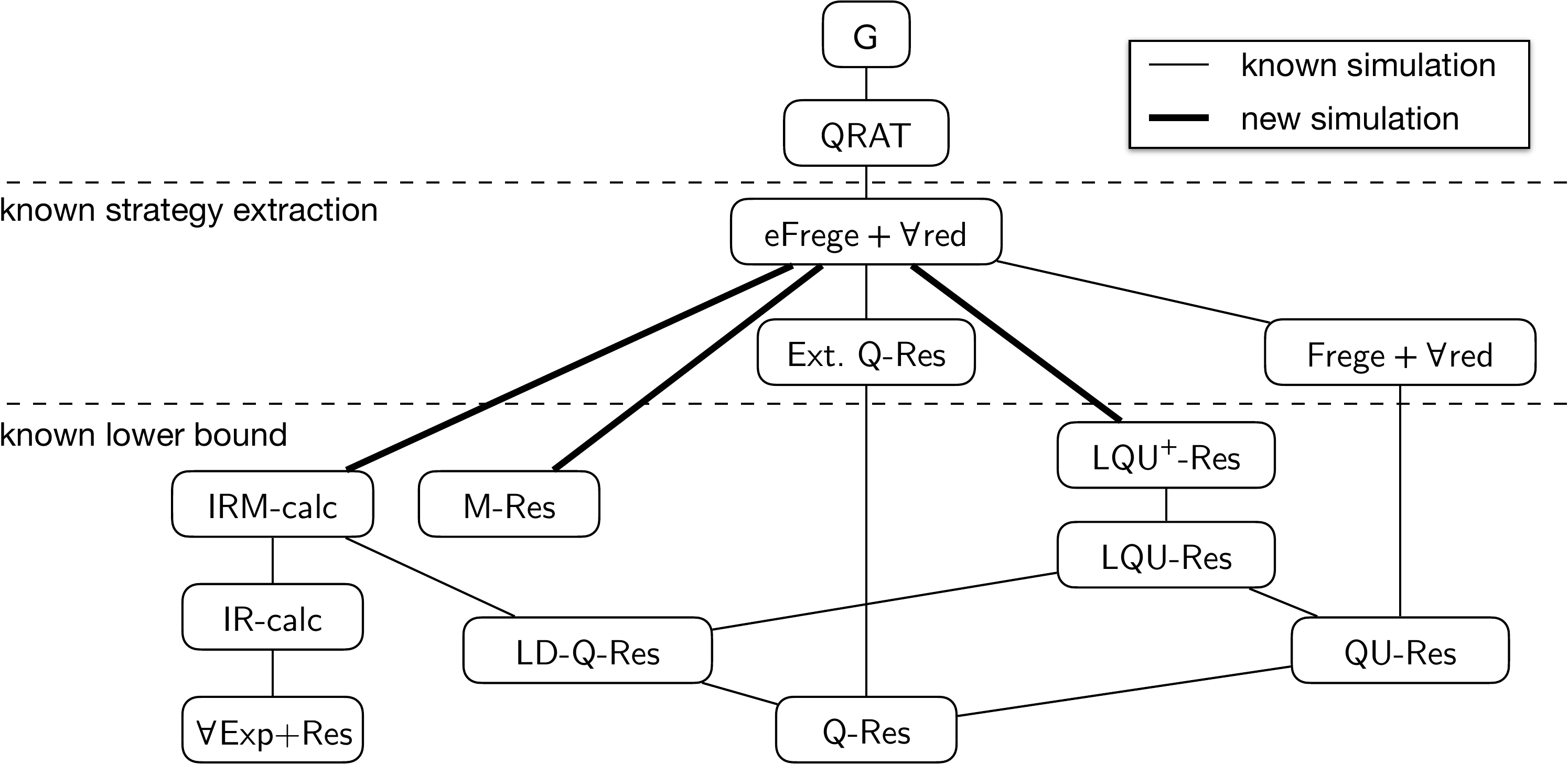}
	\caption{Hasse diagram for polynomial simulation order of QBF calculi \cite{BCJ19, BWJ14,BBCP20,HeuleSB17, ChewSat21,DBLP:journals/fmsd/BalabanovJ12, Gelder12,ChewHSAT22, BBM18}. In this diagram all proof systems below the first line are known to have strategy extraction, and all below the second line have an exponential lower bound. \Gfull and \qrat have strategy extraction if and only if $\Pe=\PSPACE$. \label{fig:simstructure}}
\end{figure}

Our proofs crucially rely on a property of QBF proof systems known as strategy extraction. Here, ``strategy'' refers to winning strategies of a set of \PSPACE\ two-player games (see Section~\ref{sec:prelim} for more details) each of which corresponds exactly to some QBF.
A proof system is said to have strategy extraction if a strategy for the two-player game associated with a QBF can be computed from a proof of the formula in polynomial time.
Balabanov and Jiang discovered~\cite{DBLP:journals/fmsd/BalabanovJ12} that Q-Resolution admitted a form of strategy extraction where a circuit computing a winning strategy could be extracted in linear time from the proofs.
Strategy extraction was subsequently proven for many QBF proof systems (cf. Figure~\ref{fig:simstructure}):
the expansion based systems
\ecalculus \cite{BCJ19},
\irc \cite{BCJ19} and
\irmc \cite{BCJ19},
Long-Distance Q-Resolution \cite{BEW13},
including with dependency schemes \cite{BEW13},
Merge Resolution~\cite{BBM18},
Relaxing Stratex \cite{Che16}
and
\C-\FregeRed systems 
including \eFregeRed \cite{BBCP20}.
Strategy extraction also gained notoriety because it became a method to show Q-resolution lower bounds~\cite{BCJ19}.  
Beyersdorff et al. \cite{BBCP20, BCMS16CPjournal} generalised this approach to more powerful proof systems, allowing them to establish a tight correspondence between lower bounds for \eFregeRed and two major open problems in circuit complexity and propositional proof complexity: they showed that proving a lower bound for \eFregeRed is equivalent to either proving a lower bound for $\mathsf{P/poly}$ or a lower bound for propositional \eFrege. It was conjectured by Chew~\cite{ChewSat21} that all the aforementioned proof systems that had strategy extraction were very likely to be simulated by \eFregeRed. An outline of how to use strategy extraction to obtain the corresponding simulations was also provided. 

We follow this outline in proving simulations for multiple systems by \eFregeRed. While the strategy extraction for expansion based systems \cite{BCJ19} has been known for a while using the technique from Goultiaeva et. al~\cite{Goultiaeva-ijcai11}, there currently is no intuitive way to formalise this strategy extraction into a simulation proof.
Here we specifically studied a new strategy extraction technique given by Schlaipfer et al. \cite{SSWZ20}, that creates local strategies for each \ecalculus line. Inductively, we can affirm each of these local strategies and prove the full strategy extraction this way. This local strategy extraction technique is based on arguments of Suda and Gleiss \cite{SG18}, which allow it to be generalised to the expansion based system \irmc. We thus manage to prove a simulation for \ecalculus and generalise it to \irc and then to \irmc. We also show a much more straight-forward simulation of \mrc and an adaptation of the \irmc argument to \lquprc.

The remainder of the paper is structured as follows. In Section~\ref{sec:prelim} we go over general preliminaries and the definition of \eFregeRed. The remaining sections are each dedicated to simulations of different calculi by \eFregeRed. In Section~\ref{sec:mrc} we begin with a simulation of \mrc as a relatively easy example. 

In Section~\ref{sec:expansion} we find show how \eFregeRed simulates expansion based systems.
We find a  propositional interpretation and a local strategy for \irc. 
This leads to a full simulation of \irc by \eFregeRed.
In Section~\ref{sec:irm} we extend this simulation to \irmc which involves dealing with merged literals. 
 In Section~\ref{sec:lquprc} we study the strongest CDCL proof system \lquprc and explain why it is also simulated by \eFregeRed, using a similar argument to \irmc. We leave some of the finer details of the simulation of \irmc and \lquprc in the Appendix.

\section{Preliminaries}\label{sec:prelim}
\subsection{Quantified Boolean Formulas}
A Quantified Boolean Formula (QBF) is a propositional formula augmented with Boolean quantifiers $\forall, \exists$ that range over the Boolean values $\bot, \top$ (the same as $0,1$). 
Every propositional formula is already a QBF. 
Let $\phi$ be a QBF. The semantics of the quantifiers are that:
$\forall x \phi(x)\equiv \phi(\bot)\wedge \phi(\top)$
and 
$\exists x \phi(x)\equiv \phi(\bot)\vee \phi(\top)$.

When investigating QBF in computer science we want to standardise the input formula.
In a \emph{prenex} QBF, all quantifiers appear outermost in a \emph{(quantifier) prefix}, and are followed by a propositional formula, called the \emph{matrix}.
If every propositional variable of the matrix is bound by some quantifier in the prefix we say the QBF is a 
\emph{closed} prenex QBF.
We often want to standardise the propositional matrix, and so we can take the same approach as seen often in propositional logic.
We denote the set of universal variables as $U$, and the set of existential variables as $E$.
A \emph{literal} is a propositional variable ($x$) or its negation ($\neg x$ or $\bar x$).
A \emph{clause} is a disjunction of literals. 
Since disjunction is idempotent, associative and commutative we can think of a clause simultaneously as a set of literals.
The empty clause is just false.
A \emph{conjunctive normal form (CNF)} is a conjunction of clauses. 
Again, since conjunction is idempotent, associative and commutative a CNF can be seen as set of clauses. 
The empty CNF is true, and a CNF containing an empty clause is false.
Every propositional formula has an equivalent formula in CNF, we therefore restrict our focus to closed \emph{PCNF} QBFs, that is closed prenex QBFs with CNF matrices.

\subsection{QBF Proof Systems}
\subsubsection{Proof Complexity}
A proof system \cite{CR79} is a polynomial-time checking function that checks that every proof maps to a valid theorem. Different proof systems have varying strengths, in one system a theorem may require very long proofs, in another the proofs could be considerably shorter.
We use \emph{proof complexity} to analyse the strength of proof systems~\cite{krajivcek2019proof}.
A proof system is said to have an $\Omega(f(n))$-lower bound, if there is a family of theorems such that shortest proof (in number of symbols) of the family are bounded below by $\Omega(f(n))$ where $n$ is the size (in number of symbols) of the theorem.
Proof system $p$ is said to \emph{simulate} proof system $q$ if there is a fixed polynomial $P(x)$ such that for every $q$-proof $\pi$ of every theorem $y$ there is a $p$-proof of $y$ no bigger than  $P(|\pi|)$ where $|\pi|$ denotes the size of $\pi$. 
A stricter condition, proof system $p$ is said to p-simulate proof system $q$ if there is a polynomial-time algorithm that takes $q$-proofs to $p$-proofs preserving the theorem. 

\subsubsection{Extended Frege+\Ared}

Frege systems are ``text-book'' style proof systems for propositional logic. They consist of a finite set of axioms and rules where any variable can be replaced by any formula (so each rule and axiom is actually a schema). 
A Frege system needs also to be sound and complete. 
Frege systems are incredibly powerful and can handle simple tautologies with ease. No lower bounds are known for Frege systems and all Frege systems are p-equivalent \cite{CR79, Rec76}. For these reasons we can assume all Frege-systems can handle simple tautologies and syllogisms without going into details. 

Extended \Frege (\eFrege) takes a \Frege system and allows the introduction of new variables that do not appear in any previous line of the proof. These variables abbreviate formulas, but since new variables can be consecutively nested, they can be understood to represent circuits. The rule works by introducing the axiom of $v\leftrightarrow f$ for new variable $v$ (not appearing in the formula $f$). Alternatively one can consider \eFrege as the system where lines are circuits instead of formulas.

Extended \Frege is a very powerful system, it was shown \cite{Kra95, Bey09} that any propositional proof system $f$ can be simulated by $\eFrege+||\phi||$ where $\phi$ is a polynomially recognisable axiom scheme.
The QBF analogue is \eFregeRed, which adds the reduction rule to all existing \eFrege rules~\cite{BBCP20}.
\eFregeRed is refutationally sound and complete for closed prenex QBFs. 
The reduction rules allows one to substitute a universal variable in a formula with $\bot$ or with $\top$ as long as no other variable appearing in that formula is right of it in the prefix.
Extension variables now must appear in the prefix and must be quantified right of the variables used to define it, we can consider them to be defined immediately right of these variables as there is no disadvantage to this.

\subsection{QBF Strategies}
With a closed prenex QBF $\Pi\phi $, the semantics of a QBF has an alternative definition in games.
The two-player QBF game has an $\exists$-player and a $\forall$-player. The game is played in order of the prefix $\Pi$ left-to-right, whoever's quantifier appears must assign the quantified variable to $\bot$ or $\top$. The existential player is trying to make the matrix $\phi$ become true. The universal player is trying to make the matrix become false. 
$\Pi\phi $ is true if and only if there winning strategy for the $\exists$ player.
$\Pi\phi $ is false if and only if there winning strategy for the $\forall$ player.

A \emph{strategy} for a false QBF is a set of functions $f_u$ for each universal variable $u$ on variables left of $u$ in the prefix. In a \emph{winning strategy} the propositional matrix must evaluate to false when every $u$ is replaced by $f_u$.
A QBF proof system has \emph{strategy extraction} if there is a polynomial time  program that takes in a refutation $\pi$ of some QBF $\Psi$ and outputs circuits that represent the functions of a winning strategy. 

A \emph{policy} is similarly defined as a strategy but with partial functions for each universal variables instead of a fully defined function.

\section{Extended Frege+\texorpdfstring{\Ared}{∀-Red} p-simulates M-Res}\label{sec:mrc}
In this section we show a first example of how the \eFregeRed simulation argument works in practice for systems that have strategy extraction. 
Merge resolution provides a straightforward example because the strategies themselves are very suitable to be managed in propositional logic. In later theorems where we simulate calculi like \irc and \irmc, representing strategies is much more of a challenge. 

\subsection{Merge Resolution}
Merge resolution (\mrc) was first defined by Beyersdorff, Blink\-horn and Mahajan \cite{BBM18}. Its lines combine clausal information with a merge map, for each universal variable. Merge maps give a ``local'' strategy which when followed forces the clause to be true or the original CNF to be false.

\subsubsection{Definition of Merge Resolution}
Each line of an \mrc proof consists of a clause on existential variables and partial universal strategy functions for universal variables. These functions are represented by \emph{merge maps}, which are defined as follows. 
For universal variable $u$, let $E_u$ be the set of existential variables left of $u$ in the prefix. For line number $i$, A non-trivial merge map $M^u_i$ is a collection of nodes in $[i]$ along with the construction function $M^u_i$, which details the structure. For $j\in [i]$, the construction function $M^u_i(j)$ is either in $\{\bot,\top\}$ for leaf nodes  or $E_u\times[j]\times[j]$ for internal nodes.
The root $r(u,i)$ is the highest value of all the nodes $M^u_i$.
The strategy function $h^u_{i,j}:\{0,1\}^{E_u}\rightarrow \{0,1\}$ maps assignments of existential variables $E_u$ in the dependency set of  $u$ to a value for $u$.
The function  $h^u_{i,t}$ for leaf nodes $t$ is simply the truth value $M^u_i(t)$.
For internal nodes $a$ with $M^u_i(a)= (x, b,c)$, we should interpret $h^u_{i,a}$ as ``If $x$ then $h^u_{i,b}$, else $h^u_{i,c}$'' or 
$h^u_{i,a}= (x \wedge h^u_{i,b}) \vee (\neg x \wedge h^u_{i,c})$.
In summary the merge map $M^u_i(j)$ is a representation of the strategy given by function $h^u_{i,r(u,i)}$.

The merge resolution proof system inevitably has merge maps for the same universal variable interact, and we have two kinds of relations on pairs of merge maps.

\begin{defi}
	Merge maps $M^u_j$ and $M^u_k$ are said to be \emph{consistent} if $M^u_j(i)=M^u_k(i)$ for each node $i$ appearing in both $M^u_j$ and $M^u_k$.
\end{defi}

\begin{defi}
	Merge maps $M^u_j$ and $M^u_k$ are said to be isomorphic if there exists a bijection $f$ from the nodes of $M^u_j$ to the nodes of $M^u_k$ such that if $M^u_j(a)=(x ,b, c)$ then $M^u_k(f(a))= (x, f(b), f(c))$ and if $M^u_j(t)=p\in \{\bot ,\top\}$ then  $M^u_k(f(t))=p$.
\end{defi}

With two merge maps $M^u_j$ and $M^u_k$, we define two operations as follows:
\begin{itemize}
\item $\Select(M^u_j,M^u_k)$ returns $M^u_j$ if $M^u_k$ is trivial (representing a ``don't care''), or $M^u_j$ and $M^u_k$ are isomorphic and returns $M^u_k$ if $M^u_j$ is trivial and not isomorphic to $M^u_k$.
If neither $M^u_j$ or $M^u_k$ is trivial and the two are not isomorphic then the operation fails.

\item $\Merge(x,M^u_j,M^u_k)$  returns the map $M^u_i$ with $i>j,i>k$ when $M^u_j,M^u_k$ are consistent where if $a$ is a node in $M^u_j$ then $M^u_i(a)=M^u_j(a)$ and 
if $a$ is a node in $M^u_k$ then $M^u_i(a)=M^u_k(a)$. Merge map $M^u_i$ has a new node $r(u,i)$ as a root node (which is greater than the maximum node in each of $M^u_i(a)$ or $M^u_j(a)$), and is defined as
$M^u_i(r(u,i))=(x,r(u,j),r(u,k))$.
\end{itemize}

Proofs in \mrc consist of lines, where every line is a pair $(C_i, \{ M^u_i\mid u \in U\})$.
Here, $C_i$ is a purely existential clause (it contains only literals that are from existentially quantified variables). The other part is a set containing merge maps for each universal variable (some of the merge maps can be trivial, meaning they do not represent any function).
Each line is derived by one of two rules:

\textbf{Axiom}:
$C_i=\{l \mid l\in C, \var(l)\in E\}$ is the existential subset of some clause $C$ where $C$ is a clause in the matrix.
If universal literals $u, \bar u$ do not appear in $C$, let $M^u_i$ be trivial.
If universal variable $u$ appears in $C$ then let $i$ be the sole node of $M^u_i$ with
$M^u_i(i)=\bot$.
Likewise if $\neg u$ appears in $C$ then let $i$ be the sole node of $M^u_i$ with
$M^u_i(i)=\top$.

\textbf{Resolution:}
Two lines $(C_j, \{ M^u_j\mid u \in U \})$ and $(C_k, \{ M^u_k\mid u \in U \})$
can be resolved to obtain a line $(C_i, \mid \{ M^u_i\mid u \in U\})$
if there is literal $\neg x\in C_j$
and $x \in C_k$
such that $C_i= C_j\cup C_k \setminus\{x, \neg x\}$, and every $M^u_i$ can either be defined as $\Select(M^u_j, M^u_k)$, when $M^u_j$ and $M^u_k$ are isomorphic or one is trivial,
or as $\Merge(x, M^u_j, M^u_k)$ when $x < u$ and $M^u_j$ and $M^u_k$ are consistent.

\subsection{Simulation of Merge Resolution}
We now state the main result of this section.
\begin{thm}\label{thm:mrc}
	\eFregeRed simulates \mrc.
\end{thm}
For a false QBF $\Pi \phi$ refuted by \mrc, the final set of merge maps represent a falsifying strategy for the universal player, the strategy can be asserted by a proposition $S$ that states that all universal variables are equivalent to their strategy circuits. It then should be the case that if $\phi$ is true, $S$ must be false, a fact that can be proved propositionally, formally $\phi \vdash \neg S$. 

To build up to this proof we can inductively find  a local strategy $S_i$ for each clause $C_i$ that appears in an \mrc line $(C_i, \{M^u_i\})$ such that $\phi \vdash S_i\rightarrow C_i$. Elegantly, $S_i$ is really just a circuit expressing that each $u\in U$ takes its value in $M^u_i$ (if non-trivial). Extension variables are used to represent these local strategy circuits and so the proof ends up as a propositional extended Frege proof.

The final part of the proof is the technique suggested by Chew \cite{ChewSat21} which was originally used by Beyersdorff et al. \cite{BBCP20}. That is, to use universal reduction starting from the negation of a universal strategy and arrive at the empty clause. 

\begin{proof}[Proof of Theorem~\ref{thm:mrc}]
	
	\textbf{Definition of extension variables.}
	We create new extension variables for each node in every non-trivial merge map appearing in a proof.
	$s^u_{i}$ is created for the node $i$ in merge map $M^u_i$.
	$s^u_{i}$ is defined as a constant when $i$ is leaf node in $M^u_i$.
	If $i$ is an internal node $s^u_{i}$ is defined as $s^u_{i}:= (x \wedge s^u_{b})\vee (\neg x \wedge s^u_{c})$, when $M^u_i(i)= (x, b, c)$.
	Because $x$ has to be before $u$ in the prefix,
	$s^u_{i}$ is always defined before universal variable $u$.

	\noindent\textbf{Induction Hypothesis:} It is easy for \eFrege to prove $\bigwedge_{u\in U_i} (u\leftrightarrow s^u_{r(u,i)})\rightarrow C_i$ from the axioms of $\phi$, where $r(u,i)$ is the index of the root node of Merge map $M^u_i$. $U_i$ is the subset of $U$ for which $M^u_i$ is non-trivial.
	
	\noindent\textbf{Base Case: Axiom:} Suppose $C_i$ is derived by axiom download of clause $C$. If $u$ has a strategy, it is because it appears in a clause and so $u\leftrightarrow s^u_{i}$, where $s^u_{i}\leftrightarrow c_u$ for $c_u\in{\top, \bot}$, $c_u$ is correctly chosen to oppose the literal in $C$ so that $C_i$ is just the simplified clause of $C$ replacing all universal $u$ with their $c_u$. This is easy for \eFrege to prove.

	\noindent\textbf{Inductive Step: Resolution:} If $C_j$ is resolved with $C_k $ to get $C_i$ with pivots $\neg x\in C_j$ and $x\in C_k$, we first show $\bigwedge_{u\in U_i} ( u \leftrightarrow s^u_{r(u,i)} )\rightarrow C_j$ and $\bigwedge_{u\in U_i} ( u \leftrightarrow s^u_{r(u,i)}  )\rightarrow C_k$, where $r(u,i)$ is the root index of the Merge map for $u$ on line $i$. We resolve these together.
	
	To argue that $\bigwedge_{u\in U_i} ( u \leftrightarrow s^u_{r(u,i)} )\rightarrow C_j$ we prove by induction that we can replace 
	$u \leftrightarrow s^u_{r(u,j)}$ with $u \leftrightarrow s^u_{r(u,i)}$ one by one.
	
	\textbf{Induction Hypothesis}: $U_i$ is partitioned into $W$ the set of adjusted variables and $V$ the set of variables yet to be adjusted.
	
	$(\bigwedge_{v\in V\cap U_j} ( v \leftrightarrow s^v_{r(v,j)} ))
	\wedge
	(\bigwedge_{v\in W} ( v \leftrightarrow s^v_{r(v,i)} ))
	\rightarrow C_j$
	
	\textbf{Base Case}:
	$(\bigwedge_{v\in U_i\cap U_j} ( v \leftrightarrow s^v_{j,r(v,j)} )\rightarrow C_j$ is the premise of the (outer) induction hypothesis, since $U_j \subseteq U_i$.
	
	\textbf{Inductive Step:}
	Starting with 
	$(\bigwedge_{v\in V\cap U_j} ( v \leftrightarrow s^v_{r(v,j)} ))
	\wedge
	(\bigwedge_{w\in W} (w \leftrightarrow s^w_{r(w,i)} ))
	\rightarrow C_j$.
	We pick a $u\in V$ to show 
	$
	( u \leftrightarrow s^u_{r(u,i)} )
	\wedge
	(\bigwedge^{v\neq u}_{v\in V\cap U_j} ( v \leftrightarrow s^v_{r(v,j)} ))
	\wedge
	(\bigwedge_{w\in W} ( w \leftrightarrow s^w_{r(w,i)} ))
	\rightarrow C_j$.
	We have four cases:
	\begin{enumerate}
		\item \Select chooses $M^u_i=M^u_j$.
		\item \Select chooses $M^u_i=M^u_k$ because $M^u_j$ is trivial.
		\item \Select chooses  $M^u_i=M^u_k$ because there is an isomorphism $f$ that maps $M^u_j$ to $M^u_k$.
		\item \Merge so that $M^u_i$ is the merge of $M^u_j$ and $M^u_k$ over pivot $x$.
	\end{enumerate}

	In (1)  $( u \leftrightarrow s^u_{r(u,j)} )$ is already 
	$( u \leftrightarrow s^u_{r(u,i)} )$ as $r(u,j)=r(u,i)$.
	
	In (2) we are simply weakening the implication.
	
	In (3) we prove inductively from the leaves to the root that $s^u_{f(t)}= s^u_{j,t}$. Eventually, we end up with $s^u_{f(r(u,k))}=s^u_{r(u,i)}$.
	Then $( u \leftrightarrow s^u_{r(u,j)} )$ can be replaced by 
	$( u \leftrightarrow s^u_{f(r(u,j))} )$. 
	As $f$ is an isomorphism $f(r(u,j))=r(u,k)$ and because \Select is used $r(u,k)=r(u,i)$. 
	Therefore we have  $( u \leftrightarrow s^u_{r(u,i)} )$.

	In (4) We need to replace $s^u_{r(u,j)}$ with $s^u_{r(u,i)}$. 
	For this we use the definition of merging that $x \rightarrow (s^u_{r(u,i)}\leftrightarrow s^u_{r(u,j)})$ and so we have $ (s^u_{r(u,i)}\leftrightarrow s^u_{r(u,j)})\vee \neg x$ but the $\neg x$ is absorbed by the $C_j$ in right hand side of the implication.
	
	\textbf{Finalise Inner Induction:}
	At the end of this inner induction, we have $\bigwedge_{u\in U_i} ( u \leftrightarrow s^u_{r(u,i)} )\rightarrow C_j$ and symmetrically $\bigwedge_{u\in U_i} ( u \leftrightarrow s^u_{r(u,i)} )\rightarrow C_k$.
	We can then prove $\bigwedge_{u\in U_i} ( u \leftrightarrow s^u_{r(u,i)} )\rightarrow C_i$.
	
	\noindent\textbf{Finalise Outer Induction:}
	Note that we have done three nested inductions on the nodes in a merge maps, on the universal variables, and then on the lines of an \mrc proof.
	Nonetheless, this gives a quadratic size \eFrege proof in the number of nodes appearing in the proof.
	In \mrc the final line will be the empty clause and its merge maps. 
	The induction gives us $\bigwedge_{u\in U_l} ( u \leftrightarrow s^u_{r(u,l)} )\rightarrow \bot$.
	In other words, if $U_l=\{u_1, \dots u_n\}$, where $u_i$ appears before $u_{i+1}$ in the prefix, $\bigvee_{i=1}^{n} ( u_i \oplus s^{u_i}_{r(u_i,l)} )$.

	By reduction of $\bigvee_{i=1}^{n-k+1} ( u_i \oplus s^{u_i}_{r(u_i,l)} )$, we derive $(0 \oplus s^{u_{n-k+1}}_{r(u_{n-k+1},l)})\vee\bigvee_{i=1}^{n-k} (u_i \oplus s^{u_i}_{r(u_i,l)})$ and $(1 \oplus s^{u_{n-k+1}}_{r(u_{n-k+1},l)})\vee\bigvee_{i=1}^{n-k} (u_i \oplus s^{u_i}_{r(u_i,l)})$, which we can resolve to obtain $\bigvee_{i=1}^{n-k} ( u_i \oplus s^{u_i}_{r(u_i,l)} )$.
We continue this until we reach the empty disjunction.
\end{proof}	

\section{Extended Frege+\texorpdfstring{\Ared}{∀-Red} p-simulates \irc}\label{sec:expansion}

\subsection{Expansion-Based Resolution Systems}
The idea of an expansion based QBF proof system is to utilise the semantic identity:
$\forall u \phi(u)= \phi(0)\wedge\phi(1)$, to replace universal quantifiers and their variables with propositional formulas.
With
$\forall u \exists x\phi(u)= \exists x\phi(0)\wedge\exists x\phi(1)$
the $x$ from $\exists x\phi(0)$ and from $\exists x\phi(1)$ are actually different variables. 
The way to deal with this while maintaining prenex normal form is to introduce annotations that distinguish one $x$ from another. We will also introduce a third annotation $*$ which will be used only for the purpose of short proofs.

\begin{defiC}[\cite{BCJ19}] \label{def:annotations}\hfill
	\begin{enumerate}
		\item
		An \emph{extended assignment} is a partial mapping from the universal variables to $\{0,1,*\}$.
		We denote an extended assignment by a set or list of individual replacements i.e. $0/u,*/v$ is an extended assignment. We often use set notation where appropriate. 
		\item
		An \emph{annotated clause} is a clause where each literal is annotated by an extended assignment to universal variables.
		
		\item
		For an extended assignment $\sigma$ to universal variables we write $l^{\restr{l}{\sigma}}$ to denote an annotated literal where
		$\restr{l}{\sigma}=\comprehension{c/u\in\sigma}{\lev(u)<\lev(l)}$.
		
		\item
		Two (extended) assignments~$\tau$ and $\mu$ are called \emph{contradictory} if there exists a variable
		$x\in\domain(\tau)\cap\domain(\mu)$ with $\tau(x)\neq\mu(x)$.
	\end{enumerate}
\end{defiC}

\subsubsection{Definitions}
The most simple way to use expansion would be to expand all universal quantifiers and list every annotated clause. The first expansion based system we consider, \ecalculus (Figure~\ref{fig:QEXP}), has a mechanism to avoid this potential exponential explosion in some (but not all) cases. An annotated clause is created and then checked to see if it could be obtained from expansion. This way a refutation can just use an unsatisfiable core rather than all clauses from a fully expanded matrix. 

\begin{figure}[h]
			\begin{prooftree}
				\AxiomC{}
				\RightLabel{(Axiom)}
				\UnaryInfC{$
					\comprehension{l^{\restr{l}{\tau}}}{l\in C, l\text{ is existential}}
					\cup
					\comprehension{\tau(l)}{l\in C, l\text{ is universal}}
					$}
			\end{prooftree}
			\begin{minipage}{0.99\linewidth}
				$C$ is a clause from the matrix and
				$\tau$ is a $\{0,1\}$ assignment to all universal variables.\\
			\end{minipage}

			\begin{prooftree}
				\AxiomC{$C_1\cup\{x^{\tau}\}$}
				\AxiomC{$C_2\cup\{\lnot x^{\tau}\}$}
				\RightLabel{(Res)}
				\BinaryInfC{$C_1\cup C_2$}
			\end{prooftree}
			\caption{The rules of \ecalculus (adapted from~\cite{JM15}).}
			\label{fig:QEXP}
\end{figure}

The drawback of \ecalculus is that one might end up repeating almost the same derivations over and over again if they vary only in changes in the annotation which make little difference in that part of the proof. 
This was used to find a lower bound to \ecalculus for a family of formulas easy in system \qrc \cite{JM15}.
To rectify this, \irc improved on \ecalculus to allow a delay to the annotations in certain circumstances. Annotated clauses now have annotations with ``gaps'' where the value of the universal variable is yet to be set. When they are set there is the possibility of choosing both assignments without the need to rederive the annotated clauses with different annotations.

\begin{defiC}[\cite{BCJ19}]
	Given two partial assignments (or partial annotations) $\alpha$ and $\beta$. The completion $\complete{\alpha}{\beta}$, is a new partial assignment, where
	
	\begin{center}
	\[\complete{\alpha}{\beta}(u)= \begin{cases}\alpha(u) & \text{if } u\in \domain(\alpha) \\ \beta(u)  &\text{if } u\in \domain(\beta)\setminus\domain(\alpha)\\
	\text{unassigned}  &\text{otherwise} \end{cases}\]
	\end{center}
	
\end{defiC}

For $\alpha$ an assignment of the universal variables and $C$ an annotated clause we define $\instantiate(\alpha, C):= \bigvee_{l^{\tau}\in C} l^{\restr{l}{\complete{\tau}{\alpha}}}$. Annotation $\alpha$ here gives  values to unset annotations where one is not already defined. Because the same $\alpha$ is used throughout the clause, the previously unset values gain consistent annotations, but mixed annotations can occur due to already existing annotations. 

\begin{figure}[h]
			\begin{prooftree}
				\AxiomC{}
				\RightLabel{(Axiom)}
				\UnaryInfC{$\comprehension{l^{\restr{l}{\tau}}}{l\in C, l\text{ is existential}} $}
			\end{prooftree}
			$C$ is a non-tautological clause from the matrix. $\tau=\comprehension{0/u}{u\text{ is universal in }C}$, where the notation $0/u$ for literals $u$ is shorthand for $0/x$ if $u=x$ and $1/x$ if $u=\neg x$.
			\begin{prooftree}
				\AxiomC{$x^\tau\lor C_1 $ }
				\AxiomC{$\lnot x^\tau\lor C_2 $}
				\RightLabel{(Resolution)}
				\BinaryInfC{$C_1\cup C_2$}
				\DisplayProof\hspace{1cm}
				\AxiomC{$C$}
				\RightLabel{(Instantiation)}
				\UnaryInfC{$\instantiate(\tau,C)$}
			\end{prooftree}
			\text{$\tau$ is an assignment to universal variables with $\range(\tau) \subseteq \{0,1\}$.}
			\caption{ The rules of \irc \cite{BCJ19}.}
			\label{fig:IR}

\end{figure}

The definition of \irc is given in Figure~\ref{fig:IR}.
Resolved variables have to match exactly, including that missing values are missing in both pivots.
However, non-contradictory but different annotations may still be used for a later resolution step after the instantiation rule is used to make the annotations match the annotations of the pivot.

\subsubsection{Local Strategies for $\forall$Exp+Res}

The work from Schlaipfer et al.~\cite{SSWZ20} creates a conversion of each annotated clause $C$ appearing in some \ecalculus proof into a propositional formula $\con(C)$ defined in the original variables of $\phi$ (so without creating new annotated variables). 
$C$ appearing in a proof asserts that there is some (not necessarily winning) strategy for the universal player to force $\con(C)$ to be true under $\phi$. 
The idea is that for each line $C$ in an \ecalculus refutation of $\Pi \phi$ there is some local strategy $S$ such that $S\wedge \phi \rightarrow \con(C)$.
If $C$ is empty, then $S$ is a winning strategy for the universal player. Otherwise, $S$ only wins if the existential player cooperates by playing according to one of the annotated literals $l^\tau \in C$, that is, if the existential player promises to falsify the literal $l$ whenever the assignment chosen by the universal player is consistent with the annotation~$\tau$.
Suda and Gleiss showed that the resolution rule can then be understood as combining strategies so that the ``promises'' of the existential player corresponding to the pivot literals $x^\tau$ and $\neg x^\tau$ cancel out~\cite{SG18}.

The extra work by Schlaipfer et al. is that the strategy circuits (for each $u$) can be constructed in polynomial time, and can be defined in variables left of $u_i$ in the prefix.
Let $u_1 \dots u_n$ be all universal variables in order. 
For each line in an \ecalculus proof we have a strategy which we will here call $S$.
For each $u_i$ there is an extension variable $\val_S^i$, before $u_i$, that represents the value assigned to $u_i$ by $S$ (under an assignment of existential variables). Using these variables, we obtain a propositional formula representing the strategy as $S=\bigwedge_{i=1}^{n} u_i \leftrightarrow \val_S^i$.
Additionally, we define a conversion of annotated logic in \ecalculus to propositional logic as follows.
For annotations $\tau$ let $\ann(\tau)= \bigwedge_{1/u_i\in \tau} u_i \wedge \bigwedge_{0/u_i\in \tau} \bar u_i$.
We convert annotated literals as $\con(l^\tau) = l\wedge \ann(\tau)$ and clauses as $\con(C)= \bigvee_{l\in C} \con(l)$.

\subsection{Policies and Simulating \irc}\label{sec:ir}
The conversion needs to be revised for \irc. 
In particular the variables not set in the annotations need to be understood.
The solution is to basically treat unset as a third value, and work with local strategies that do not set all universal variables.
Following Suda and Gleiss, we refer to such (partial) strategies as \emph{policies}~\cite{SG18}.

In practice, this requires new $\set^i_S$ variables (left of $u_i$) which state that the $i$th universal variable is set by policy $S$.
We include these variables in our encoding of policy $S$ and let $S=\bigwedge_{i=1}^{n} \set_S^i \rightarrow(u_i \leftrightarrow \val_S^i)$.
The conversion of annotations, literals and clauses also has to be changed.  
For annotations $\tau$ of some quantified variable $x$ let 
\[\textstyle
\ann_{x,S}(\tau)= \bigwedge_{1/u_i\in \tau} (\set^i_{S} \wedge u_i) \wedge \bigwedge_{0/u_i\in \tau} (\set^i_{S}\wedge\bar u_i ) \wedge \bigwedge^{u_i\notin \domain(\tau)}_{u_i<_\Pi x} \neg \set^i_S.\]

Let $\con_S(l^\tau) = l\wedge \ann_{x,S}(\tau)$ and $\con_S(C)= \bigvee_{l\in C} \con_S(l)$ similarly to before, we just reference a particular policy~$S$. This means that we again want $S\wedge \phi \rightarrow \con_S(C)$ for each line, note that $\set^i_{S}$ variables are defined in their own way.

The most crucial part of simulating \irc is that after each application of the resolution rule we can obtain a working policy.
\begin{lem}\label{thm:res}
	Suppose, there are policies $L$ and $R$ such that $L\rightarrow \con_{L}(C_1 \vee \neg x^\tau)$ and $R\rightarrow \con_{R}(C_2 \vee x^\tau)$ then there is a policy $B$ such that $B \rightarrow \con_{B}(C_1 \vee C_2)$ can be obtained in a short \eFrege proof.
\end{lem}

The proof of the simulation of \irc relies on Lemma~\ref{thm:res}. To prove this we have to first give the precise definitions of the policy $B$ based on policies $L$ and $R$. Schlaipfer et al.'s work~\cite{SSWZ20} is used to crucially make sure the strategy $B$, respects the prefix ordering.

\subsubsection{Building the Strategy}

We start to define $\val^i_B$ and $\set^i_B$ on lower $i$ values first. In particular we will always start with $1\leq i \leq m$ where $u_m$ is the rightmost universal variable still before the pivot variable $x$ in the prefix.
Starting from $i=0$, the initial segments of $\ann_{x,L}(\tau)$ and $\ann_{x,R}(\tau)$ may eventually reach such a point $j$ where one is contradicted. Before this point $L$ and $R$ are detailing the same strategy (they may differ on $\val^i$ but only when $\set^i$ is false) so this part of $B$ can be effectively played as both $L$ and $R$ simultaneously.
Without loss of generality, as soon as $L$ contradicts $\ann_{x,L}(\tau)$, we know that  $\con_L( x^\tau)$ is not satisfied by $L$ and thus it makes sense for $B$ to copy $L$, at this point and the rest of the strategy as it will satisfy $\con_B(C_1)$.
It is entirely possible that we reach $i=m$ and not contradict either $\ann_{x,L}(\tau)$ or $\ann_{x,R}(\tau)$.
Fortunately after this point in the game we now know the value the existential player has chosen for $x$.
We can use the $x$ value to decide whether to play $B$ as $L$ (if $x$ is true) or $R$ (if $x$ is false).

To build the circuitry for $\val^i_B$ and $\set^i_B$ we will introduce other circuits that will act as intermediate.
First we will use constants $\set^i_\tau$ and $\val^i_\tau$ that make $\ann_{x,S}(\tau)$ equivalent to $ \bigwedge_{u_i<_\Pi x} (\set^i_{S}\leftrightarrow \set^i_{\tau} )\wedge \set^i_{\tau}\rightarrow (u_i\leftrightarrow \val^i_{\tau} )$. This mainly makes our notation easier.
Next we will define circuits that represent two strategies being equivalent up to the $i$th universal variable. This is a generalisation of what was seen in the local strategy extraction for \ecalculus \cite{SSWZ20}.
\begin{center}$\equ^{0}_{f=g}:=1, \equ^{i}_{f=g}:=\equ^{i-1}_{f=g}\wedge (\set^i_f\leftrightarrow \set^i_g)\wedge (\set^i_f\rightarrow(\val^i_f \leftrightarrow \val^i_g) ).$
\end{center}
We specifically use this for a trigger variable that tells you which one of $L$ and $R$ differed from $\tau$ first.

\begin{center}$\diff_L^0:= 0\text{ and }\diff_L^i:= \diff_L^{i-1} \vee (\equ_{R=\tau}^{i-1}\wedge ((\set^i_L\oplus \set^i_\tau)\vee (\set^i_\tau\wedge(\val^i_L\oplus \val^i_\tau))))$ \\
  $\diff_R^0:= 0\text{ and }\diff_R^i:= \diff_R^{i-1} \vee (\equ_{L=\tau}^{i-1}\wedge ((\set^i_R\oplus \set^i_\tau)\vee (\set^i_\tau\wedge(\val^i_R\oplus \val^i_\tau))))$
\end{center}

Note that $\diff_L^i$ and $\diff_R^i$ can both be true but only if the strategies start to differ from $\tau$ at the same point.

Using these auxiliary variables, we can define a bottom policy $B$ that chooses between the left policy $L$ and the right policy $R$ as indicated above, following Suda and Gleiss's \texttt{Combine} operation~\cite{SG18}.
If one of the policies is inconsistent with the annotation~$\tau$ (this includes setting a variable that is not set by~$\tau$), policy $B$ follows whichever policy is inconsistent first, picking $L$ if both policies start deviating at the same time.
If both policies are consistent with $\tau$, policy $B$ follows $R$ if the pivot $x$ is false, otherwise it follows $L$.
\begin{defi}[Definition of resolvent policy for \textsf{IR-calc}]\label{def:B}

	For $0\leq i\leq m$,
	define $\val^i_B$ and $\set_B^i$ such $\val^i_B=\val_R^i$ and $\set_B^i=\set_R^i$ if 
	\begin{center}$
	\neg \diff_L^{i-1} \wedge (\diff_R^{i-1} \vee (\neg\set_{\tau}^i \wedge \neg\set_L^i \wedge \set_R^i) \vee (\set_{\tau}^i \wedge \set_L^i \wedge (\val_\tau^i \leftrightarrow \val_L^i)))
	$
	\end{center}
	and $\val^i_B=\val_L^i$ and $\set_B^i=\set_L^i$,
	otherwise.
	
	For $i>m$, define $\val^i_B$ and $\set_B^i$ such $\val^i_B=\val_R^i$ and $\set_B^i=\set_R^i$ if 
	\begin{center}$\neg \diff_L^m \wedge (\diff_R^m \vee \bar x)$ \end{center}
	and $\val^i_B=\val_L^i$ and $\set_B^i=\set_L^i$,
	otherwise.
	
\end{defi}

We will now define variables $B_L$ and $B_R$. These say that $B$ is choosing $L$ or $R$, respectively. These variables can appear rightmost in the prefix, as they will be removed before reduction takes place. The purpose of $B_L$ (resp. $B_R$) is that $\con_B$ becomes the same as $\con_L$ (resp. $\con_R$).
\begin{itemize}
	\item $B_L:= \bigwedge_{i=1}^n (\set^i_B\leftrightarrow \set^i_L)\wedge(\set^i_B\rightarrow (\val^i_B\leftrightarrow \val^i_L))$
	
	\item $B_R:= \bigwedge_{i=1}^n (\set^i_B\leftrightarrow \set^i_R)\wedge(\set^i_B\rightarrow (\val^i_B\leftrightarrow \val^i_R))$
\end{itemize}

The important points are that $B$ is set up so that it either takes values in $L$ or $R$, i.e. $B\rightarrow B_L\vee B_R$,
specifically we need that whenever the propositional formula $\ann_{x,B}(\tau)$ is satisfied, 
$B=B_L$ when $x$, and $B=B_R$ when $\neg x$. The variables $\set_B^i$ and $\val_B^i$ that comprise the policy are carefully constructed to come before $u_i$. A number of technical lemmas involving all these definitions is necessary for the simulation.

\begin{lem}\label{lem:chain}
	For $0<j\leq m$ the following propositions have short derivations in Extended Frege:
	\begin{itemize}
		\item $\diff_L^j \rightarrow \bigvee_{i=1}^j \diff_L^i \wedge \neg \diff_L^{i-1}$
		\item $\diff_R^j \rightarrow \bigvee_{i=1}^j \diff_R^i \wedge \neg \diff_R^{i-1}$
		\item $\neg \equ_{f=g}^j \rightarrow \bigvee_{i=1}^j \neg \equ_{f=g}^i \wedge \equ_{f=g}^{i-1}$. For $f,g \in \{L, R, \tau\}$.
	\end{itemize}
\end{lem}

\begin{proof}
	\textbf{Induction Hypothesis on $j$:}
	$\diff_L^j \rightarrow \bigvee_{i=1}^j \diff_L^i \wedge \neg \diff_L^{i-1}$ has an $O(j)$-size proof 
	
	\noindent\textbf{Base Case $j=1$}: 
	$\diff_L^1 \rightarrow \diff_L^1$ 
	is a basic tautology with a constant-size Frege proof,
	$\diff_L^0$ 
	is false by definition so Frege can assemble 
	$\diff_L^1 \rightarrow \diff_L^1\wedge \neg \diff_L^0$.
	
	\begin{sloppypar}
		\noindent\textbf{Inductive Step $j+1$}:
		$\neg \diff_L^j \vee  \diff_L^j$ 
		and 
		$\diff_L^{j+1}\rightarrow \diff_L^{j+1} $ 
		are tautologies with a constant-size Frege proof.
		Putting them together we get 
		$\diff_L^{j+1}\rightarrow\diff_L^{j+1}\wedge (\neg \diff_L^j \vee  \diff_L^j)$
		and weaken to 
		$\diff_L^{j+1}\rightarrow(\diff_L^{j+1}\wedge \neg \diff_L^j) \vee  \diff_L^j$.
		Using the induction hypothesis,
		$\diff_L^j \rightarrow \bigvee_{i=1}^j \diff_L^i \wedge \neg \diff_L^{i-1}$,
		we can change this tautology to 
	\end{sloppypar}
	\begin{center}$\diff_L^{j+1}\rightarrow(\diff_L^{j+1}\wedge \neg \diff_L^j)  \vee \bigvee_{i=1}^j \diff_L^i \wedge \neg \diff_L^{i-1}$
	\end{center}
	
	Note that since $\neg \diff^0_R, \equ^0_{L=\tau}, \equ^0_{R=\tau}  $ are all true. The proofs for $\diff^j_R$, $\neg \equ^j_{L=\tau}$ and $\neg \equ^j_{R=\tau}$ are identical modulo the variable names.
\end{proof}

\begin{lem}\label{lem:impl}
	For $0\leq i \leq j\leq m$ the following propositions that describe the monotonicity of $\diff$ have short derivations in Extended Frege:
	\begin{itemize}
		\item $\diff_L^i \rightarrow \diff_L^j$
		\item $\diff_R^i \rightarrow \diff_R^j$
		\item $\neg \equ_{f=g}^i \rightarrow \neg \equ_{f=g}^j$
	\end{itemize}
\end{lem}

\begin{proof}
	For $\diff_L$ and $\diff_R$,
	
	\noindent\textbf{Induction Hypothesis on $j$:}
	$\diff_L^i \rightarrow \diff_L^j$ has an $O(j)$ proof.
	
	\noindent\textbf{Base Case $j=i$}: 
	$\diff_L^i \rightarrow \diff_L^i$ 
	is a tautology with a constant-size Frege proof.
	
	\noindent\textbf{Inductive Step $j+1$}: 
	$\diff^{j+1}_L:= \diff^j_L \vee A$ 
	where  $A$ is an expression.
	Therefore in all cases $\diff^j_L\rightarrow \diff^{j+1}_L$ is a straightforward corollary with a constant-size number of additional Frege steps.
	Using the induction hypothesis 
	$\diff^i_L\rightarrow \diff^j_L$ 
	we can get 
	$\diff^i_L\rightarrow \diff^{j+1}_L$.
	The proof is symmetric for $R$. 
	
	For $\neg \equ_{f=g}$,
	
	\noindent\textbf{Induction Hypothesis on $j$:}
	$\neg \equ_{f=g}^i \rightarrow \neg \equ_{f=g}^j$ has an $O(j)$ proof.
	
	\noindent\textbf{Base Case $j=i$}: 
	$\neg \equ_{f=g}^i \rightarrow \neg \equ_{f=g}^i$ 
	is a tautology that Frege can handle.
	
	\noindent\textbf{Inductive Step $j+1$}: 
	$\equ_{f=g}^{j+1}:= \equ_{f=g}^{j} \wedge A$ 
	where  $A$ is an expression.
	Therefore in all cases $\neg \equ_{f=g}^{j}\rightarrow \neg \equ_{f=g}^{j+1}$ is a straightforward corollary with a constant-size number of additional Frege steps.
	Using the induction hypothesis 
	$\neg \equ_{f=g}^i \rightarrow \neg \equ_{f=g}^j$
	we can get 
	$\neg \equ_{f=g}^i \rightarrow \neg \equ_{f=g}^{j+1}$.
\end{proof}

\begin{lem}\label{lem:rel}
	For $0\leq i \leq j\leq m$ the following propositions describe the relationships between the different extension variables and have short derivations in Extended Frege:
	
	\begin{itemize}
		\item $\equ^{i}_{L=\tau}\rightarrow \neg \diff^i_L$
		\item $\diff_L^i \wedge \neg \diff_L^{i-1}\rightarrow \equ^{i-1}_{R=\tau}$
		\item $\diff_L^i \wedge \neg \diff_L^{i-1}\rightarrow \neg \diff_R^{i-1}$
		\item $\equ^{i}_{R=\tau}\rightarrow \neg \diff^i_R$
		\item $\diff_R^i \wedge \neg \diff_R^{i-1}\rightarrow \equ^{i-1}_{L=\tau}$
		\item $\diff_R^i \wedge \neg \diff_R^{i-1}\rightarrow \neg \diff_L^{i-1}$
	\end{itemize}
\end{lem}

\begin{proof}
	
	\noindent\textbf{Induction Hypothesis on $i$: }
	$\equ^{i}_{L=\tau}\rightarrow \neg \diff^i_L$ in an $O(i)$-size Frege proof.
	
	\noindent\textbf{Base Case $i=0$: } $\diff^i_L$ is defined as $0$ so $\neg \diff^i_L$ is true and trivially implied by 
	$\equ^{i}_{L=\tau}$. This can be shown in a constant-size Frege proof.
	
	\begin{sloppypar}
		\noindent\textbf{Inductive Step $i+1$:}
		If $\set^{i+1}_\tau$ is false then $\equ^{i+1}_{L=\tau}$ is equivalent to  $\equ^{i}_{L=\tau}\wedge \neg \set^{i+1}_L$ and $\neg \diff^{i+1}_L$ is equivalent to $\neg \diff^{i}_L\wedge \neg \set^{i+1}_L\vee \neg \equ^{i}_{L=\tau}$. Induction hypothesis is 
		$\equ^{i}_{L=\tau}\rightarrow \neg \diff^{i}_L$, 
		now $\equ^{i+1}_{L=\tau}$ implies $\neg \diff^{i}_L$ and $\neg \set^{i+1}$ which is enough for $\neg \diff^{i+1}_L$.
		If $\set^{i+1}_\tau$ is true then $\equ^{i+1}_{L=\tau}$ is equivalent to  $\equ^{i}_{L=\tau}\wedge \set^{i+1}_L\wedge (\val^{i+1}_L \leftrightarrow \val^{i+1}_\tau)$ and $\neg \diff^{i+1}_L$ is equivalent to $\neg \diff^{i}_L\wedge \set^{i+1}_L\wedge (\val^{i+1}_L \leftrightarrow \val^{i+1}_\tau)\vee \neg \equ^{i}_{L=\tau}$.
		Again, using the induction hypothesis, $\equ^{i+1}_{L=\tau}$ now implies ,$\diff^{i}_L$ $\set^{i+1}_L$ and $(\val^{i+1}_L \leftrightarrow \val^{i+1}_\tau)$ which is enough for $\diff^{i+1}_L$.
		
		Therefore using the induction hypothesis $\equ^{i+1}_{L=\tau}\rightarrow \neg \diff^{i+1}_L$. This can be shown in a constant number of Frege steps.
		Similarly for $R$.
	\end{sloppypar}
	
	The formulas $\diff^i_L\wedge \neg \diff^{i-1}_L \rightarrow \equ^{i-1}_{R=\tau}$ are simple corollaries of the inductive definition of $\diff^i_L$, and combined with $\equ^{i-1}_{R=\tau}\rightarrow \neg  \diff^{i-1}_R$ we get $\diff^i_L\wedge \neg \diff^{i-1}_L \rightarrow \neg  \diff^{i-1}_R$. Similarly if we swap $L$ and $R$.
\end{proof}

\begin{lem}\label{lem:tau}
	For any $0\leq i\leq m$
	the following propositions are true and have short Extended Frege proofs.
	
	\begin{itemize}
		\item $L\wedge \diff_L^i \rightarrow \neg \ann_{x,L}(\tau)$
		\item $R\wedge \diff_R^i \rightarrow \neg \ann_{x,R}(\tau)$
	\end{itemize}
\end{lem}

\begin{proof}
	We primarily use the disjunction in Lemma~\ref{lem:chain}
	$\diff_L^j \rightarrow \bigvee_{i=1}^j \diff_L^i \wedge \neg \diff_L^{i-1}$.
	
	Each individual disjunct $\diff_L^i \wedge \neg \diff_L^{i-1}$ is saying the difference triggers at that point. We can represent that in a proposition that can be proven in Extended Frege:
	$\diff_L^i \wedge \neg \diff_L^{i-1} \rightarrow ((\set^i_L\oplus \set^i_\tau)\vee (\set^i_\tau\wedge(\val^i_L\oplus \val^i_\tau)))$. We want to show that this also triggers the negation of $\ann_{x,L}(\tau)$.
	If $L$ differs from $\tau$ on a $\set^i_L$ value we contradict $\ann_{x,L}(\tau)$ in one of two ways:
	$L\wedge (\set^i_L\oplus \set^i_\tau)\wedge \set^i_L \rightarrow \neg \set^i_\tau$ or 
	$L\wedge (\set^i_L\oplus \set^i_\tau)\wedge \neg \set^i_L \rightarrow \set^i_\tau$.
	
	If $L$ differs from $\tau$ on a $\val^i_L$ value when $\set_L^i=\set_\tau^i=1$   we contradict $\ann_{x,L}(\tau)$ in one of two ways:
	\begin{itemize}
		\item $L\wedge \set^i_L \wedge \set^i_\tau \wedge  (\set^i_\tau\rightarrow(\val^i_L\oplus \val^i_\tau))\wedge \val^i_L \rightarrow \neg \val^i_\tau \wedge u_i$
		\item $L\wedge \set^i_L \wedge \set^i_\tau \wedge (\set^i_\tau\rightarrow(\val^i_L\oplus \val^i_\tau))\wedge \neg \val^i_L \rightarrow  \val^i_\tau \wedge \neg u_i$.
	\end{itemize}
Each disjunct is a constant size Frege derivation 
	When put together with the big disjunction this lends itself to a linear-size (in $m$) Frege derivation which is also symmetric for $R$.
\end{proof}

\begin{lem}\label{lem:nLnR}
	For any $1\leq j\leq m$
	the following propositions are true and have a short Extended Frege proof.
	\begin{itemize}
		\item $\neg \diff_L^j \wedge \neg \diff_R^j \rightarrow \equ^j_{L=\tau}$ 
		\item $\neg \diff_L^j \wedge \neg \diff_R^j \rightarrow \equ^j_{R=\tau}$
		\item $\neg \diff_L^j \wedge \neg \diff_R^j \rightarrow (\set^j_B \leftrightarrow \set^j_L)$
		\item $\neg \diff_L^j \wedge \neg \diff_R^j \rightarrow \set^j_B\rightarrow(\val^j_B \leftrightarrow \val^j_L)$
		\item $\neg \diff_L^j \wedge \neg \diff_R^j \rightarrow (\set^j_B \leftrightarrow \set^j_R)$
		\item $\neg \diff_L^j \wedge \neg \diff_R^j \rightarrow \set^j_B\rightarrow(\val^j_B \leftrightarrow \val^j_R)$
		
	\end{itemize} 
	
\end{lem}

\begin{proof}
	\begin{sloppypar}
		We first show $\neg \equ_{L=\tau}^j \rightarrow \neg \equ_{R=\tau}^{j-1}\vee \diff_L^j\vee \diff_R^j$ and 
		$\neg \equ_{R=\tau}^j \rightarrow \neg \equ_{L=\tau}^{j-1}\vee \diff_L^j\vee \diff_R^j$.
		$\neg \equ_{R=\tau}^{j-1}$ and $\neg \equ_{L=\tau}^{j-1}$ are the problems here respectively, but they can be removed via induction to eventually get
		$\neg \diff_L^j \wedge \neg \diff_R^j \rightarrow \equ^j_{L=\tau}$  and $\neg \diff_L^j \wedge \neg \diff_R^j \rightarrow \equ^j_{R=\tau}$. 
		The remaining implications are corollaries of these and rely on the definition of $\equ$, $\set_B$ and $\val_B$.
	\end{sloppypar}
	
	\noindent\textbf{Induction Hypothesis on j:} $\neg \diff_L^j \wedge \neg \diff_R^j \rightarrow \equ^j_{L=\tau}$ and $\neg \diff_L^j \wedge \neg \diff_R^j \rightarrow \equ^j_{R=\tau}$.
	
	\noindent\textbf{Base Case $j=0$:} $ \equ_{L=\tau}^j$ and $ \equ_{R=\tau}^j$ are both true by definition so the implications automatically hold.
	
	\begin{sloppypar}
	\noindent\textbf{Inductive Step $j$:} 
	$\neg \equ_{L=\tau}^{j} \rightarrow \neg \equ_{L=\tau}^{j-1} \vee (\set^j_L\oplus\set^j_\tau ) \vee (\set^j_L\wedge (\val^j_L \oplus \val^j_\tau) )$ and  $(\set^j_L\oplus\set^j_\tau ) \vee (\set^j_L\wedge (\val^j_L \oplus \val^j_\tau) )\rightarrow \diff_L^j \vee \neg \equ^{j-1}_{R=\tau}$ so we get $\neg \equ_{L=\tau}^{j} \rightarrow \neg \equ_{L=\tau}^{j-1} \vee \diff_L^j \vee \neg \equ^{j-1}_{R=\tau}$, which using the induction hypothesis to remove 
	$\neg \equ^{j-1}_{L=\tau}$ and $\neg \equ^{j-1}_{R=\tau}$ gives us $\neg \equ_{L=\tau}^{j} \rightarrow \diff_R^{j-1} \vee \diff_L^{j-1}$ which 
	can be weakened to $\neg \equ_{L=\tau}^{j} \rightarrow \diff_R^j \vee \diff_L^j$ which is equivalent to $\neg \diff_L^j \wedge \neg \diff_R^j \rightarrow \equ^j_{L=\tau}$. This is done similarly when swapping $L$ and $R$. 
	\end{sloppypar}
	
	We can obtain the remaining propositions as corollaries by using the definition of $\equ$.
\end{proof}

\begin{lem}\label{lem:LR}
	For any $0\leq i\leq m$
	the following propositions are true and have short Extended Frege proofs.
	\begin{itemize}
		\item $\diff_L^i \rightarrow (\val^i_B \leftrightarrow \val^i_L)\wedge (\set^i_B \leftrightarrow \set^i_L)$
		\item $\neg \diff_L^i\wedge \diff_R^i \rightarrow (\val^i_B \leftrightarrow \val^i_R)\wedge (\set^i_B \leftrightarrow \set^i_R)$
	\end{itemize}
\end{lem}

\begin{proof}
	
	Suppose we want to prove $\diff_L^i \rightarrow (\val^i_B \leftrightarrow \val^i_L)\wedge (\set^i_B \leftrightarrow \set^i_L)$. 
	We will assume the definition 
	$$\textstyle\diff_L^i:= \diff_L^{i-1} \vee (\equ_{R=\tau}^{i-1}\wedge ((\set^i_L\oplus \set^i_\tau)\vee (\set^i_\tau\wedge(\val^i_L\oplus \val^i_\tau))))$$
	and show that following proposition (that determines $B$) is falsified
	$$\textstyle
	\neg \diff_L^{i-1} \wedge (\diff_R^{i-1} \vee (\neg\set_{\tau}^i \wedge \neg\set_L^i \wedge \set_R^i) \vee (\set_{\tau}^i \wedge \set_L^i \wedge (\val_\tau^i \leftrightarrow \val_L^i)))
	$$
	
	The first thing is that we only need to consider $\diff_L^i\wedge \neg \diff_L^{i-1}$ as $\diff_L^{i-1}$ already falsifies our proposition.
	Next we show $\neg \diff_R^{i-1}$ is forced to be true in this situation. To do this we need Lemma~\ref{lem:rel} for
	$\diff_L^i\wedge \neg\diff_L^{i-1} \rightarrow \neg \diff_R^{i-1}$.

	Now  we use $\diff_L^i\wedge \neg \diff_L^{i-1}\rightarrow ((\set^i_L\oplus \set^i_\tau)\vee (\set^i_\tau\wedge(\val^i_L\oplus \val^i_\tau)))$, we break this down into three cases 
	
	\begin{enumerate}
		\item $\diff_L^i\wedge\neg \diff_L^{i-1} \wedge \neg \set^i_L\wedge \set^i_\tau$
		\item $\diff_L^i\wedge\neg \diff_L^{i-1} \wedge  \set^i_L\wedge \neg\set^i_\tau$
		\item $\diff_L^i\wedge\neg \diff_L^{i-1}\wedge (\set^i_\tau\wedge(\val^i_L\oplus \val^i_\tau))$
	\end{enumerate}
	
	\begin{enumerate}
		\item $\diff_L^i\wedge\neg \diff_L^{i-1}$ contradicts $\diff_R^{i-1}$,
		$\set^i_\tau$ contradicts $(\neg\set_{\tau}^i \wedge \neg\set_L^i \wedge \set_R^i)$,
		and 
		$\neg \set^i_L$ contradicts $(\set_{\tau}^i \wedge \set_L^i \wedge (\val_\tau^i \leftrightarrow \val_L^i))$.
		
		\item $\diff_L^i\wedge\neg \diff_L^{i-1}$ contradicts $\diff_R^{i-1}$,
		$\set^i_L$ contradicts $(\neg\set_{\tau}^i \wedge \neg\set_L^i \wedge \set_R^i)$,
		and 
		$\neg \set^i_\tau$ contradicts $(\set_{\tau}^i \wedge \set_L^i \wedge (\val_\tau^i \leftrightarrow \val_L^i))$.

		\item \begin{sloppypar}$\diff_L^i\wedge\neg \diff_L^{i-1}$ contradicts $\diff_R^{i-1}$,
		$\set^i_\tau$ contradicts $(\neg\set_{\tau}^i \wedge \neg\set_L^i \wedge \set_R^i)$ and
		$(\val^i_L\oplus \val^i_\tau)$ contradicts $(\set_{\tau}^i \wedge \set_L^i \wedge (\val_\tau^i \leftrightarrow \val_L^i))$.
		\end{sloppypar}
	\end{enumerate}
	\begin{sloppypar}
		Since in all cases we contradict $ \neg \diff_L^{i-1} \wedge (\diff_R^{i-1} \vee (\neg\set_{\tau}^i \wedge \neg\set_L^i \wedge \set_R^i) \vee (\set_{\tau}^i \wedge \set_L^i \wedge (\val_\tau^i \leftrightarrow \val_L^i)))$ then as per definition ($\val_B,\set_B$)=($\val_L,\set_L$).
		Using $\diff_L^i\rightarrow (\diff_L^i\wedge\neg \diff_L^{i-1})\vee \diff_L^{i-1}$ we get $\diff_L^i \rightarrow (\val^i_B \leftrightarrow \val^i_L)\wedge (\set^i_B \leftrightarrow \set^i_L)$, in a polynomial number of Frege lines.

		Now we suppose we want to prove the second proposition $\neg \diff_L^i \wedge \diff_R^i\rightarrow (\val^i_B \leftrightarrow \val^i_R)\wedge (\set^i_B \leftrightarrow \set^i_R)$. 
		We need $\neg \diff_L^i \wedge \diff_R^i$ to satisfy 
		$
		\neg \diff_L^{i-1} \wedge 
		(
		\diff_R^{i-1} \vee 
		(
		\neg \set_{\tau}^i \wedge \neg\set_L^i \wedge \set_R^i
		) 
		\vee 
		(
		\set_{\tau}^i \wedge \set_L^i \wedge 
		(
		\val_{\tau}^i \leftrightarrow \val_L^i
		)
		)
		)
		$ instead.
	\end{sloppypar}	
	
	Lemma~\ref{lem:impl} gives us that $\neg \diff_L^i\rightarrow \neg \diff_L^{i-1}$. $\diff_R^{i-1}$ is enough to satisfy the formula, so the case we need to explore is when $\diff_R^{i-1}$ is false.
	We can show that $\neg \diff_L^{i-1}\wedge\neg \diff_R^{i-1} \rightarrow \equ_{L=\tau}^{i-1}$ using Lemma~\ref{lem:nLnR}. This allows us to examine just the part where $\diff_R$ is being triggered to be true by definition: $\neg \diff_L^i \wedge \neg \diff_R^{i-1} \rightarrow (\set_{\tau}^i \leftrightarrow \set_L^i)\wedge (\set_{\tau}^i\rightarrow (\val_\tau^i \leftrightarrow \val_L^i))$.
	
	Suppose the term $(\neg\set_{\tau}^i \wedge \neg\set_L^i \wedge \set_R^i)$ is false, assuming $\diff_R^{i-1}$ is also false,  we have to show that $(\set_{\tau}^i \wedge \set_L^i \wedge (\val_\tau^i \leftrightarrow \val_L^i)$ will be satisfied. We look at the three ways the term $(\neg\set_{\tau}^i \wedge \neg\set_L^i \wedge \set_R^i)$ can be falsified and show that all the parts of the remaining term must be satisfied when assuming $\neg \diff_L^{i}\wedge \diff_R^{i}\wedge \neg \diff_R^{i-1}$.
	
	\begin{enumerate}
		\item $\set_{\tau}^i$, in this case $(\val_\tau^i \leftrightarrow \val_L^i)$ is active and $\set_L^i$ is implied by $(\set_{\tau}^i \leftrightarrow \set_L^i)$.
		\item  $\set_{L}^i$, $\set_\tau^i$ is implied by $(\set_{\tau}^i \leftrightarrow \set_L^i)$, then $(\val_\tau^i \leftrightarrow \val_L^i)$ is active.
		\item $\neg \set_{R}^i$, then using $\diff_{R}^i$ and $\neg \diff_{R}^{i-1}$ we must have $\set_{\tau}^i$ (as this is the only allowed way $\diff$ can trigger). Once again, $(\val_\tau^i \leftrightarrow \val_L^i)$ is active and $\set_L^i$ is implied by $(\set_{\tau}^i \leftrightarrow \set_L^i)$.
	\end{enumerate}
	
	\begin{sloppypar}
	Since our trigger formula is always satisfied when $\neg \diff_L^{i}\wedge \diff_R^{i}\wedge \neg \diff_R^{i-1}$, it means that $(\val^i_B,\set^i_B)=(\val^i_R,\set^i_R)$.
	Using $\diff_R^i\rightarrow (\diff_R^i\wedge\neg \diff_R^{i-1})\vee \diff_R^{i-1}$ we get $\neg \diff_L^i\wedge \diff_R^i \rightarrow (\val^i_B \leftrightarrow \val^i_R)\wedge (\set^i_B \leftrightarrow \set^i_R)$, in a polynomial number of Frege lines.
	\end{sloppypar}
\end{proof}

\begin{lem}\label{lem:Bdiff}
	The following propositions are true and have short Extended Frege proofs.
	\begin{itemize}
		\item $B\wedge \diff_L^m\rightarrow B_L$
		\item $B\wedge \neg \diff_L^m\wedge \diff_R^m\rightarrow B_R$
	\end{itemize}
	
\end{lem}

\begin{proof}
	We use the disjunction
	$\diff_L^m \rightarrow \bigvee_{j=1}^m \diff_L^j \vee \neg \diff_L^{j-1}$ from Lemma~\ref{lem:chain}.
	So there is some $j$ where this is the case. $i$ can be looked at in cases, where $(\val_B^i,\set_B^i)$ is determined by Definition~\ref{def:B}.
	\begin{itemize}
		\item For $1\leq i< j$ observe that
		$\diff_L^j \vee \neg \diff_L^{j-1} \rightarrow \neg \diff_R^{j-1}$.
		Now these negative literals propagate downwards.
		$\neg \diff_L^{j-1} \wedge \neg \diff_R^{j-1} \rightarrow \neg \diff_L^{i} \wedge \neg \diff_R^{i}$ for $0\leq i<j$
		and $\neg \diff_L^{i} \wedge \neg \diff_R^{i}$ means that $B$ and $L$ are consistent  for those $i$ as proven in Lemma~\ref{lem:nLnR}.
		
		\item For $j\leq i\leq m$, $\diff_L^j \rightarrow \diff_L^i$ and $\diff_L^i$ means $B$ and $L$ are consistent on those $i$ as proven in Lemma~\ref{lem:LR}.
		
		\item For indices greater than $m $, $B\wedge \diff_L^m$ falsifies $\neg \diff_L^m \wedge (\diff_R^m \vee \bar x)$, so $B$ and $L$ are consistent on those indices.
		
	\end{itemize}
	With the second proposition $\diff_R^m \rightarrow \bigvee_{j=1}^m \diff_R^j \vee \neg \diff_R^{j-1}$ once again. 
	So there is some $j$ where this is the case.
	Note that $\neg \diff_L^m \rightarrow \neg \diff_L^i$ for $i\leq m$.
	\begin{itemize}
		\item For $1\leq i< j$, both $\neg \diff_L^i$ and $\neg \diff_R^i$ occur so then $B$ and $R$ are consistent for these values.
		\item For $j\leq i\leq m$, $\diff_R^j \rightarrow \diff_R^i$ and $\diff_R^i\wedge  \neg \diff_L^i$ means $B$ and $R$ are consistent on those $i$ as proven in Lemma~\ref{lem:LR}.
		\item For indices greater than $ m $, $B\wedge \diff_R^m \wedge \neg \diff_L^m$ satisfies $\neg \diff_L^m \wedge (\diff_R^m \vee \bar x)$, so $B$ and $R$ are consistent on those indices.
	\end{itemize} 
	
	Each of these use a polynomial number of Frege steps and uses of previous lemmas (each of which consist of a polynomial number of Frege steps).
\end{proof}

\begin{lem}\label{lem:Bndiff}
	The following propositions are true and have short Extended Frege proofs.
	\begin{itemize}
		\item $B\wedge \neg \diff_L^m\wedge\neg \diff_R^m \rightarrow B_L \vee \neg x$
		\item $B\wedge \neg \diff_L^m\wedge\neg \diff_R^m \rightarrow B_R \vee  x$
	\end{itemize}
	
\end{lem}

\begin{proof}
	For indices $1\leq i\leq m$, since $\neg \diff_{L}^m \rightarrow \neg \diff_{L}^i$ and $\neg \diff_{R}^m \rightarrow \neg \diff_{R}^i$, Lemma~\ref{lem:nLnR} can be used to show that $B\wedge \neg\diff_{L}^m \wedge \neg\diff_{R}^m$ leads to $\set_B^i=\set_L^i=\set_R^i$ and $\val_B^i=\val_L^i=\val_R^i$ whenever $\set_B^i$ is also true. Extended Frege can prove the  $O(m)$ propositions that show these equalities for $1\leq i \leq m$.
	
	For $i>m$, by definition $B\wedge \neg \diff_L^m\wedge\neg \diff_R^m \wedge x$ gives $\set_B^i=\set_L^i$ and $\val_B^i=\val_L^i$. And $B\wedge \neg \diff_L^m\wedge\neg \diff_R^m \wedge \neg x$ gives $\set_B^i=\set_R^i$ and $\val_B^i=\val_R^i$.
	The sum of this is that $B\wedge \diff_{L}^m \wedge \diff_{R}^m \wedge x \rightarrow B_L$ and $B\wedge \diff_{L}^m \wedge \diff_{R}^m \wedge \neg x \rightarrow B_R$.
\end{proof}

\begin{lem}\label{lem:ir1}
	The following proposition is true and has a short Extended Frege proof.
	$B\rightarrow B_L \vee B_R$
\end{lem}

\begin{proof}
	This roughly says that $B$ either is played entirely as $L$ or is played as $R$. We can prove this by combining Lemmas~\ref{lem:Bdiff}~and~\ref{lem:Bndiff}, it essentially is a case analysis in formal form. 
\end{proof}

\begin{lem}\label{lem:ir2}
	The following propositions are true and have short Extended Frege proofs.
	\begin{itemize}
		\item $B\wedge \ann_{x,B}(\tau) \wedge x\rightarrow B_L$,
		
		\item $B\wedge \ann_{x,B}(\tau) \wedge \neg x\rightarrow B_R$
	\end{itemize}
\end{lem}

\begin{proof}
	We start with 
	$B\wedge \neg \diff_L^m\wedge\neg \diff_R^m \rightarrow B_L \vee \neg x$
	and  $B\wedge \neg \diff_L^m\wedge\neg \diff_R^m \rightarrow B_R \vee  x$.
	It remains to remove $\neg \diff_L^m\wedge\neg \diff_R^m$ from the left hand side. This is where we use $L\wedge \diff_L^i \rightarrow \neg \ann_{x,L}(\tau)$
	and  $R\wedge \diff_R^i \rightarrow \neg \ann_{x,R}(\tau)$ from Lemma~\ref{lem:tau}. These can be simplified to $B\wedge B_L\wedge \diff_L^m \rightarrow \neg \ann_{x,B}(\tau)$
	and  $B\wedge B_R\wedge \diff_R^m \rightarrow \neg \ann_{x,B}(\tau)$. 
	The $B_L$ and $B_R$ can be removed by using $B\wedge \diff_L^m\rightarrow B_L$
	and $B\wedge \neg \diff_L^m\wedge \diff_R^m\rightarrow B_R$ and we can end up with $B \rightarrow \neg \ann_{x,B}(\tau)\vee (\neg\diff_R^m\wedge \neg\diff_L^m)$. We can use this to resolve out $(\neg\diff_R^m\wedge \neg\diff_L^m)$ and get $B\wedge \ann_{x,B}(\tau) \wedge x\rightarrow B_L$ and 
	$B\wedge \ann_{x,B}(\tau) \wedge \neg x\rightarrow B_R$.
\end{proof}

\begin{proof}[Proof of Lemma~\ref{thm:res}]
	Since $B\wedge B_L\rightarrow L $ and $B\wedge B_R\rightarrow R $,
	$L\rightarrow \con_{L}(C_1 \vee \neg x^\tau)$ and $R\rightarrow \con_{R}(C_2 \vee x^\tau)$ imply
	$B\wedge B_L \rightarrow \con_{B}(C_1 \vee C_2)\vee \ann_{x,B}(\tau)$, 
	$B\wedge B_R \rightarrow \con_{B}(C_1 \vee C_2)\vee \ann_{x,B}(\tau)$,
	$B\wedge B_L \rightarrow \con_{B}(C_1 \vee C_2)\vee \neg x$ and
	$B\wedge B_R \rightarrow \con_{B}(C_1 \vee C_2)\vee x$.
	
	We combine $B \rightarrow B_L \vee B_R$
	with $B\wedge B_L \rightarrow \con_{B}(C_1 \vee C_2)\vee \ann_{x,B}(\tau)$ (removing $B_L$) and
	$B\wedge B_R \rightarrow \con_{B}(C_1 \vee C_2)\vee \ann_{x,B}(\tau)$ (removing $B_R$) to gain 
	$B \rightarrow \con_{B}(C_1 \vee C_2)\vee \ann_{x,B}(\tau)$.
	Next, we aim to derive $B \rightarrow \con_{B}(C_1 \vee C_2)\vee \neg \ann_{x,B}(\tau)$.
        Policy $B$ is set up so that $B\wedge \ann_{x,B}(\tau) \wedge x \rightarrow B_L$ and $B \wedge \ann_{x,B}(\tau) \wedge \neg x \rightarrow B_R$ have short proofs (Lemma~\ref{lem:ir2}).
         We resolve these, respectively, with
       	$B\wedge B_R \rightarrow \con_{B}(C_1 \vee C_2)\vee x$ (on $x$) to obtain
        $B \wedge \ann_{x,B}(\tau) \wedge B_R \rightarrow B_L\vee \con_B(C_1 \vee C_2)$, and with
        $B \wedge B_L \rightarrow \con_{B}(C_1 \vee C_2)\vee \neg x$ (on $\neg x$) to obtain
        $B \wedge \ann_{x,B}(\tau) \wedge B_L \rightarrow B_R \vee \con_B(C_1 \lor C_2)$.
        Putting these together allows us to remove $B_L$ and $B_R$, deriving $B \wedge \ann_{x,B}(\tau) \rightarrow \con_B(C_1 \vee C_2)$, which can be rewritten as $B \rightarrow \con_B(C_1 \vee C_2) \vee \neg \ann_{x,B}(\tau)$.

        We now have two formulas $B \rightarrow \con_{B}(C_1 \vee C_2)\vee \neg \ann_{x, B}(\tau)$  and $B \rightarrow \con_{B}(C_1 \vee C_2)\vee  \ann_{x, B}(\tau)$, which resolve to get $B \rightarrow \con_{B}(C_1 \vee C_2)$.
\end{proof}

\begin{thm}\label{thm:irc}
	\eFregeRed p-simulates \irc.
\end{thm}

\begin{proof}
	We prove by induction that every annotated clause $C$ appearing in an \irc proof has a local policy $S$ such that $\phi \vdash_{\eFrege} S\rightarrow \con_S(C)$ and this can be done in a polynomial-size proof.
	
	\textbf{Axiom:}
	Suppose $C\in \phi$ and $D= \instantiate(C, \tau)$ for partial annotation $\tau$.
	We construct policy~$B$ such that $B\rightarrow \con_B(D)$ can be derived from $C$.

\begin{center}{$\set^j_B= \begin{cases} 1 & \text{if }u_j\in \domain(\tau)\\ 0& u_j\notin \domain(\tau)\end{cases}$, $\val^j_B= \begin{cases} 1 & \text{ if }1/u_j\in\tau\\ 0& \text{ if }0/u_j\in\tau\end{cases}$}
	
\end{center}

	\textbf{Instantiation:}
	Suppose we have an instantiation step for $C$ on a single universal variable $u_i$ using instantiation $0/u_i$, so the new annotated clause is $D= \instantiate(C, 0/u_i)$.
	From the induction hypothesis $T\rightarrow \con_T(C)$ we will develop $B$ such that $B\rightarrow \con_B(D)$.

\begin{center}{$\set^j_B= \begin{cases} 1 & \text{if }j=i\\ \set^j_T & \text{if }j\neq i\end{cases}$, $\val^j_B= \begin{cases} \val^j_T \wedge \set^j_T & \text{if }j=i\\ \val^j_T & \text{if }j\neq i\end{cases}$}
\end{center}

$\val^j_T \wedge \set^j_T$ becomes $\val^j_T \vee \neg \set^j_T$ for instantiation by $1/u_j$. Either case means $B$ satisfies the matching annotations $\ann$ as $T$ appearing in our converted clauses $\con_B(C)$ and  $\con_B(D)$, proving the rule as an inductive step.  
	
	\textbf{Resolution:}
	See Lemma~\ref{thm:res}.
	
	\textbf{Contradiction:}
	At the end of the proof we have $T\rightarrow \con_T(\bot)$.
	$T$ is a policy, so we turn it into a full strategy $B$ by having for each $i$: $\val^i_B \leftrightarrow (\val^i_T \wedge \set^i_T)$ and $\set^i_B= 1$.
	Effectively this instantiates $\bot$ by the assignment that sets everything to $0$ and we can argue that $B\rightarrow \con_B(\bot)$ although $\con_B(\bot)$ is just the empty clause. So we have $\neg B$. 
	But $\neg B$ is just $\bigvee_{i=1}^n (u_i \oplus  \val_B^i)$.
	Furthermore, just as in Schlaipfer et al.'s work~\cite{SSWZ20}, we have been careful with the definitions of the extension variables $\val_B^i$ so that they are left of $u_i$ in the prefix. 
	In \eFregeRed we can use the reduction rule (this is the first time we use the reduction rule). 
	We show an inductive proof of $\bigvee_{i=1}^{n-k} (u_i\oplus  \val_B^i)$ for increasing $k$  eventually leaving us with the empty clause. This essentially is where we use the \Ared rule. Since we already have $\bigvee_{i=1}^{n} (u_i\oplus  \val_B^i)$ we have the base case and we only need to show the inductive step.
	
	We derive from $\bigvee_{i=1}^{n+1-k} (u_i\oplus  \val_B^i)$ both $(0\oplus  \val_B^{n-k+1})\vee\bigvee_{i=1}^{n-k} (u_i\oplus  \val_B^i)$ and $(1\oplus  \val_B^{n-k+1})\vee\bigvee_{i=1}^{n-k} (u_i\oplus  \val_B^i)$ from reduction.
	We can resolve both to derive $\bigvee_{i=1}^{n-k} (u_i\oplus  \val_B^i)$.
	
	We continue this until we reach the empty disjunction.
\end{proof}
\begin{cor}
	\eFregeRed p-simulates \ecalculus. 
\end{cor}

	While this can be proven as a corollary of the simulation of \irc, a more direct simulation can be achieved by defining the resolvent strategy by removing the $\set^i$ variables (i.e. by considering them as always true).

\section{Extended Frege+\texorpdfstring{\Ared}{∀-Red} p-simulates \irmc}\label{sec:irm}

\subsection{\irmc}
\irmc was designed to compress annotated literals in clauses in order to simulate \lqrc \cite{BCJ14}. Like that system it uses the $*$ symbol, but since universal literals do not appear in an annotated clause, the $*$ value is added to the annotations, $0/u, 1/u, */u$ being the first three possibilities in an extended annotation (we can consider the fourth to be when $u$ does not appear in the annotation).

\begin{figure}[h]
			Axiom and instantiation rules as in \irc in Figure~\ref{fig:IR}.
			\begin{prooftree}
				\AxiomC{$x^{\tau\cup\xi}\lor C_1 $ }
				\AxiomC{$\lnot x^{\tau\cup\sigma}\lor C_2 $}
				\RightLabel{(Resolution)}
				\BinaryInfC{$\instantiate(\sigma,C_1)\cup\instantiate(\xi,C_2)$}
			\end{prooftree}
			\text{$\domain(\tau)$, $\domain(\xi)$ and $\domain(\sigma)$ are mutually disjoint.}\\\text{$\tau$ is a partial assignment to the universal variables with $\mathsf{codomain}(\tau)=\{0,1\}$. }\\\text{$\sigma$ and $\xi$ are extended partial assignments with $\mathsf{codomain}(\sigma)=\mathsf{codomain}(\xi)=\{0,1, *\}$.}
			\begin{prooftree}
				\AxiomC{$C\lor b^\mu\lor b^\sigma$}
				\RightLabel{(Merging)}
				\UnaryInfC{$C\lor b^\xi$}
			\end{prooftree}
			\text{$\domain(\mu)=\domain(\sigma)$. $\xi=\comprehension{c/u}{c/u\in\mu,c/u\in\sigma}\cup
				\comprehension{*/u}{c/u\in\mu,d/u\in\sigma,c\neq d}$ }.
			\caption{ The rules of \irmc \cite{BCJ19}.
			\label{fig:IRM}}
\end{figure}

The rules of \irmc as given in Figure~\ref{fig:IRM}, become more complicated as a result of the $*/u$ annotations. In particular resolution is no longer done  between matching pivots but matching is done internally in the resolution steps.
$*/u$ annotations are meant to represent ambiguous annotations so it could mean a  pair of pivots literals that each have a $*/u$ annotation do not actually match on $u$. The solution to this is to allow compatibility where one pivot has a $*/u$ annotation where the other has no annotation in $u$. The idea is that the blank annotation is instantiated on-the-fly with the correct function for $*/u$ so that the annotations truly match. The resolvent takes this into account by joining the instantiated clauses minus the pivot. 

Additionally in order to introduce $*$ annotations a merge rule is used.  

It is in \irmc where the positive $\set$ literals introduced in the simulation of \irc become useful. In most ways $\set^i_S$ asserts the same things as $*/u_i$, that $u_i$ is given a value, but this value does not have to be specified.

\subsection{Policies and Simulating \irmc}
\subsubsection{Conversion}

The first major change from \irc is that while $\ann_S$ worked on three values in \irc, in \irmc we effectively run in four values $\set^i_S, \neg\set^i_S, \set^i_S\wedge u_i$ and $\set^i_S\wedge \neg u_i$. $\set^i_S$ is the new addition deliberately ambiguous as to whether $u_i$ is true or false. Readers familiar with the $*$ used in \irmc may notice why $\set^i_S$ works as a conversion of $*/u_i$, as $\set^i_S$ is just saying our policy has given a value but it may be different values in different circumstances. 

\noindent$\ann_{x,S}(\tau)= \bigwedge_{1/u_i\in \tau} (\set_S^i \wedge u_i) \wedge \bigwedge_{0/u_i\in \tau} (\set_S^i \wedge\bar u_i) \wedge \bigwedge_{*/u_i\in \tau} (\set_S^i) \wedge\bigwedge_{u_i\notin \domain(\tau)} (\neg \set_S^i)$.

\noindent$\con_{S}(x^\tau)= x \wedge \ann_{x,S}(\tau)$, $\con_{S}(C_1)= \bigvee_{x^\tau\in C_1} \con(x^\tau)$

\subsubsection{Policies}
Like in the case of \irc, most work needs to be done in the \irmc resolution steps, although here it is even more complicated. A resolution step in \irmc is in two parts. 
Firstly $C_1 \vee \neg x^{\tau \sqcup \sigma}$, $C_2 \vee x^{\tau \sqcup \xi}$ are both instantiated (but by $*$ in some cases), secondly they are resolved on a matching pivot. We simplify the resolution steps so that  $\sigma$ and $\xi$ only contain $*$ annotations, for the other constant annotations that would normally be found in these steps suppose we have already instantiated them in the other side so that they now appear in $\tau$ (this does not affect the resolvent). 

Again we assume that there are policies $L$ and $R$ such that $L\rightarrow \con_L(C_1 \vee \neg x^{\tau \sqcup \sigma})$ and  $R\rightarrow \con_R(C_2 \vee x^{\tau \sqcup \xi})$.
We know that if $L$ falsifies $\ann_{x,L}({\tau \sqcup \sigma})$ then $\con_L(C_1)$ and likewise if $R$ falsifies $\ann_{x,R}({\tau \sqcup \xi})$ then $\con_R(C_2)$ is satisfied.
These are the safest options, however this leaves cases when $L$ satisfies  $\ann_{x,L}({\tau \sqcup \sigma})$ and $R$ satisfies $\ann_{x,R}({\tau \sqcup \xi})$ but $L$ and $R$ are not equal. 
This happens either when $\set^i_L$ and $\neg \set^i_R$ both occur for $*/u_i\in \sigma$ or  when $\neg \set^i_L$ and $ \set^i_R$ both occur for $*/u_i\in \xi$.

This would cause an issue if $B$ had to choose between $L$ and $R$ to satisfy $\con_B(C_1\vee C_2)$, as previously in \irc we would be able to be agreeable to both $L$ and $R$ and defer our choice later down the prefix (which could be necessary). 
Fortunately, we are not trying to satisfy $\con_B(C_1\vee C_2)$ but $\con_B(\instantiate(\xi, C_1)\vee\instantiate(\sigma, C_2))$, so we have to choose between a policy that will satisfy $\con_B(\instantiate(\xi, C_1))$ and a policy that will satisfy $\con_B(\instantiate(\sigma, C_2))$. This is similar to doing the internal instantiation steps separately from the resolution steps, but the instantiation step need a slight bit more care as they instantiate by functions rather than constants. What this looks like is that in addition to $L$ we will occasionally borrow values from $R$ and vice versa. 
By borrowing values from the opposite policy we obtain a working new policy that does not have to choose between left and right any earlier than we would have for \irc. 

\subsubsection{Difference and Equivalence Variables}

We update our functions to take into account the 4 values. Note here again we assume $\sigma$ and $\xi$ only contain $*$ annotations.

\noindent$\equ^{0}_{f=g}:=1$

\noindent $\equ^{i}_{f=g}:=\equ^{i-1}_{f=g}\wedge (\set^i_f\leftrightarrow \set^i_g)\wedge (\set^i_f\rightarrow(\val^i_f \leftrightarrow \val^i_g) )$ when $*/u_i\notin g$

\noindent $\equ^{i}_{f=g}:=\equ^{i-1}_{f=g}\wedge \set^i_f$ when $*/u_i\in g$

\noindent$\diff_L^0:= 0$ and 
\noindent$\diff_R^0:= 0$

\noindent{For $u_i\notin \domain(\tau\sqcup \sigma \sqcup \xi)$},
$\diff_L^i:= \diff_L^{i-1} \vee (\equ_{R=\tau\sqcup \xi}^{i-1}\wedge \set^i_L)$,

\noindent{For $u_i\in \domain(\tau)$}, $\diff_L^i:= \diff_L^{i-1} \vee (\equ_{R=\tau\sqcup \xi}^{i-1}\wedge (\neg\set^i_L\vee (\set^i_\tau\wedge(\val^i_L\oplus \val^i_\tau))))$

\noindent {For $u_i\in \domain(\sigma)$},
$\diff_L^i:= \diff_L^{i-1} \vee (\equ_{R=\tau\sqcup \xi}^{i-1}\wedge \neg\set^i_L)$

\noindent{For $u_i\in \domain(\xi)$}, $\diff_L^i:= \diff_L^{i-1} \vee (\equ_{R=\tau\sqcup \xi}^{i-1}\wedge \set^i_L)$

\subsubsection{Policy Variables}
We define the policy variables $\val^i_B$ and $\set^i_B$ based on a number of cases, in all cases $\val_B^i$ and $\set_B^i$ are defined on variables left of $u_i$.

{For $u_i\notin \domain(\tau\sqcup \sigma \sqcup \xi)$, $u_i<x$},

\noindent$(\val^i_B, \set^i_B)=
\begin{cases} 
	(\val^i_R, \set^i_R)& \text{if } \neg \diff_{L}^{i-1}\wedge(\diff_R^{i-1} \vee\neg \set_L^{i})\\ 
	(\val^i_L, \set^i_L) & \text{otherwise. }
\end{cases}$

{For $u_i\in\domain(\tau)$},

\noindent$(\val^i_B, \set^i_B)= 
\begin{cases} 
	(\val^i_R, \set^i_R)& \text{if } \neg \diff_{L}^{i-1}\wedge(\diff_R^{i-1} \vee(\set_L^{i} \wedge (\val_L^{i}\leftrightarrow \val_\tau^{i})))\\ 
	(\val^i_L, \set^i_L) & \text{otherwise. }\end{cases}$

{For $*/u_i\in\sigma$},

\noindent$(\val^i_B, \set^i_B)= 
\begin{cases}
	(0,1)  & \text{if }\neg \diff_L^{i-1}\wedge\diff_R^{i-1}\wedge \neg \set_R^i\\
	(\val^i_R, \set^i_R)& \text{if } \neg \diff_{L}^{i-1}\wedge \set_R^{i}\wedge(\diff_R^{i-1} \vee\set_L^{i})\\
	(\val^i_L, \set^i_L) & \text{otherwise. }
\end{cases}$

{For $*/u_i\in\xi$},

\noindent$(\val^i_B, \set^i_B)= 
\begin{cases}
	(0,1)  & \text{if }\diff_L^{i-1}\wedge \neg \set_L^i\\
	(\val^i_R, \set^i_R)& \text{if } \neg \diff_{L}^{i-1}\wedge(\diff_R^{i-1} \vee\neg\set_L^{i})\\
	(\val^i_L, \set^i_L) & \text{otherwise. }
\end{cases}$

For $u_i>x$,

\noindent$(\val^i_B, \set^i_B)= 
\begin{cases}
	(\val^i_R, \set^i_R)& \text{if } \neg \diff_{L}^{m}\wedge(\diff_R^{m} \vee \neg x)\\
	(\val^i_L, \set^i_L) & \text{otherwise. }
\end{cases}$

The idea for the policy $B$ is to stick to $\tau\sqcup \sigma \sqcup \xi$ until either $L$ or $R$ differ, then commit to whichever strategy that is differing (and default to $L$ when both start to differ at the same time). However there are cases where a $\set_L$ or $\set_R$ value may differ from  $\tau\sqcup \sigma \sqcup \xi$ but it should not be counted as a true difference for $L$ or $R$. An example is when $*/u_i\in \sigma$ and $\set_R$ is false, we should not commit to $R$ here, but instead borrow the set and value pair from $L$ for this case.
Once we commit to $L$ or $R$ we may still have make sure $B$ satisfies the instantiated resolvent so a few cases where we have force $\set_B$ to be true and we set $\val_B$ to be false.
Finally if no difference is found along $\tau\sqcup \sigma \sqcup \xi$ we surely have to commit to either $L$ or $R$ depending on the value of the existential literal $x$.

\subsection{Proof in \texorpdfstring{\eFregeRed}{eFrege+∀red}}

\begin{lem}\label{lem:irm:chain}
	For $0<j\leq m$ the following propositions have short derivations in Extended Frege:
	\begin{itemize}
		\item $\diff_L^j \rightarrow \bigvee_{i=1}^j \diff_L^i \wedge \neg \diff_L^{i-1}$
		\item $\diff_R^j \rightarrow \bigvee_{i=1}^j \diff_R^i \wedge \neg \diff_R^{i-1}$
		\item $\neg \equ_{L=\tau\sqcup\sigma}^j \rightarrow \bigvee_{i=1}^j \neg \equ_{L=\tau\sqcup\sigma}^i \wedge  \equ_{L=\tau\sqcup\sigma}^{i-1}$
		\item $\neg \equ_{R=\tau\sqcup\xi}^j \rightarrow \bigvee_{i=1}^j \neg \equ_{R=\tau\sqcup\xi}^i \wedge  \equ_{R=\tau\sqcup\xi}^{i-1}$
		
	\end{itemize}
\end{lem}
\begin{proof}
	The proof of Lemma~\ref{lem:chain} still works despite the modifications to definition.
\end{proof}

\begin{lem}\label{lem:irm:impl}
	For $0\leq i \leq j\leq m$ the following propositions that describe the monotonicity of $\diff$ and $\equ$ have short derivations in Extended Frege:
	\begin{itemize}
		\item $\diff_L^i \rightarrow \diff_L^j$
		\item $\diff_R^i \rightarrow \diff_R^j$
		\item $\neg \equ_{f=g}^i \rightarrow \neg \equ_{f=g}^j$
	\end{itemize}
\end{lem}

\begin{proof} The proofs of Lemma~\ref{lem:impl} still work despite the modifications to definition.
\end{proof}

\begin{lem}\label{lem:irm:rel}
	For $0\leq i \leq j\leq m$ the following propositions describe the relationships between the different extension variables
	
	\begin{itemize}
		\item $\equ^{i}_{L=\tau\sqcup\sigma}\rightarrow \neg \diff^i_L$
		\item $\diff_L^i \wedge \neg \diff_L^{i-1}\rightarrow \equ^{i-1}_{R=\tau\sqcup\xi}$
		\item $\diff_L^i \wedge \neg \diff_L^{i-1}\rightarrow \neg \diff_R^{i-1}$
		\item $\equ^{i}_{R=\tau\sqcup\xi}\rightarrow \neg \diff^i_R$
		\item $\diff_R^i \wedge \neg \diff_R^{i-1}\rightarrow \equ^{i-1}_{L=\tau\sqcup\xi}$
		\item $\diff_R^i \wedge \neg \diff_R^{i-1}\rightarrow \neg \diff_L^{i-1}$
	\end{itemize}
\end{lem}

\begin{proof}
	\noindent\textbf{Induction Hypothesis on $i$: }
	$\equ^{i}_{L=\tau\sqcup\sigma}\rightarrow \neg \diff^i_L$ in an $O(i)$-size eFrege proof.
	
	\noindent\textbf{Base Case $i=0$: } $\diff^i_L$ is defined as $0$ so $\neg \diff^i_L$ is true and trivially implied by 
	$\equ^{i}_{L=\tau\sqcup\sigma}$. 
	Frege can manage this.
	
	\begin{sloppypar}
		\noindent\textbf{Inductive Step $i+1$:}
		This breaks into cases depending on the domains of $u_{i+1}$.
		If $u_{i+1}\notin\domain(\sigma)$,
		$\equ^{i+1}_{L=\tau\sqcup\sigma}:=\equ^{i}_{L=\tau\sqcup\sigma}\wedge (\set^{i+1}_L\leftrightarrow \set^{i+1}_{\tau\sqcup\sigma})\wedge (\set^{i+1}_L\rightarrow(\val^{i+1}_L \leftrightarrow \val^{i+1}_{\tau\sqcup\sigma}) )$
		further if 
		$u_{i+1}\notin\domain(\tau\sqcup \sigma) $ 
		then
		$\diff_L^{i+1}:= \diff_L^{i} \vee (\equ_{R=\tau\sqcup \xi}^{i}\wedge \set^{i+1}_L)$.
		Note that here $\set^{i+1}_{\tau\sqcup\sigma}$ is defined as $0$ so 
		$\equ^{i+1}_{L=\tau\sqcup\sigma}\rightarrow(\equ^{i}_{L=\tau\sqcup\sigma}\wedge \neg \set^{i+1}_L)$.
		Adding the induction hypothesis gives 
		$\equ^{i+1}_{L=\tau\sqcup\sigma}\rightarrow \neg \diff^{i}_L \wedge \neg \set^{i+1}_L$. 
		Note that because $\neg \diff^{i}_L \wedge \neg \set^{i+1}_L$ directly refutes  $\diff_L^{i} \vee (\equ_{R=\tau\sqcup \xi}^{i}\wedge \set^{i+1}_L)$
		we get
		$\equ^{i+1}_{L=\tau\sqcup\sigma}\rightarrow\neg \diff^{i+1}_L$.
		Now if $u_{i+1}\in \domain (\tau) $ 
		then 
		\begin{center}$\diff_L^i:= \diff_L^{i-1} \vee (\equ_{R=\tau\sqcup \xi}^{i-1}\wedge (\neg\set^i_L\vee (\set^i_\tau\wedge(\val^i_L\oplus \val^i_\tau))))$
		\end{center}
		Now $\set^{i+1}_{\tau\sqcup\sigma}$ is defined as $1$.
		If $1/u_{i+1}\in \tau$,
		$\val^{i+1}_{\tau\sqcup\sigma}:=1$ so 
		$\diff_L^{i+1}:= \diff_L^{i} \vee (\equ_{R=\tau\sqcup \xi}^{i-1}\wedge (\neg\set^{i+1}_L\vee \val^{i+1}_L))$
		and 
		$\equ^{i+1}_{L=\tau\sqcup\sigma}
		\rightarrow\equ^{i}_{L=\tau\sqcup\sigma}\wedge \set^{i+1}_{L}\wedge\val^{i+1}_L$.
		Adding the induction hypothesis gives 
		$\equ^{i+1}_{L=\tau\sqcup\sigma}\rightarrow \neg \diff_L^i\wedge \set^{i+1}_L\wedge\val^{i+1}_L$.
		But $\neg \diff^{i}_{L}\wedge\set^{i+1}_L\wedge\val^{i+1}_L$ falsifies $\diff_L^{i} \vee (\equ_{R=\tau\sqcup \xi}^{i}\wedge (\neg\set^{i+1}_L\vee (\set^{i+1}_L\wedge(\val^{i+1}_L))))$.
		So $\equ^{i+1}_{L=\tau\sqcup\sigma}\rightarrow \neg \diff_L^{i+1}$.
		This works similarly if $0/u_{i+1}\in \tau$.
		If $u_{i+1}\in\domain(\sigma)$,
		$\equ^{i+1}_{L=\tau\sqcup\sigma}:=\equ^{i}_{L=\tau\sqcup\sigma}\wedge \set^{i+1}_L$
		and
		$\diff_L^{i+1}:= \diff_L^{i} \vee (\equ_{R=\tau\sqcup \xi}^{i}\wedge\neg \set_L^{i+1})$.
		But adding from the induction hypothesis we can have 
		$\equ^{i+1}_{L=\tau\sqcup\sigma} \rightarrow \neg \diff_{L}^i \wedge \set^{i+1}_L$
			and $\neg \diff_{L}^i \wedge \set^i_L$ directly contradicts 
		$\diff_L^{i} \vee (\equ_{R=\tau\sqcup \xi}^{i}\wedge\neg \set_L^{i+1})$
		so then 
		$\equ^{i+1}_{L=\tau\sqcup\sigma} \rightarrow \neg \diff_L^i$.
		Each case require a constant number of Frege steps.
	\end{sloppypar}
	
	In every case $\diff_L^{i}= \diff_L^{i-1} \vee (\equ^i_{R=\tau\sqcup \xi}\wedge A)$	where $A$ is a formula dependent on the domain of $u_i$
	$\neg \diff_L^{i-1} \wedge \diff_L^{i}$
	means that $\equ^i_{R=\tau\sqcup \xi}$ must be true.
	So we have $\diff_L^i \wedge \neg \diff_L^{i-1}\rightarrow \equ^{i-1}_{R=\tau\sqcup\xi}$ in a constant size eFrege proof.
	
	If we combine the above we have a linear size proof of 
	$\diff_L^i \wedge \neg \diff_L^{i-1}\rightarrow \diff_R^{i-1}$.
	The same proofs symmetrically work for $R$.
\end{proof}

\begin{lem}\label{lem:irm:tau}
	For any $0\leq i\leq m$
	the following propositions are true and have short Extended Frege proofs.
	
	\begin{itemize}
		\item $L\wedge \diff_L^i \rightarrow \neg \ann_{x,L}(\tau\sqcup\sigma)$
		\item $R\wedge \diff_R^i \rightarrow \neg \ann_{x,R}(\tau\sqcup\xi)$
	\end{itemize}
\end{lem}

\begin{proof}
	
	If $u_i\notin \domain(\tau\sqcup \sigma)$, then
	$\diff_L^i \wedge \neg \diff_L^{i-1} \rightarrow \set^i_{L}$ is a simple corollary of the definition line
	$\diff_L^i \leftrightarrow \diff_L^{i-1} \vee (\equ_{R=\tau\sqcup \xi}^{i-1}\wedge \set^i_L)$.
	But as $\ann_{x,L}(\tau\sqcup\sigma)$ insists on $\neg\set^i_{L}$,
	we can get 
	$\diff_L^i \wedge \neg \diff_L^{i-1} \rightarrow \neg\ann_{x,L}(\tau\sqcup\sigma)$.
	
	If $1/u_i\in \tau$, then 
	$\diff_L^i \wedge \neg \diff_L^{i-1} \rightarrow \neg \set^i_{L} \vee \neg \val^i_L$ is a simple corollary of the definition lines
	$\diff_L^i\leftrightarrow\diff_L^{i-1} \vee (\equ_{R=\tau\sqcup \xi}^{i-1}\wedge (\neg\set^i_L\vee (\set^i_\tau\wedge(\val^i_L\oplus \val^i_\tau))))$, $\set_\tau^i$ and $\val_\tau^i$
	But as $\ann_{x,L}(\tau\sqcup\sigma)$ insists on $\set^i_{L}\wedge u_i$,
	and $L$ insists on $\val_L^i\leftrightarrow u_i$
	we get
	$L \wedge \diff_L^i \wedge \neg \diff_L^{i-1} \rightarrow \neg\ann_{x,L}(\tau\sqcup\sigma) $.
	
	Similarly, if $0/u_i\in \tau$, then 
	$\diff_L^i \wedge \neg \diff_L^{i-1} \rightarrow \neg \set^i_{L} \vee  \val^i_L$ is a simple corollary of the definition lines
	$\diff_L^i\leftrightarrow \diff_L^{i-1} \vee (\equ_{R=\tau\sqcup \xi}^{i-1}\wedge (\neg\set^i_L\vee (\set^i_\tau\wedge(\val^i_L\oplus \val^i_\tau))))$, $\set_\tau^i$ and $\neg \val_\tau^i$
	But as $\ann_{x,L}(\tau\sqcup\sigma)$ insists on $\set^i_{L}\wedge \neg u_i$,
	and $L$ insists on $\val_L^i\leftrightarrow u_i$
	we get
	$\diff_L^i \wedge \neg \diff_L^{i-1} \rightarrow \neg\ann_{x,L}(\tau\sqcup\sigma) $.
	
	Finally if $*/u_i\in \sigma $, then 
	$\diff_L^i \wedge \neg \diff_L^{i-1} \rightarrow \neg \set^i_{L}$
	is a corollary of the definition line 
	$\diff_L^i\leftrightarrow \diff_L^{i-1} \vee (\equ_{R=\tau\sqcup \xi}^{i-1}\wedge \neg\set^i_L)$.
	But as $\ann_{x,L}(\tau\sqcup\sigma)$ insists on $\set^i_{L}$.
	we get
	$\diff_L^i \wedge \neg \diff_L^{i-1} \rightarrow \neg\ann_{x,L}(\tau\sqcup\sigma)$.

	$L\wedge \diff_L^i \wedge \neg \diff_L^{i-1} \rightarrow \neg\ann_{x,L}(\tau\sqcup\sigma)$ is not quite as strong as 
	$L\wedge \diff_L^i \wedge  \rightarrow \neg\ann_{x,L}(\tau\sqcup\sigma) $.
	However here we can use 
	$\diff_L^j \rightarrow \bigvee_{i=1}^j \diff_L^i \wedge \neg \diff_L^{i-1}$
	which will give us 
	$L\wedge \diff_L^j\rightarrow \neg\ann_{x,L}(\tau\sqcup\sigma)$ in a linear size proof which is also symmetric for $R$.
\end{proof}

\begin{lem}\label{lem:irm:nLnR}
	For any $0\leq j\leq m$
	the following propositions are true and have a short Extended Frege proof.
	\begin{itemize}
		\item $\neg \diff_L^j \wedge \neg \diff_R^j \rightarrow \equ^j_{L=\tau\sqcup\sigma}$ 
		\item $\neg \diff_L^j \wedge \neg \diff_R^j \rightarrow \equ^j_{R=\tau\sqcup\xi}$
		\item $\neg \diff_L^j \wedge \neg \diff_R^j \rightarrow (\neg \set_B^j \wedge \neg \set_L^j\wedge \neg \set_R^j)$ when $u_j\notin\domain(\tau\sqcup\sigma\sqcup\xi)$.
		\item $\neg \diff_L^j \wedge \neg \diff_R^j \rightarrow (\set_B^j \wedge \set_L^j\wedge \set_R^j \wedge (\val_B^j\leftrightarrow\val_L^j )\wedge(\val_B^j\leftrightarrow\val_R^j ))$ when $u_j\in\domain(\tau)$.
		\item $\neg \diff_L^j \wedge \neg \diff_R^j \rightarrow (\set_B^j \wedge \set_L^j\wedge \neg \set_R^j \wedge (\val_B^j\leftrightarrow\val_L^j ))$ when $*/u_j\in\sigma$.
		\item $\neg \diff_L^j \wedge \neg \diff_R^j \rightarrow (\set_B^j \wedge \neg \set_L^j\wedge \set_R^j \wedge (\val_B^j\leftrightarrow\val_R^j ))$ when $*/u_j\in\xi$.
		
	\end{itemize} 
	
\end{lem}

\begin{lem}\label{lem:irm:Bdiff2}
	Suppose $L\rightarrow \con_{L}(C_1 \vee \neg x^{\tau\cup\sigma})$ and $R\rightarrow \con_{R}(C_1 \vee x^{\tau\cup\xi})$
	The following propositions are true and have short Extended Frege proofs.
	\begin{itemize}
		\item $B\wedge \diff_L^m\rightarrow L$
		\item $B\wedge \neg \diff_L^m\wedge \diff_R^m\rightarrow R$
		\item $B\wedge \diff_L^m\rightarrow \con_B(\instantiate(\xi,C_1))$
		\item $B\wedge \neg \diff_L^m\wedge \diff_R^m\rightarrow \con_B(\instantiate(\sigma,C_2))$
	\end{itemize}
	
\end{lem}

\begin{proof}[Sketch Proof]
	
	\begin{sloppypar}
		
		We break each of these statements up into constituent parts that we will prove individually and piece together through conjunction.
		
		Take $B\wedge \diff_L^m\rightarrow L$, we can prove this by showing for each index $i$ that  $(\diff^m_L\wedge(\set_B^i\rightarrow (u_i\leftrightarrow \val_B^i)))\rightarrow (\set_L^i\rightarrow (u_i\leftrightarrow \val_L^i))$. We can split up $B\wedge \neg \diff_L^m\wedge \diff_R^m\rightarrow R$ similarly.
		
		For $B\wedge \diff_L^m\rightarrow \con_B(\instantiate(\xi,C_1))$, 
		we first have to derive  $(L\rightarrow \con_{L}(C_1))\rightarrow(B\wedge L \wedge \diff^m_L \rightarrow \con_{B}(\instantiate(\xi,C_1))$. 
		We can cut out the $L$ with $B\wedge \diff_L^m\rightarrow L$.
		We will also remove $(L\rightarrow \con_{L}(C_1))$.
		By using the premise $(L\rightarrow \con_{L}(C_1\vee \neg x^{\tau\sqcup\sigma}))$ and crucially Lemma~\ref{lem:irm:tau}. $L \wedge \diff^m_L\rightarrow \neg \ann_{x,L}(\tau\sqcup \sigma)$, so $L \wedge \diff^m_L\rightarrow \neg \con_{L}(\neg x^{\tau\sqcup\sigma})$, and thus $(L \wedge \diff^m_L\rightarrow \con_{L}(C_1))$.
		
		We want to again split this up to the component parts.
		We first split by
		individual literals of $C_1$ as a proof of  $(L\rightarrow \con_{L}(l^\alpha))\rightarrow(B\wedge L \wedge \diff^m_L \rightarrow \con_{B}(\instantiate(\xi,l^\alpha))$ for each 
		literal $l^\alpha\in C_1$. We then split this between existential literal 
		$(L\rightarrow l)\rightarrow (B\wedge L \wedge \diff^m_L\rightarrow l)$ (which is a basic tautology) and universal annotation 
		$(L\rightarrow \ann_{l,B}(\alpha))\rightarrow(B\wedge L \wedge \diff^m_L \rightarrow \ann_{l,B}(\complete{\alpha}{\restr{l}{\xi}}))$. 
		
		The latter part splits further. A maximum of one of  $\neg\set_B^i$, $\set_B^i$, $\set_B^i\wedge u_i$ and $\set_B^i\wedge \neg u_i$ appears in $\ann_{l,B}(\complete{\alpha}{\restr{l}{\xi}})$, we treat $\ann_{l,B}(\complete{\alpha}{\restr{l}{\xi}})$ as a set containing these subformulas. 	We show that if formula $c_i\in \ann_{l,B}(\complete{\alpha}{\restr{l}{\xi}})$, when $c_i$ is equal to  $\neg\set_B^i$, $\set_B^i$, $\set_B^i\wedge u_i$ or  $\set_B^i\wedge \neg u_i$ then  $(L\rightarrow \ann_{l,B}(\alpha))\rightarrow(B\wedge L \wedge \diff^m_L \rightarrow c_i)$.  A similar breakdown happens for $B\wedge \neg \diff_L^m\wedge \diff_R^m\rightarrow \con_B(\instantiate(\sigma,C_2))$.

		Each of these individual cases is a constant size proof. You need to multiply for the length of each annotation (including missing values) and then do this again for each annotated literal in the clause.
		The proof size will be $O(wm)$ where $w$ is the width or number of literals in $\instantiate(\xi,C_1)\sqcup\instantiate(\sigma,C_2)$ and $m$ is the number of universal variables in the prefix.
	\end{sloppypar}
		
		We detail all cases for $L$ and $R$ in the Appendix.
\end{proof}

\begin{lem}\label{lem:irm:Bndiff}
	Suppose $L\rightarrow \con_{L}(C_1 \vee \neg x^{\tau \sqcup \sigma})$ and $R\rightarrow \con_{R}(C_2 \vee x^{\tau \sqcup \xi})$.
	The following propositions are true and have short Extended Frege proofs.
	\begin{itemize}
		\item $B\wedge \neg \diff_L^m\wedge\neg \diff_R^m \rightarrow \con_B(\instantiate(\xi, C_1)) \vee \neg x$
		\item $B\wedge \neg \diff_L^m\wedge\neg \diff_R^m \rightarrow \con_B(\instantiate(\sigma, C_2)) \vee  x$
	\end{itemize}	
\end{lem}

\begin{proof}
	\begin{sloppypar}
	Suppose that $L\rightarrow \con_L(l^\alpha)$,
	we will show that $B\wedge \neg \diff_L^m\wedge \neg \diff_R^m\rightarrow \con_R(\instantiate(\xi, l^\alpha))$.
	\end{sloppypar}
	
	We observe first that $\set^i_L\wedge \neg \diff^i_L\wedge \neg \diff^i_R\rightarrow \set^i_B \wedge (\val_B^i\leftrightarrow \val_L^i)$
	this is true in each $i:1\leq i\leq m$ by observing each case in Lemma~\ref{lem:irm:nLnR}.
	For $i>m$, $ \neg \diff^m_L\wedge \neg \diff^m_R\wedge x \rightarrow ((\set_L^i \leftrightarrow \set_B^i) \wedge (\val_B^i\leftrightarrow\val_L^i))$.
	So for all $i $, either $\neg \set^i_L$ or $\set^i_B \wedge(\val_B^i\leftrightarrow \val_L^i)$ when $\neg \diff^i_L\wedge \neg \diff^i_R$.
	
	This we can use to show $B\wedge \neg \diff^i_L\wedge \neg \diff^i_R \wedge x\rightarrow L$ by taking a conjunction of all these. We then can derive 
	$(L\rightarrow l) \rightarrow (B\wedge \neg \diff_L^m\wedge \neg \diff_R^m\wedge x \rightarrow l)$ for existential literal $l$.
	
	We still have to show that $(L\rightarrow \ann_{l,L}(\alpha)) \rightarrow (B\wedge \neg \diff_L^m\wedge \neg \diff_R^m\wedge x \rightarrow \ann_{l,B}(\complete{\alpha}{\xi}))$ for $l$'s annotation $\alpha$. 
	We can do this via cases, but we have already done all cases when $\set^i_L$ is true.
	We next show that $\neg \set^i_L\wedge \neg \diff^i_L\wedge \neg \diff^i_R\rightarrow \neg \set^i_B$ when $u_i\notin\domain(\xi)$. We can do this by simply observing the lines in Lemma~\ref{lem:irm:nLnR} when $\neg \set^i_L$ is permitted. 
	And in the final case we show $\neg \set^i_L\wedge \neg \diff^i_L\wedge \neg \diff^i_R\rightarrow \set^i_B$ when $u_i\in\domain(\xi)$.

	Remembering that $\neg \diff^m_S\rightarrow \neg \diff^i_S$ for $S\in\{L,R\}$ and $1\leq i\leq m$.
	We can now know that if $L$ satisfies  $\ann_{l,L}(\alpha)$ then $\neg \diff_{L}^m\wedge \neg \diff_{R}^m\wedge x$ will force $B$ to satisfy $\ann_{l,L}(\complete{\alpha}{\xi})$ and we can prove this in \eFrege as 
	\begin{center}$(L\rightarrow \ann_{l,L}(\alpha))\rightarrow (B\wedge \neg \diff_{L}^m\wedge \neg \diff_{R}^m\wedge x\rightarrow \ann_{l,B}(\complete{\alpha}{\xi}))$\end{center}
	
	Adding $(L\rightarrow l)\rightarrow (B\wedge \neg \diff_{L}^m\wedge \neg \diff_{R}^m\wedge x\rightarrow l)$ and for every literal $l^\alpha\in C_1$ and annotation in $C_1$ we can assemble 
	
	\begin{center}$(L\rightarrow \con_{L}(l^\alpha))\rightarrow (B\wedge \neg \diff_{L}^m\wedge \neg \diff_{R}^m\wedge x\rightarrow \con_{B}(\instantiate(\xi,l^\alpha)))$\end{center}
	
	Using $\con_{B}(\neg x^{\tau\sqcup\sigma\sqcup \xi})\rightarrow \neg x$ we can get 
	
	\begin{center}$(L\rightarrow \con_{L}(C_1\vee \neg x^{\tau \sqcup \sigma})\rightarrow (B\wedge \neg \diff_{L}^m\wedge \neg \diff_{R}^m\wedge x\rightarrow \con_{B}(\instantiate(\xi,C_1)))$\end{center}
	
	And symmetrically we can make a derivation of 
	
	\begin{center}$(R\rightarrow \con_{R}(C_2\vee  x^{\tau \sqcup \xi})\rightarrow (B\wedge \neg \diff_{L}^m\wedge \neg \diff_{R}^m\wedge  \neg x\rightarrow \con_{B}(\instantiate(\sigma,C_2)))$\end{center}
	
	The proofs here are polynomial, in this proof section we argue for each literal in the clause, and for each universal variable, but also refer to  Lemmas~\ref{lem:irm:nLnR}~and~\ref{lem:irm:impl} which have linear proofs.
	So we have cubic size proofs in the worst case or more specifically $O(wn^2)$, where $w$ is the number of literals in the derived clause $\instantiate(\sigma,C_2)\cup \instantiate(\xi,C_2)$.
\end{proof}

\begin{lem}\label{lem:irm:res}
	Suppose $L\rightarrow \con_{L}(C_1 \vee \neg x^{\tau\sqcup\sigma})$ and $R\rightarrow \con_{R}(C_1 \vee x^{\tau\sqcup\xi})$ then $B \rightarrow \con_{B}(\instantiate(\xi, C_1) \vee \instantiate(\sigma, C_2))$ has a short \eFrege proof.
\end{lem}

\begin{proof}
	$B \wedge \diff^m_L\rightarrow \con_{B}(\instantiate(\xi,C_1))$,
	$B \wedge \neg \diff^m_L\wedge \diff_R \rightarrow \con_{B}(\instantiate(\sigma,C_2))$, and
	$B \wedge \neg \diff^m_L\wedge \neg \diff_R \rightarrow \con_{B}(\instantiate(\xi, C_1) \vee \instantiate(\sigma, C_2))$
	and we can resolve on $\diff^m_L$ and $\diff^m_R$.
\end{proof}
 
\begin{thm}\label{thm:irmc}
	\eFregeRed simulates \irmc.
\end{thm}

\begin{proof}
	For each line $C$ we create a policy $S$ such that $S\rightarrow \con_S(C)$.
	
	\noindent\textbf{Axiom}
	Suppose $C\in \phi$ and it  is downloaded as $D= \instantiate(C, \tau)$ for partial annotation $\tau$.
	We construct strategy $B$ so that
	$B\rightarrow \con_B(D)$.
	
	\begin{itemize}
		\item $\set^j_B=1$ if $u_j\in \domain(\tau)$
		\item $\set^j_B=0$ if $u_j\notin \domain(\tau)$
		\item $\val^j_B=1$ if $1/u_j\in\tau$
		\item $\val^j_B=0$ if $0/u_j\in\tau$
	\end{itemize}
	
	\noindent\textbf{Instantiation}
	Suppose we have instantiation step on $C$ on a single universal variable $u_i$ using instantiation $0/u_i$. 
	So the new annotated clause is $D= \instantiate(C, 0/u)$.
	
	From the induction hypothesis $T\rightarrow \con_T(C)$ we will develop $B$ such that $B\rightarrow \con_B(D)$.
	
	\begin{itemize}
		\item $\val^i_B \leftrightarrow \val^i_T \wedge \set^i_T$ (for instantiation by $1$ we use a disjunction instead)
		\item $\set^i_B= 1$
		\item $\val^j_B \leftrightarrow \val^j_T $, 	for $j\neq i$
		\item $\set^j_B\leftrightarrow \set^j_T$, 	for $j\neq i$
	\end{itemize}
	
	\noindent\textbf{Merge}
	When merging the local strategy need not change. When literals $l^\alpha$ and $l^\beta$ are merged the strategy only has to occasionally satisfy a $\set_B^i$ variable instead of a  $\set_B^i\wedge u_i$ or $\set_B^i\wedge \neg u_i$, so the condition that needs to be satisfied is weaker. 
	
	\noindent\textbf{Resolution}
	See the definition of $B$ and Lemma~\ref{lem:irm:res}.
	
	\noindent\textbf{Contradiction}
	At the end of the proof we have $T\rightarrow \con_T(\bot)$.
	$T$ is a policy, so we turn it into a strategy $B$ by having for each $i$
	\begin{itemize}
		\item $\val^i_B \leftrightarrow (\val^i_T \wedge \set^i_T)$
		\item $\set^i_B= 1$.
	\end{itemize}
	Effectively this instantiates $\bot$ by the assignment that sets everything to $0$ and we can argue that $B\rightarrow \con_B(\bot)$ although $\con_B(\bot)$ is just the empty clause. so we have $\neg B$. 
	But $\neg B$ is just $\bigvee_{i=1}^n (u_i \oplus  \val_B^i)$.
	In \eFregeRed we can use the reduction rule (this is the first time we use the reduction rule). The proof follows from \cite{ChewSat21}.
	We show an inductive proof of $\bigvee_{i=1}^{n-k} (u_i\oplus  \val_B^i)$ for increasing $k$  eventually leaving us with the empty clause. This essentially is where we use the \Ared rule. Since we already have $\bigvee_{i=1}^{n} (u_i\oplus  \val_B^i)$ we have the base case and we only need to show the inductive step.
	
	We derive from $\bigvee_{i=1}^{n+1-k} (u_i\oplus  \val_B^i)$ both $(0\oplus  \val_B^{n-k+1})\vee\bigvee_{i=1}^{n-k} (u_i\oplus  \val_B^i)$ and $(1\oplus  \val_B^{n-k+1})\vee\bigvee_{i=1}^{n-k} (u_i\oplus  \val_B^i)$ from reduction.
	We can resolve both with the easily proved tautology $(0\leftrightarrow   \val_B^{n-k+1})\vee(1\leftrightarrow  \val_B^{n-k+1})$ which allows us to derive $\bigvee_{i=1}^{n-k} (u_i\oplus  \val_B^i)$.
	
	We continue this until we reach the empty disjunction.
\end{proof}

\begin{cor}
	\eFregeRed simulates \lqrc.
\end{cor}

\section{Extended Frege+\texorpdfstring{\Ared}{∀-Red} p-simulates \lquprc} \label{sec:lquprc}
\subsection{QCDCL Resolution Systems}

 The most basic and important CDCL
system is \emph{Q-resolution (\qrc)} by
Kleine B{\"u}ning~et~al.~\cite{KBKF95}. 
\emph{Long-distance resolution (\lqrc)} appears originally in the work of Zhang and
Malik~\cite{DBLP:conf/iccad/ZhangM02}
and was formalised into a calculus by Balabanov and Jiang~\cite{DBLP:journals/fmsd/BalabanovJ12}.
It merges complementary literals of a universal variable~$u$
into the special literal~$u^*$.
These special literals prohibit certain resolution steps.
\emph{QU-resolution (\qurc)} \cite{Gelder12} removes the restriction from \qrc that the resolved variable must be an existential variable and allows resolution of universal variables.
\emph{\lquprc}~\cite{BWJ14} extends \lqrc by allowing short and long distance resolution pivots to be universal, however, the pivot is never a merged literal $z^*$. \lquprc encapsulates \qrc, \lqrc and \qurc. 
\begin{figure}[h]
			\begin{prooftree}
				\AxiomC{}
				\RightLabel{(Axiom)}
				\UnaryInfC{$C$}
				\DisplayProof\hspace{2cm}
				\AxiomC{$D\cup\{u\}$}
				\RightLabel{($\forall$-Red)}
				\UnaryInfC{$D$}
				\DisplayProof\hspace{2cm}
				\AxiomC{$D\cup\{u^*\}$}
				\RightLabel{($\forall$-Red$^*$)}
				\UnaryInfC{$D$}
			\end{prooftree}
			\begin{minipage}{0.99\linewidth}
				$C$ is a clause in the original matrix. Literal $u$ is universal and $\lev(u)\geq\lev(l)$ for all $l\in D$.
			\end{minipage}
			\begin{prooftree}
				\AxiomC{$C_1\cup U_1\cup\{\lnot{x}\}$}
				\AxiomC{$C_2 \cup U_2\cup\{x\}$}
				\RightLabel{(Res)}
				\BinaryInfC{$C_1\cup C_2\cup U^*$}
			\end{prooftree}
			\begin{minipage}{0.99\linewidth}
				We consider two settings of the Res-rule:\\
				\textbf{SR:} If $z\in C_1$, then $\lnot{z}\notin C_2$. $U_1=U_2=U^*=\emptyset$.\\
				\textbf{LR:} If  $l_1\in C_1, l_2\in C_2$, and $\var(l_1)=\var(l_2)=z$ then $l_1=l_2\neq z^*$.
				$U_1, U_2$ contain only universal literals with $\var(U_1)=\var(U_2)$.
				$\ind(x)<\ind(u)$ for each $u\in\var(U_1)$.\\
				If $w_1\in U_1, w_2\in U_2$, $\var(w_1)=\var(w_2)=u$ then $w_1=\lnot w_2$ or $w_1=u^*$ or $w_2=u^*$. $U^*=\{u^* \mid u\in \var(U_1)\}$.\\
				
				For $b=\{1,2\}$, define $V_b=\{u^*\mid u^* \in C_b \}$. In other words $V_b$ is the subclause of $C_b \vee U_b$ of starred literals left of $x$.
			\end{minipage}
			\caption{The rules of \lquprc \cite{BWJ14}.}
			\label{fig:lquprc}
\end{figure}

\subsection{Conversion to Propositional Logic and Simulation}
\lquprc and \irmc are mutually incomparable in terms of proof strength, however both share enough similarities to get the simulation working. Once again we can use $\set^i_S$ variables to represent an $u_i^*$, and a $\neg \set_S^i$ to represent that policy $S$ chooses not to issue a value to $u_i$. 

For any set of universal variables  $Y$, let $\ann_{x,S}(Y)= \bigwedge_{u_j<x}^{u_j\notin Y} \neg\set^j_S \wedge \bigwedge_{u_j<x}^{u_j\in Y} \set^j_S$. 
Note that we do not really need to add polarities to the annotations, these are taken into account by the clause literals. 
Literals $u$ and $\bar u$ do not need to be assigned by the policy, they are now treated as a consequence of the CNF.
Because they can be resolved we treat them like existential variables in the conversion.
For universal variable $u_i$,
$\con_{S,C}(u_i)= u_i \wedge \neg \set^i_S \wedge \ann_{x,S}(\{u\mid u^*\in C\})$ and $\con_{S,C}(\neg u_i)= \neg u_i \wedge \neg \set^i_S \wedge \ann_{x,S}(\{u\mid u^*\in C\})$.
We reserve $\set^j_S$ for starred literals as they cannot be removed.
For existential literal~$x$,
$\con_{S,C}(x)= x \wedge \ann_{x,S}(\{u\mid u^*\in C\})$.
Finally, $\con_{S,C}(u^*)= \bot$, because we do not treat $u^*$ as a literal but part of the ``annotation'' to literals right of it.
Also, $u^*$ cannot be resolved but it is automatically reduced when no more literals are to the right of it. 
For clauses in \lquprc, we let $\con_S(C)=\bigvee_{l\in C}\con_{S,C}(l) $. In summary, in comparison to \irmc the conversion now includes universal variables and gives them annotations, but removes polarities from the annotations. 
Policies still remain structured as they were for \irc, with extension variables $\val^i_S$ and $\set^i_S$, where $S=\bigwedge_{i=1}^{n} \set_S^i \rightarrow(u_i \leftrightarrow \val_S^i)$.

We will once again focus on the resolution case, using the notation as given in Figure~\ref{fig:lquprc}. 

\begin{observation}
	$V_1\cap V_2=\emptyset$ by definition of resolution in \lquprc (see Figure~\ref{fig:lquprc}).
\end{observation}

We use $L$ to denote the local policy of $C_1\cup U_1\cup\{\lnot{x}\}$, $R$ to denote the local policy of $C_2\cup U_2\cup\{x\}$, and $B$ is intended to be the local policy for the resolvent $C_1\cup C_2\cup U$. Once again we use $\set^i_L, \set^i_R, \set^i_B, \val^i_L, \val^i_R, \val^i_B$ to describe the constituent parts of it. 

\subsubsection{Equivalence}
The notation for equivalence slightly changes due to the fact we are no longer working with annotations, but present starred literals. These work in much the same way. Let $b$ be in $\{1,2\}$

\noindent$\equ^{0}_{f,V_b}:=1$

\noindent $\equ^{i}_{f,V_b}:=\equ^{i-1}_{f=g}\wedge \set^i_f$ when $u_i^*\in V_b$

\noindent $\equ^{i}_{f=g}:=\equ^{i-1}_{f=g}\wedge (\neg \set^i_f)$ when $u_i^*\notin V_b$

\subsubsection{Difference}

\noindent$\diff_L^0:= 0$ and 
\noindent$\diff_R^0:= 0$

{For $u_i^*\notin C_1\cup C_2$},

\noindent $\diff_L^i:= \diff_L^{i-1} \vee (\equ_{R, V_2}^{i-1}\wedge \set^i_L)$

\noindent $\diff_R^i:= \diff_R^{i-1} \vee (\equ_{L, V_2}^{i-1}\wedge \set^i_R)$

{For $u_i^*\in C_1$},

\noindent $\diff_L^i:= \diff_L^{i-1} \vee (\equ_{R, V_2}^{i-1}\wedge \neg\set^i_L)$

\noindent $\diff_R^i:= \diff_R^{i-1} \vee (\equ_{L, V_1}^{i-1}\wedge \set^i_R)$

{For $u_i^*\in C_2$},

\noindent $\diff_L^i:= \diff_L^{i-1} \vee (\equ_{R, V_2}^{i-1}\wedge \set^i_L)$

\noindent $\diff_R^i:= \diff_R^{i-1} \vee (\equ_{L, V_1}^{i-1}\wedge \neg \set^i_R)$

\subsubsection{Policy Variables} \label{sec:pol}

{For $u_i^*\notin C_1 \cup C_2$}, $i\leq m$

\noindent$(\val^i_B, \set^i_B)=
\begin{cases} 
(\val^i_R, \set^i_R)& \text{if } \neg \diff_{L}^{i-1}\wedge(\diff_R^{i-1} \vee\neg \set_L^{i})\\ 
(\val^i_L, \set^i_L) & \text{otherwise. }
\end{cases}$

{For $u_i^*\in C_1$}, $i\leq m$

\noindent$(\val^i_B, \set^i_B)= 
\begin{cases}
(0,1)  & \text{if }\neg \diff_L^{i-1}\wedge\diff_R^{i-1}\wedge \neg \set_R^i\\
(\val^i_R, \set^i_R)& \text{if } \neg \diff_{L}^{i-1}\wedge \set_R^{i}\wedge(\diff_R^{i-1} \vee\set_L^{i})\\
(\val^i_L, \set^i_L) & \text{otherwise. }
\end{cases}$

{For $u_i^*\in C_2$}, $i\leq m$

\noindent$(\val^i_B, \set^i_B)= 
\begin{cases}
(0,1)  & \text{if }\diff_L^{i-1}\wedge \neg \set_L^i\\
(\val^i_R, \set^i_R)& \text{if } \neg \diff_{L}^{i-1}\wedge(\diff_R^{i-1} \vee\neg\set_L^{i})\\
(\val^i_L, \set^i_L) & \text{otherwise. }
\end{cases}$

For $u_i\in \domain(U)$, $i>m$

\noindent$(\val^i_B, \set^i_B)= 
\begin{cases}
(\val^i_R, \set^i_R)&\text{if } \set_R^i\wedge \neg \diff_{L}^{m}\wedge(\diff_R^{m} \vee \neg x)  \\
(0, 1)&\text{if } u_i\in U_2\text{ and }\neg \set_R^i\wedge \neg \diff_{L}^{m}\wedge(\diff_R^{m} \vee \neg x)\\
(1, 1)&\text{if } \neg u_i\in U_2\text{ and }\neg \set_R^i\wedge \neg \diff_{L}^{m}\wedge(\diff_R^{m} \vee \neg x)\\
(0, 1)&\text{if } u_i^*\in U_2\text{ and }\neg \set_R^i\wedge \neg \diff_{L}^{m}\wedge(\diff_R^{m} \vee \neg x)\\
(\val^i_L, \set^i_L) & \set_L^i\wedge \diff_{L}^{m}\vee(\neg \diff_R^{m} \wedge x) \\
(0, 1)&\text{if } u_i\in U_1\text{ and }\neg \set_L^i\wedge (\diff_{L}^{m}\vee(\neg \diff_R^{m} \wedge x))\\
(1, 1)&\text{if } \neg u_i\in U_1\text{ and }\neg \set_L^i\wedge (\diff_{L}^{m}\vee(\neg \diff_R^{m} \wedge x))\\
(0, 1)&\text{if } u_i^*\in U_1\text{ and }\neg \set_L^i\wedge( \diff_{L}^{m}\vee(\neg \diff_R^{m} \wedge x))\\
\end{cases}$

For $u_i\notin \domain(U)$, $i>m$

\noindent$(\val^i_B, \set^i_B)= 
\begin{cases}
(0,1) &\text{if } u^*\in V_2 \text{ and } \neg \set_L^i\wedge (\diff_L^m \vee (\neg \diff_R^{m} \wedge x)) \\
(\val^i_L, \set^i_L) &\text{if } u^*\in V_2 \text{ and } \set_L^i\wedge (\diff_L^m \vee (\neg \diff_R^{m} \wedge x)) \\
(\val^i_R, \set^i_R) &\text{if } u^*\in V_2 \text{ and } \neg \diff_L^m \wedge (\diff_R^{m} \vee \neg x) \\
(0,1) &\text{if } u^*\in V_1 \text{ and } \neg \set_R^i\wedge (\neg \diff_L^m \wedge (\diff_R^{m} \vee \neg x))\\
(\val^i_R, \set^i_R) &\text{if } u^*\in V_1 \text{ and } \set_R^i\wedge (\neg \diff_L^m \wedge (\diff_R^{m} \vee \neg x) )\\
(\val^i_L, \set^i_L) &\text{if } u^*\in V_1 \text{ and } \diff_L^m \vee (\neg \diff_R^{m} \wedge x) \\
(\val^i_R, \set^i_R)& \text{if } u^*\notin V_1\cup V_2 \text{ and } \neg \diff_{L}^{m}\wedge(\diff_R^{m} \vee \neg x)\\
(\val^i_L, \set^i_L) & \text{if } u^*\notin V_1\cup V_2 \text{ and } \diff_L^m \vee (\neg \diff_R^{m} \wedge x)
\end{cases}$

One may notice there are a larger number of cases for $i>m$ than in previous sections, this is because $u$ and $\neg u$ become $u^*$ and end up joining the annotation and policies. It should also be pointed out that there are cases resulting $(0,1)$ than to $(1,1)$ this a is simply matter of using $0$ as the default value when some set has to be made.

\begin{lem}\label{lem:lquprc:chain}
	For $0<j\leq m$ the following propositions have short derivations in Extended Frege:
	\begin{itemize}
		\item $\diff_L^j \rightarrow \bigvee_{i=1}^j \diff_L^i \wedge \neg \diff_L^{i-1}$
		\item $\diff_R^j \rightarrow \bigvee_{i=1}^j \diff_R^i \wedge \neg \diff_R^{i-1}$
		\item $\neg \equ_{L,V_1}^j \rightarrow \bigvee_{i=1}^j \neg \equ_{L,V_1}^i \wedge  \equ_{L,V_1}^{i-1}$
		\item $\neg \equ_{R, V_2}^j \rightarrow \bigvee_{i=1}^j \neg \equ_{R, V_2}^i \wedge  \equ_{R, V_2}^{i-1}$
		
	\end{itemize}
\end{lem}

\begin{proof}
	The proof of Lemma~\ref{lem:chain} still works despite the modifications to definition.
\end{proof}

\begin{lem}\label{lem:lquprc:impl}
	For $0\leq i \leq j\leq m$ the following propositions that describe the monotonicity of $\diff$ and $\equ$ have short derivations in Extended Frege:
	\begin{itemize}
		\item $\diff_L^i \rightarrow \diff_L^j$
		\item $\diff_R^i \rightarrow \diff_R^j$
		\item $\neg \equ_{L, V_1}^i \rightarrow \neg \equ_{L, V_1}^j$
		\item $\neg \equ_{R, V_2}^i \rightarrow \neg \equ_{R, V_2}^j$
	\end{itemize}
\end{lem}

\begin{proof} The proofs of Lemma~\ref{lem:impl} still work despite the modifications to definition.
\end{proof}

\begin{lem}\label{lem:lquprc:tau}
	For any $0\leq i\leq m$
	the following propositions are true and have short Extended Frege proofs.
	
	\begin{itemize}
		\item $\diff_L^i \rightarrow \neg \ann_{x,L}(V_1)$
		\item $\diff_R^i \rightarrow \neg \ann_{x,R}(V_2)$
	\end{itemize}
\end{lem}

\begin{proof}
	If $u_i \notin V_1$ then $\diff_{L}^i\wedge \neg \diff_{L}^{i-1}\rightarrow \set^i_{L}$ but $\ann_{x,L}(V_1)$ insists on $\neg\set^i_{L}$.
	
	If $u_i \in V_1$ then $\diff_{L}^i\wedge \neg \diff_{L}^{i-1}\rightarrow \neg \set^i_{L}$ but $\ann_{x,L}(V_1)$ insists on $\set^i_{L}$.
	This is done similarly for $R$.
\end{proof}

\begin{lem}\label{lem:lquprc:nLnR}
	For any $0\leq j\leq m$
	the following propositions are true and have a short Extended Frege proof.
	\begin{itemize}
		\item $\neg \diff_L^j \wedge \neg \diff_R^j \rightarrow \equ^j_{L,V_1}$ 
		\item $\neg \diff_L^j \wedge \neg \diff_R^j \rightarrow \equ^j_{R,V_2}$
		\item $\neg \diff_L^j \wedge \neg \diff_R^j \rightarrow (\neg \set_B^j \wedge \neg \set_L^j\wedge \neg \set_R^j)$ when $u_j^*\notin C_1\vee C_2$.
		\item $\neg \diff_L^j \wedge \neg \diff_R^j \rightarrow (\set_B^j \wedge \set_L^j\wedge \neg \set_R^j \wedge (\val_B^j\leftrightarrow\val_L^j ))$ when $u_j^*\in C_1$.
		\item $\neg \diff_L^j \wedge \neg \diff_R^j \rightarrow (\set_B^j \wedge \neg \set_L^j\wedge \set_R^j \wedge (\val_B^j\leftrightarrow\val_R^j ))$ when $u_j^*\in C_2$.
		
	\end{itemize} 
	
\end{lem}

\begin{proof}
	\begin{sloppypar}
		We show that $\neg \equ^{j+1}_{L, V_1}\rightarrow \neg \equ^j_{L, V_1} \vee \neg \equ^j_{R, V_2} \vee \diff_L^{j+1}$ and $\neg \equ^{j+1}_{R, V_2}\rightarrow \neg \equ^j_{R, V_2} \vee \neg \equ^j_{L, V_2} \vee \diff_R^{j+1}$.
		Suppose $u_{j+1}^*\in V_1$ then $\neg \equ^{j+1}_{L, V_1}\wedge \equ^j_{L, V_1}\rightarrow \set^{j+1}_L$ and $\set^{j+1}_L\rightarrow \neg \equ^j_{R,V_2}\vee \diff_L^{j+1}$, so we have $\neg \equ^{j+1}_{L, V_1}\wedge \rightarrow \neg \equ^j_{R, V_2} \vee \neg \equ^j_{L, V_1} \vee \diff_R^{j+1}$. This is symmetric for $R$ and for $u_{j+1}^*\notin V_1$.
	\end{sloppypar}

	\noindent \textbf{Induction Hypothesis (on $j$):} $(\neg \equ_{L,V_1}^j \vee \neg \equ_{R, V_2}^j) \rightarrow (\diff_L^j\vee \diff_R^j)$.
	
	\noindent\textbf{Base Case ($j=1$):} $\neg \equ_{L,V_1}^1\wedge \equ_{L,V_1}^{0} \rightarrow \diff_L^1 \vee \neg \equ^{0}_{R, V_2}$
	, and
	$\neg \equ_{R, V_2}^1\wedge \equ_{R, V_2}^{0} \rightarrow \diff_R^1 \vee \neg \equ^{0}_{L,V_1}$.
	
	However since $\equ^{0}_{L,V_1}$ and $\equ^{0}_{R, V_2}$ are both true it simplifies to $\neg \equ_{L,V_1}^1\rightarrow \diff_L^1$ 
	and 
	$\neg \equ_{R, V_2}^1\rightarrow \diff_R^1$
	which can be combined to get 
	$(\neg \equ_{L,V_1}^1 \vee \neg \equ_{R, V_2}^1) \rightarrow (\diff_L^1\vee \diff_R^1)$.
	
	\noindent\textbf{Inductive Step ($j+1$):} 
	\begin{sloppypar}
		The Induction Hypothesis 
		$(\neg \equ_{L, V_1}^j \vee \neg \equ_{R, V_2}^j) \rightarrow (\diff_L^j\vee \diff_R^j)$
		can be weakened to 
		$(\neg \equ_{L, V_1}^j \vee \neg \equ_{R, V_2}^j) \rightarrow (\diff_L^{j+1}\vee \diff_R^{j+1})$,
		using
		$\diff_L^j\rightarrow \diff_L^{j+1}$
		and
		$\diff_R^j\rightarrow \diff_R^{j+1}$.
		
		We now need to replace $(\neg \equ_{L, V_1}^j \vee \neg \equ_{R, V_2}^j) $ with $(\neg \equ_{L, V_1}^{j+1} \vee \neg \equ_{R, V_2}^{j+1}) $. 
		Suppose $u_{j+1}\in V_1$, note that $\neg \equ_{L, V_1}^{j+1} \rightarrow \neg \equ_{L, V_1}^{j} \vee \neg \set^{j+1}_L$. 
		$\neg \set^{j+1}_L\wedge \equ_{R, V_2}^j\rightarrow \diff_R^{j+1}$.		
		We show that $\neg \equ^{j+1}_{L, V_1}\rightarrow \neg \equ^j_{L, V_1} \vee \neg \equ^j_{R, V_2} \vee \diff_L^{j+1}$ and $\neg \equ^{j+1}_{R, V_2}\rightarrow \neg \equ^j_{R, V_2} \vee \neg \equ^j_{L, V_2} \vee \diff_R^{j+1}$.
		
		Suppose $u_{j+1}^*\in V_1$ then $\neg \equ^{j+1}_{L, V_1}\wedge \equ^j_{L, V_1}\rightarrow \set^{j+1}_L$ and $\set^{j+1}_L\rightarrow \neg \equ^j_{R,V_2}\vee \diff_L^{j+1}$, so we have $\neg \equ^{j+1}_{L, V_1}\wedge \rightarrow \neg \equ^j_{R, V_2} \vee \neg \equ^j_{L, V_1} \vee \diff_R^{j+1}$. This is symmetric for $R$ and for $u_{j+1}^*\notin V_1$.
		
		We can use these formulas to show $\neg \equ^{j+1}_{L, V_1}\wedge \neg \equ^{j+1}_{R, V_2} \rightarrow \neg \equ^j_{L, V_1} \vee \neg \equ^j_{R, V_2} \vee \diff_L^{j+1}\vee \diff_R^{j+1}$ and we can simplify this to $\neg \equ^{j+1}_{L, V_1}\wedge \neg \equ^{j+1}_{R, V_2} \rightarrow \diff_L^{j+1}\vee \diff_R^{j+1}$.

		$\neg \diff^j_L\wedge \neg\diff^j_R \rightarrow \equ^j_{L, V_1}$,
		$\neg\diff^j_L\wedge \neg\diff^j_R \rightarrow \equ^j_{R, V_2}$ are corollaries of this. 
		$\neg \diff^j_L\wedge \neg\diff^j_R$ means $\neg \diff^{j-1}_L\wedge \neg\diff^{j-1}_R$.
		$u_j^*\in C_1$ implies $u_j^*\notin C_2$, so $\set^j_L$ and $\neg \set^j_R$, and that makes ($\val_B^j, \set_B^j$)=($\val_L^j, \set_L^j$).
		
		$u_j^*\in C_2$ implies $u_j^*\notin C_1$ so $\neg \set^j_L$ and $\set^j_R$, and that makes ($\val_B^j, \set_B^j$)=($\val_R^j, \set_R^j$).
			\end{sloppypar}
		$u_j^*\notin C_1\cup C_2$ implies $\neg \set_{L}^j$ and $\neg \set_{L}^j$, therefore ($\val_B^j, \set_B^j$)=($\val_L^j, \set_L^j$).
\end{proof}

\begin{lem}\label{lem:lquprc:Bdiff}
	The following propositions are true and have short Extended Frege proofs, given $(L\rightarrow \con_L(C_1\cup  U_1\vee \neg x))$ and $(R\rightarrow \con_R(C_2\cup  U_2\vee x))$
	\begin{itemize}
		\item $B\wedge \diff_L^m\rightarrow L$
		\item $B\wedge \neg \diff_L^m\wedge \diff_R^m\rightarrow R$
		\item $B\wedge \diff_L^m\rightarrow \con_B(C_1\vee V_2 \vee U^*)$
		\item $B\wedge \neg \diff_L^m\wedge \diff_R^m\rightarrow \con_B(C_2\vee V_1 \vee U^*)$
	\end{itemize}
	
\end{lem}

\begin{proof}[Sketch Proof]
	\begin{sloppypar}
	We break $B\wedge \diff_L^m\rightarrow L$ into individual parts $\set_B^i \rightarrow (u_i\leftrightarrow \val^i_B)\wedge \diff_L^m  \rightarrow (\set_L^i\rightarrow (u_i\leftrightarrow \val_L^i))$ which we join by conjunction. We can do similarly for 
$B\wedge \neg \diff_L^m\wedge \diff_R^m\rightarrow R$.

For $B\wedge \diff_L^m\rightarrow \con_B(C_1\vee V_2 \vee U^*)$ we
first derive
$(L\rightarrow \con_{L}(C_1\cup U_1 \vee \neg x)) \rightarrow (B\wedge L \wedge \diff_L^m \rightarrow \con_B(C_1\vee V_2 \vee U^*) ) $, you can cut out $L$ using $B\wedge \diff_L^m\rightarrow L$. Removing $(L\rightarrow \con_{L}(C_1\cup U_1 \vee \neg x))$, uses the premise $(L\rightarrow \con_L(C_1\cup  U_1\vee \neg x))$. 

To derive 	$(L\rightarrow \con_{L}(C_1\cup U_1 \vee \neg x)) \rightarrow (B\wedge L \wedge \diff_L^m \rightarrow \con_B(C_1\vee V_2 \vee U^*) ) $ we
break this by non-starred literals $l\in C_1\cup U_1$ so we will show that $(L\rightarrow \con_{L,C_1\cup U_1 \vee \neg x}(l))\rightarrow(B\wedge \diff_L^m\rightarrow \con_{B,V_2\cup C_1\cup U}(l)) $.		$\diff^m_{L}\rightarrow \neg \ann_{x, L}(V_1)$ is used to remove the $x$ literal.

For $p\in\{1,2\}$ let $W_p=\{u^* \mid u^*\in U_p \}$.
For each $i$, either $\set_B^i$ or $\neg \set_B^i$ appears in $\ann_{l,B}(V_1\cup V_2 \cup U^*)$, so we treat $\ann_{l,B}(V_1\cup V_2 \cup U^*)$ as a set containing these subformulas. We show that if $c_i\in \ann_{l,B}(V_1\cup V_2\cup U^*)$ when $c_i= \set_B^i$ or $c_i= \neg\set_B^i$ then $L\rightarrow \ann_{l,L}(V_1\cup W_1)\rightarrow B \wedge \diff_{L}^m \rightarrow c_i$ and we also have $(L\rightarrow l)\rightarrow (B \wedge \diff_{L}^m \rightarrow l)$.

For existential $l$, we can put these all together to get $(L\rightarrow \con_{L,C_1\cup U_1}(l))\rightarrow(B\wedge L \wedge \diff_L^m\rightarrow \con_{B,V_2\cup C_1\cup U}(l)) $. 
For universal literals $u_k$ we also need to show $\neg \set^k_B$ is preserved when $u_k$ is not merged.
For universal literals $u_k$ that are merged $\con_{B,V_2\cup C_1\cup U^*}(u_k*) = \bot$ so we show that the strategy for $B$ causes a contradiction between $B$ and $L\rightarrow u_k$.
We do similarly for $B\wedge \neg \diff_L^m\wedge \diff_R^m\rightarrow \con_B(C_2\vee V_1 \vee U^*)$.
		\end{sloppypar}
We detail all cases for $L$ and $R$ in the Appendix.
	\end{proof}

\begin{lem}\label{lem:lquprc:Bndiff}
	The following propositions are true and have short Extended Frege proofs, given $(L\rightarrow \con_L(C_1\cup  U_1\vee \neg x))$ and $(R\rightarrow \con_R(C_2\cup  U_2\vee x))$.
	\begin{itemize}
		\item $B\wedge \neg \diff_L^m\wedge\neg \diff_R^m \rightarrow \con_B(C_1\vee V_2 \vee U^*) \vee \neg x$
		\item $B\wedge \neg \diff_L^m\wedge\neg \diff_R^m \rightarrow \con_B(C_2\vee V_1 \vee U^* )  \vee  x$
	\end{itemize}
	
\end{lem}

\begin{proof}
	For indices $1\leq i\leq m$, but since $\neg \diff_{L}^m \rightarrow \neg \diff_{L}^i$ and $\neg \diff_{R}^m \rightarrow \neg \diff_{R}^i$, Lemma~\ref{lem:nLnR} can be used to show that $B\wedge \diff_{L}^m \wedge \diff_{R}^m$ leads to $\set_B^i$ taking the a value consistent with both $V_1\cup V_2$, if $L$ was consistent with $V_1$ and $R$ was consistent with $V_2$.
	
	For $i>m$,  $\neg \diff_{R}^m \wedge \neg \diff_{L}^m $ will make the policy $B$ pick between the left and right policy based on  $x$. However in either case $\set_B^i$ will be forced to update based on the new annotations.  
\end{proof}

\begin{lem}\label{lem:lquprc:res}
	Suppose, there are policies $L$ and $R$ such that $L\rightarrow \con_{L}(C_1 \vee \neg x \vee U_1)$ and $R\rightarrow \con_{R}(C_2 \vee x \vee U_2)$ then there is a policy $B$ such that $B \rightarrow \con_{B}(C_1 \vee C_2\vee U^*)$ can be obtained in a short \eFrege proof, where $C_1$, $C_2$, $U_1$, $U_2$ and $U^*$ follow the same definitions as in Figure~\ref{fig:lquprc}.
\end{lem}
\begin{proof}

	From Lemmas~\ref{lem:lquprc:Bndiff}~and~\ref{lem:lquprc:Bdiff}, $\con_B(C_1\vee V_2 \vee U^*)$ and $\con_B(C_2\vee V_1 \vee U^*)$ can be weakened to $\con_B(C_1\vee C_2 \vee U^*)$. 
	These can all be combined over the different possibilities to give  $B \rightarrow \con_{B}(C_1 \vee C_2\vee U^*)$.
\end{proof}

\begin{thm}
	\eFregeRed simulates \lquprc.
\end{thm}

\begin{proof}
	
	We inductively build a policy $S$ such that $S\rightarrow \con_S(C)$ can be proved from $\phi$ using \eFrege, for every clause $C$ in an \lquprc proof. At the end we have the empty clause and a strategy and we can use reduction to remove the strategy and obtain the empty clause as in Theorems~\ref{thm:mrc}~and~\ref{thm:irc}.
	
	\noindent\textbf{Axiom}
	Each Axiom is treated with the empty policy.

	\noindent\textbf{Reduction ($u_i$ or $\neg u_i$)}
	If the clause contains literal $u_i$, we know that $T\rightarrow \con_T(C\vee u_i)$. We define $S$ so that
	
	$(\val^j_S, \set^j_S)= \begin{cases}(\val^j_T, \set^j_T) & j\neq i \end{cases}$
	
	$(\val^i_S, \set^i_S)= \begin{cases}(\val^i_T, \set^i_T) & \text{if } \set^i_T \vee\con_T(C) \text{ is satisfied, }\\ (0,1) & \text{ otherwise. }\end{cases}$
	
	We need to show that 
	$S\rightarrow \con_S(C)$.	
	Note that $\con_T(C\vee u_i)=\con_T(C)\vee \con_{T,C}(u_i)$. 
	Therefore $T\rightarrow \con_T(C)$ or $T\rightarrow \neg \set_T^{i}\wedge u_i$.
	If $\set_T^{i}$ is true or $\con_T(C)$ then $T\rightarrow \con_T(C)$ is true and as $S$ will match $T$, $S\rightarrow \con_S(C)$.
	Suppose $\set_T^{i}$ and $\con_T(C)$ are both false. If~$S$ is true, then $u_i$ is false by construction. Moreover, since $S$ agrees with $T$ on every variable except $u_i$, and $T$ does not set $u_i$, $T$ must be true as well. But since $\con_T(C)$ is false, we must have $T \rightarrow \neg \set_T^{i} \wedge u_i$. In particular, $u_i$ must be true, a contradiction. We conclude that the implication $S \rightarrow \con_S(C)$ holds in this case.
	
	\noindent\textbf{Reduction ($u_i^*$)}
	If $T\rightarrow \con_T(C\vee u_i^*)$ and we reduce $u_i^*$ we need to define the strategy $S$ so that  $S\rightarrow \con_S(C)$. Since $u_i^*$ is the rightmost literal in the clause $\con_T(C\vee u_i^*)= \con_T(C)$ so we define $S$ the same way as $T$.
	
	\noindent\textbf{Resolution} See Lemma~\ref{lem:lquprc:res}.
	
	\noindent\textbf{Contradiction} Just as in \irc we have to give a complete assignment to the missing values in the policy. We then have simply the negation of the strategy for which we can apply our same technique to reduce to the empty clause.
\end{proof}

\section{Conclusion}
Our work reconciles many different QBF proof techniques under the single system \eFregeRed.
Although \eFregeRed itself is likely not a good system for efficient proof checking, our results have implications for other systems that are more promising in this regard, such as \qrat, which inherits these simulations. In particular, \qrat's simulation of \ecalculus is upgraded to a simulation of \irmc, and we do not even require the extended universal reduction rule.
Existing \qrat checkers can be used to verify converted \eFregeRed proofs.
Further, extended QU-resolution is polynomially equivalent to \eFregeRed~\cite{ChewSat21}, and has previously been proposed as a system for unified QBF proof checking~\cite{Jus07}.
Since our simulations split off propositional inference from a standardised reduction part at the end, another option is to use (highly efficient) propositional proof checkers instead.
Our simulations use many extension variables that are known to negatively impact the checking time of existing tools such as \textsf{DRAT}-trim, but one may hope that they can be refined to become more efficient in this regard.

There are other proof systems, particularly ones using dependency schemes, such as \qdrc and \lqdrc that have strategy extraction~\cite{PeitlSS19a}. Local strategy extraction and ultimately a simulation by \eFregeRed seem likely for these systems, whether it can be proved directly or by generalising the simulation results from this paper.

\bibliographystyle{alphaurl}
\bibliography{compl}

\newcommand{\etalchar}[1]{$^{#1}$}
\begin{thebibliography}{JKMC16}

\bibitem[AB09]{AB09}
Sanjeev Arora and Boaz Barak.
\newblock {\em Computational Complexity: A Modern Approach}.
\newblock Cambridge University Press, USA, 1st edition, 2009.

\bibitem[BBCP20]{BBCP20}
Olaf Beyersdorff, Ilario Bonacina, Leroy Chew, and Jan Pich.
\newblock Frege systems for quantified {B}oolean logic.
\newblock {\em J. ACM}, 67(2), April 2020.
\newblock \href {https://doi.org/10.1145/3381881} {\path{doi:10.1145/3381881}}.

\bibitem[BBM18]{BBM18}
Olaf Beyersdorff, Joshua Blinkhorn, and M.~Mahajan.
\newblock Building strategies into {QBF} proofs.
\newblock In {\em Electron. Colloquium Comput. Complex.}, 2018.
\newblock \href {https://doi.org/10.1007/s10817-020-09560-1}
  {\path{doi:10.1007/s10817-020-09560-1}}.

\bibitem[BCJ14]{BCJ14}
Olaf Beyersdorff, Leroy Chew, and Mikol{\'a}\v{s} Janota.
\newblock On unification of {QBF} resolution-based calculi.
\newblock In {\em MFCS, II}, pages 81--93, 2014.
\newblock \href {https://doi.org/10.1007/978-3-662-44465-8_8}
  {\path{doi:10.1007/978-3-662-44465-8_8}}.

\bibitem[BCJ19]{BCJ19}
Olaf Beyersdorff, Leroy Chew, and Mikol{\'{a}}\v{s} Janota.
\newblock New resolution-based {QBF} calculi and their proof complexity.
\newblock {\em {ACM} Trans. Comput. Theory}, 11(4):26:1--26:42, 2019.
\newblock \href {https://doi.org/10.1145/3352155} {\path{doi:10.1145/3352155}}.

\bibitem[BCMS18]{BCMS16CPjournal}
Olaf Beyersdorff, Leroy Chew, Meena Mahajan, and Anil Shukla.
\newblock Understanding cutting planes for qbfs.
\newblock {\em Information and Computation}, 262:141 -- 161, 2018.
\newblock \href {https://doi.org/10.1016/j.ic.2018.08.002}
  {\path{doi:10.1016/j.ic.2018.08.002}}.

\bibitem[Bey09]{Bey09}
Olaf Beyersdorff.
\newblock On the correspondence between arithmetic theories and propositional
  proof systems – a survey.
\newblock {\em Mathematical Logic Quarterly}, 55(2):116--137, 2009.
\newblock \href {https://doi.org/10.1002/malq.200710069}
  {\path{doi:10.1002/malq.200710069}}.

\bibitem[BJ12]{DBLP:journals/fmsd/BalabanovJ12}
Valeriy Balabanov and Jie-Hong~R. Jiang.
\newblock Unified {QBF} certification and its applications.
\newblock {\em Formal Methods in System Design}, 41(1):45--65, 2012.
\newblock \href {https://doi.org/10.1007/s10703-012-0152-6}
  {\path{doi:10.1007/s10703-012-0152-6}}.

\bibitem[BLB10]{BrummayerLB10}
Robert Brummayer, Florian Lonsing, and Armin Biere.
\newblock Automated testing and debugging of {SAT} and {QBF} solvers.
\newblock In {\em {SAT 2010}}, volume 6175 of {\em Lecture Notes in Computer
  Science}, pages 44--57. Springer, 2010.
\newblock \href {https://doi.org/10.1007/978-3-642-14186-7_6}
  {\path{doi:10.1007/978-3-642-14186-7_6}}.

\bibitem[BWJ14]{BWJ14}
Valeriy Balabanov, Magdalena Widl, and Jie-Hong~R. Jiang.
\newblock {QBF} resolution systems and their proof complexities.
\newblock In {\em SAT 2014}, pages 154--169, 2014.
\newblock \href {https://doi.org/10.1007/978-3-319-09284-3_12}
  {\path{doi:10.1007/978-3-319-09284-3_12}}.

\bibitem[CH22]{ChewHSAT22}
Leroy Chew and Marijn J.~H. Heule.
\newblock {Relating Existing Powerful Proof Systems for QBF}.
\newblock In Kuldeep~S. Meel and Ofer Strichman, editors, {\em 25th
  International Conference on Theory and Applications of Satisfiability Testing
  (SAT 2022)}, volume 236 of {\em Leibniz International Proceedings in
  Informatics (LIPIcs)}, pages 10:1--10:22, 2022.
\newblock \href {https://doi.org/10.4230/LIPIcs.SAT.2022.10}
  {\path{doi:10.4230/LIPIcs.SAT.2022.10}}.

\bibitem[Che16]{Che16}
Hubie Chen.
\newblock Proof complexity modulo the polynomial hierarchy: Understanding
  alternation as a source of hardness.
\newblock In {\em ICALP 2016}, pages 94:1--94:14, 2016.
\newblock \href {https://doi.org/10.4230/LIPICS.ICALP.2016.94}
  {\path{doi:10.4230/LIPICS.ICALP.2016.94}}.

\bibitem[Che21]{ChewSat21}
Leroy Chew.
\newblock Hardness and optimality in {QBF} proof systems modulo {NP}.
\newblock In {\em SAT 2021}, pages 98--115, Cham, 2021. Springer.
\newblock \href {https://doi.org/10.1007/978-3-030-80223-3_8}
  {\path{doi:10.1007/978-3-030-80223-3_8}}.

\bibitem[CHJ{\etalchar{+}}17]{Cruz-FilipeHHKS17}
Lu{\'{\i}}s Cruz{-}Filipe, Marijn J.~H. Heule, Warren A.~Hunt Jr., Matt
  Kaufmann, and Peter Schneider{-}Kamp.
\newblock Efficient certified {RAT} verification.
\newblock In {\em {CADE} 2017}, volume 10395 of {\em Lecture Notes in Computer
  Science}, pages 220--236. Springer, 2017.
\newblock \href {https://doi.org/10.1007/978-3-319-63046-5_14}
  {\path{doi:10.1007/978-3-319-63046-5_14}}.

\bibitem[CR79]{CR79}
Stephen~A. Cook and Robert~A. Reckhow.
\newblock The relative efficiency of propositional proof systems.
\newblock {\em The Journal of Symbolic Logic}, 44(1):36--50, 1979.
\newblock \href {https://doi.org/10.2307/2273702} {\path{doi:10.2307/2273702}}.

\bibitem[CS21]{ChedeS21}
Sravanthi Chede and Anil Shukla.
\newblock Does {QRAT} simulate {IR}-calc? {QRAT} simulation algorithm for
  $\forall${E}xp+{R}es cannot be lifted to {IR}-calc.
\newblock {\em Electron. Colloquium Comput. Complex.}, page 104, 2021.

\bibitem[ELW13]{BEW13}
Uwe Egly, Florian Lonsing, and Magdalena Widl.
\newblock Long-distance resolution: Proof generation and strategy extraction in
  search-based {QBF} solving.
\newblock In Kenneth~L. McMillan, Aart Middeldorp, and Andrei Voronkov,
  editors, {\em LPAR 2013}, pages 291--308. Springer, 2013.
\newblock \href {https://doi.org/10.1007/978-3-642-45221-5_21}
  {\path{doi:10.1007/978-3-642-45221-5_21}}.

\bibitem[GVB11]{Goultiaeva-ijcai11}
Alexandra Goultiaeva, Allen {Van Gelder}, and Fahiem Bacchus.
\newblock A uniform approach for generating proofs and strategies for both true
  and false {QBF} formulas.
\newblock In Toby Walsh, editor, {\em IJCAI 2011}, pages 546--553. IJCAI/AAAI,
  2011.
\newblock \href {https://doi.org/10.5591/978-1-57735-516-8/IJCAI11-099}
  {\path{doi:10.5591/978-1-57735-516-8/IJCAI11-099}}.

\bibitem[HSB17]{HeuleSB17}
Marijn J.~H. Heule, Martina Seidl, and Armin Biere.
\newblock Solution validation and extraction for {QBF} preprocessing.
\newblock {\em J. Autom. Reason.}, 58(1):97--125, 2017.
\newblock \href {https://doi.org/10.1007/s10817-016-9390-4}
  {\path{doi:10.1007/s10817-016-9390-4}}.

\bibitem[JBS{\etalchar{+}}07]{Jus07}
Toni Jussila, Armin Biere, Carsten Sinz, Daniel Kr{\"o}ning, and Christoph~M.
  Wintersteiger.
\newblock A first step towards a unified proof checker for {QBF}.
\newblock In {\em {SAT} 2007}, pages 201--214, 2007.
\newblock \href {https://doi.org/10.1007/978-3-540-72788-0_21}
  {\path{doi:10.1007/978-3-540-72788-0_21}}.

\bibitem[JKMC16]{JanotaKMC16}
Mikol{\'{a}}\v{s} Janota, William Klieber, Jo{\~{a}}o Marques{-}Silva, and
  Edmund~M. Clarke.
\newblock Solving {QBF} with counterexample guided refinement.
\newblock {\em Artif. Intell.}, 234:1--25, 2016.
\newblock \href {https://doi.org/10.1007/978-3-642-31612-8_10}
  {\path{doi:10.1007/978-3-642-31612-8_10}}.

\bibitem[JM15a]{JM15}
Mikol{\'{a}}\v{s} Janota and Joao Marques{-}Silva.
\newblock Expansion-based {QBF} solving versus {Q}-resolution.
\newblock {\em Theor. Comput. Sci.}, 577:25--42, 2015.
\newblock \href {https://doi.org/10.1016/j.tcs.2015.01.048}
  {\path{doi:10.1016/j.tcs.2015.01.048}}.

\bibitem[JM15b]{JanotaM15}
Mikol{\'{a}}\v{s} Janota and Jo{\~{a}}o Marques{-}Silva.
\newblock Solving {QBF} by clause selection.
\newblock In {\em {IJCAI} 2015}, pages 325--331. {AAAI} Press, 2015.
\newblock \href {https://doi.org/10.5555/2832249.2832294}
  {\path{doi:10.5555/2832249.2832294}}.

\bibitem[KHS17]{KHS17}
Benjamin Kiesl, Marijn J.~H. Heule, and Martina Seidl.
\newblock A little blocked literal goes a long way.
\newblock In {\em {SAT} 2017}, volume 10491 of {\em Lecture Notes in Computer
  Science}, pages 281--297. Springer, 2017.
\newblock \href {https://doi.org/10.1007/978-3-319-66263-3_18}
  {\path{doi:10.1007/978-3-319-66263-3_18}}.

\bibitem[KKF95]{KBKF95}
Hans {Kleine B{\"u}ning}, Marek Karpinski, and Andreas Fl{\"o}gel.
\newblock Resolution for quantified {Boolean} formulas.
\newblock {\em Inf. Comput.}, 117(1):12--18, 1995.
\newblock \href {https://doi.org/10.1006/inco.1995.1025}
  {\path{doi:10.1006/inco.1995.1025}}.

\bibitem[KP90]{KP90}
Jan Kraj\'{\i}\v{c}ek and Pavel Pudl\'{a}k.
\newblock Quantified propositional calculi and fragments of bounded arithmetic.
\newblock {\em Zeitschrift f\"{u}r mathematische Logik und Grundlagen der
  Mathematik}, 36:29--46, 1990.
\newblock \href {https://doi.org/10.1002/malq.19900360106}
  {\path{doi:10.1002/malq.19900360106}}.

\bibitem[Kra95]{Kra95}
Jan Kraj\'{\i}\v{c}ek.
\newblock {\em Bounded Arithmetic, Propositional Logic, and Complexity Theory},
  volume~60 of {\em Encyclopedia of Mathematics and Its Applications}.
\newblock Cambridge University Press, Cambridge, 1995.
\newblock \href {https://doi.org/10.1017/CBO9780511529948}
  {\path{doi:10.1017/CBO9780511529948}}.

\bibitem[Kra19]{krajivcek2019proof}
Jan Kraj{\'\i}{\v{c}}ek.
\newblock {\em Proof complexity}, volume 170.
\newblock Cambridge University Press, 2019.
\newblock \href {https://doi.org/10.1017/bsl.2023.13}
  {\path{doi:10.1017/bsl.2023.13}}.

\bibitem[KS19]{KS19}
Benjamin Kiesl and Martina Seidl.
\newblock {QRAT} polynomially simulates {$\forall$Exp+Res}.
\newblock In {\em {SAT} 2019}, volume 11628 of {\em Lecture Notes in Computer
  Science}, pages 193--202. Springer, 2019.
\newblock \href {https://doi.org/0.1007/978-3-030-24258-9_13}
  {\path{doi:0.1007/978-3-030-24258-9_13}}.

\bibitem[LB10]{LonsingB10}
Florian Lonsing and Armin Biere.
\newblock Integrating dependency schemes in search-based {QBF} solvers.
\newblock In {\em {SAT} 2010}, volume 6175 of {\em Lecture Notes in Computer
  Science}, pages 158--171. Springer, 2010.
\newblock \href {https://doi.org/10.1007/978-3-642-14186-7_14}
  {\path{doi:10.1007/978-3-642-14186-7_14}}.

\bibitem[PSS19a]{PeitlSS19}
Tom{\'{a}}s Peitl, Friedrich Slivovsky, and Stefan Szeider.
\newblock Dependency learning for {QBF}.
\newblock {\em J. Artif. Intell. Res.}, 65:180--208, 2019.
\newblock \href {https://doi.org/10.1613/jair.1.11529}
  {\path{doi:10.1613/jair.1.11529}}.

\bibitem[PSS19b]{PeitlSS19a}
Tom{\'{a}}s Peitl, Friedrich Slivovsky, and Stefan Szeider.
\newblock Long-distance {Q-Resolution} with dependency schemes.
\newblock {\em J. Autom. Reason.}, 63(1):127--155, 2019.
\newblock \href {https://doi.org/10.1007/s10817-018-9467-3}
  {\path{doi:10.1007/s10817-018-9467-3}}.

\bibitem[Rec76]{Rec76}
Robert~A. Reckhow.
\newblock {\em On the lengths of proofs in the propositional calculus}.
\newblock PhD thesis, University of Toronto, 1976.
\newblock \href {https://doi.org/10.1145/800119.803893}
  {\path{doi:10.1145/800119.803893}}.

\bibitem[RT15]{rabe2015caqe}
Markus~N Rabe and Leander Tentrup.
\newblock {CAQE}: A certifying {QBF} solver.
\newblock In {\em {FMCAD} 2015}, pages 136--143. FMCAD Inc, 2015.
\newblock \href {https://doi.org/10.1109/FMCAD.2015.7542263}
  {\path{doi:10.1109/FMCAD.2015.7542263}}.

\bibitem[SBPS19]{ShuklaBPS19}
Ankit Shukla, Armin Biere, Luca Pulina, and Martina Seidl.
\newblock A survey on applications of quantified {B}oolean formulas.
\newblock In {\em {ICTAI} 2019}, pages 78--84. {IEEE}, 2019.
\newblock \href {https://doi.org/10.1109/ICTAI.2019.00020}
  {\path{doi:10.1109/ICTAI.2019.00020}}.

\bibitem[SG18]{SG18}
Martin Suda and Bernhard Gleiss.
\newblock Local soundness for {QBF} calculi.
\newblock In {\em {SAT} 2018}, volume 10929 of {\em Lecture Notes in Computer
  Science}, pages 217--234. Springer, 2018.
\newblock \href {https://doi.org/10.1007/978-3-319-94144-8_14}
  {\path{doi:10.1007/978-3-319-94144-8_14}}.

\bibitem[SM73]{SM73}
L.~J. Stockmeyer and A.~R. Meyer.
\newblock Word problems requiring exponential time.
\newblock {\em Proc.\ 5th {ACM} Symposium on Theory of Computing}, pages 1--9,
  1973.
\newblock \href {https://doi.org/10.1145/800125.804029}
  {\path{doi:10.1145/800125.804029}}.

\bibitem[SSWZ20]{SSWZ20}
Matthias Schlaipfer, Friedrich Slivovsky, Georg Weissenbacher, and Florian
  Zuleger.
\newblock Multi-linear strategy extraction for {QBF} expansion proofs via local
  soundness.
\newblock In {\em {SAT} 2020}, volume 12178 of {\em Lecture Notes in Computer
  Science}, pages 429--446. Springer, 2020.
\newblock \href {https://doi.org/10.1007/978-3-030-51825-7_30}
  {\path{doi:10.1007/978-3-030-51825-7_30}}.

\bibitem[VG12]{Gelder12}
Allen Van~Gelder.
\newblock Contributions to the theory of practical quantified {B}oolean formula
  solving.
\newblock In {\em Principles and Practice of Constraint Programming}, pages
  647--663. Springer, 2012.
\newblock \href {https://doi.org/10.1007/978-3-642-33558-7_47}
  {\path{doi:10.1007/978-3-642-33558-7_47}}.

\bibitem[WHJ14]{WetzlerHH14}
Nathan Wetzler, Marijn Heule, and Warren A.~Hunt Jr.
\newblock {DRAT}-trim: Efficient checking and trimming using expressive clausal
  proofs.
\newblock In {\em {SAT} 2014}, volume 8561 of {\em Lecture Notes in Computer
  Science}, pages 422--429. Springer, 2014.
\newblock \href {https://doi.org/10.1007/978-3-319-09284-3_31}
  {\path{doi:10.1007/978-3-319-09284-3_31}}.

\bibitem[ZM02]{DBLP:conf/iccad/ZhangM02}
Lintao Zhang and Sharad Malik.
\newblock Conflict driven learning in a quantified {Boolean} satisfiability
  solver.
\newblock In {\em {ICCAD} 2002}, pages 442--449, 2002.
\newblock \href {https://doi.org/10.1109/ICCAD.2002.1167570}
  {\path{doi:10.1109/ICCAD.2002.1167570}}.

\end{thebibliography}

\newpage
\appendix
\section{}

\subsection{Local Strategy Extraction for Simulation of \irmc}\label{app:irmc}

\subsubsection{Policy Variables}

For $u_i\notin \domain(\tau\sqcup \sigma \sqcup \xi)$, $u_i<x$,

\noindent$(\val^i_B, \set^i_B)=
\begin{cases} 
(\val^i_R, \set^i_R)& \text{if } \neg \diff_{L}^{i-1}\wedge(\diff_R^{i-1} \vee\neg \set_L^{i})\\ 
(\val^i_L, \set^i_L) & \text{otherwise. }
\end{cases}$

{For $u_i\in\domain(\tau)$},

\noindent$(\val^i_B, \set^i_B)= 
\begin{cases} 
(\val^i_R, \set^i_R)& \text{if } \neg \diff_{L}^{i-1}\wedge(\diff_R^{i-1} \vee(\set_L^{i} \wedge (\val_L^{i}\leftrightarrow \val_\tau^{i})))\\ 
(\val^i_L, \set^i_L) & \text{otherwise. }\end{cases}$

{For $*/u_i\in\sigma$},

\noindent$(\val^i_B, \set^i_B)= 
\begin{cases}
(0,1)  & \text{if }\neg \diff_L^{i-1}\wedge\diff_R^{i-1}\wedge \neg \set_R^i\\
(\val^i_R, \set^i_R)& \text{if } \neg \diff_{L}^{i-1}\wedge \set_R^{i}\wedge(\diff_R^{i-1} \vee\set_L^{i})\\
(\val^i_L, \set^i_L) & \text{otherwise. }
\end{cases}$

{For $*/u_i\in\xi$},

\noindent$(\val^i_B, \set^i_B)= 
\begin{cases}
(0,1)  & \text{if }\diff_L^{i-1}\wedge \neg \set_L^i\\
(\val^i_R, \set^i_R)& \text{if } \neg \diff_{L}^{i-1}\wedge(\diff_R^{i-1} \vee\neg\set_L^{i})\\
(\val^i_L, \set^i_L) & \text{otherwise. }
\end{cases}$

For $u_i>x$,

\noindent$(\val^i_B, \set^i_B)= 
\begin{cases}
(\val^i_R, \set^i_R)& \text{if } \neg \diff_{L}^{m}\wedge(\diff_R^{m} \vee \neg x)\\
(\val^i_L, \set^i_L) & \text{otherwise. }
\end{cases}$
\renewcommand{\thesection}{\arabic{section}}
\setcounter{section}{5}
\setcounter{thm}{5}

\begin{lem}
	Suppose $L\rightarrow \con_{L}(C_1 \vee \neg x^{\tau\cup\sigma})$ and $R\rightarrow \con_{L}(C_1 \vee x^{\tau\cup\xi})$.
	The following propositions are true and have short Extended Frege proofs.
	\begin{itemize}
		\item $B\wedge \diff_L^m\rightarrow L$
		\item $B\wedge \neg \diff_L^m\wedge \diff_R^m\rightarrow R$
		\item $B\wedge \diff_L^m\rightarrow \con_B(\instantiate(\xi,C_1))$
		\item $B\wedge \neg \diff_L^m\wedge \diff_R^m\rightarrow \con_B(\instantiate(\sigma,C_2))$
	\end{itemize}
	
\end{lem}

\begin{proof}
	
	\begin{sloppypar}
		
		We break each of these statements up into constituent parts that we will prove individually and piece together through conjunction.
		
		Take $B\wedge \diff_L^m\rightarrow L$, we can prove this by showing for each index $i$ that  $(\diff^m_L\wedge(\set_B^i\rightarrow (u_i\leftrightarrow \val_B^i)))\rightarrow (\set_L^i\rightarrow (u_i\leftrightarrow \val_L^i))$. We can split up $B\wedge \neg \diff_L^m\wedge \diff_R^m\rightarrow R$ similarly.
		
		For $B\wedge \diff_L^m\rightarrow \con_B(\instantiate(\xi,C_1))$, 
		we first have to derive  $(L\rightarrow \con_{L}(C_1))\rightarrow(B\wedge L \wedge \diff^m_L \rightarrow \con_{B}(\instantiate(\xi,C_1))$. 
		We can cut out the $L$ with $B\wedge \diff_L^m\rightarrow L$.
		We will also remove $(L\rightarrow \con_{L}(C_1))$. By using the premise $(L\rightarrow \con_{L}(C_1\vee \neg x^{\tau\sqcup\sigma}))$ and crucially Lemma~\ref{lem:irm:tau}. $L \wedge \diff^m_L\rightarrow \neg \ann_{x,L}(\tau\sqcup \sigma)$, so $L \wedge \diff^m_L\rightarrow \neg \con_{L}(\neg x^{\tau\sqcup\sigma})$, and thus $(L \wedge \diff^m_L\rightarrow \con_{L}(C_1))$.
		
		We want to again split this up to the component parts.
		We first split by
		individual literals of $C_1$ as a proof of  $(L\rightarrow \con_{L}(l^\alpha))\rightarrow(B\wedge L \wedge \diff^m_L \rightarrow \con_{B}(\instantiate(\xi,l^\alpha))$ for each 
		literal $l^\alpha\in C_1$. We then split this between existential literal 
		$(L\rightarrow l)\rightarrow (B\wedge L \wedge \diff^m_L\rightarrow l)$ (which is a basic tautology) and universal annotation 
		$(L\rightarrow \ann_{l,B}(\alpha))\rightarrow(B\wedge L \wedge \diff^m_L \rightarrow \ann_{l,B}(\complete{\alpha}{\restr{l}{\xi}}))$. 
		
		The latter part splits further. A maximum of one of  $\neg\set_B^i$, $\set_B^i$, $\set_B^i\wedge u_i$ and $\set_B^i\wedge \neg u_i$ appears in $\ann_{l,B}(\complete{\alpha}{\restr{l}{\xi}})$, we treat $\ann_{l,B}(\complete{\alpha}{\restr{l}{\xi}})$ as a set containing these subformulas. 	We show that if formula $c_i\in \ann_{l,B}(\complete{\alpha}{\restr{l}{\xi}})$, when $c_i$ is equal to  $\neg\set_B^i$, $\set_B^i$, $\set_B^i\wedge u_i$ or  $\set_B^i\wedge \neg u_i$ then  $(L\rightarrow \ann_{l,B}(\alpha))\rightarrow(B\wedge L \wedge \diff^m_L \rightarrow c_i)$.  A similar breakdown happens for $B\wedge \neg \diff_L^m\wedge \diff_R^m\rightarrow \con_B(\instantiate(\sigma,C_2))$.

Each of these individual cases is a constant size proof. You need to multiply for the length of each annotation (including missing values) and then do this again for each annotated literal in the clause.
The proof size will be $O(wm)$ where $w$ is the width or number of literals in $\instantiate(\xi,C_1)\sqcup\instantiate(\sigma,C_2)$ and $m$ is the number of universal variables in the prefix.
		
		Each $(L\rightarrow \ann_{l,B}(\alpha))\rightarrow(B\wedge L \wedge \diff^m_L \rightarrow c_i)$ fall into one of many cases. There are multiple ``axes'' of cases, the first being by index $i$,
		in the cases $i>m$, $j<i\leq m$, $i=j$, $i<j$. $j$ here refers to the index such that $\diff^j_L\wedge \neg \diff^{j-1}_L\wedge\neg \diff^{j-1}_R $ which we know exists via Lemmas~\ref{lem:irm:chain}~and~\ref{lem:irm:rel}. Lemma~\ref{lem:irm:chain} is crucial to stringing these together.
		The next axis of cases then by choice of annotation in $\complete{\alpha}{\restr{l}{\xi}}$. 
		Further we have to consider sub-cases of these that affect the policy variables, as detailed in Section~\ref{sec:pol}.

		We detail the cases below:

		\noindent\textbf{Suppose $i> m$. }
		
		$\diff_L^i$ refutes $\neg \diff_L^m \wedge (\diff_R^{m}\vee \neg x^i_L )$ so whenever $\diff^m_L$ is true, $(\val_B^i,\set_B^i)=(\val_L^i,\set_L^i)$, therefore $(\set_B^i\rightarrow (u_i\leftrightarrow \val_B^i))\rightarrow (\set_L^i\rightarrow (u_i\leftrightarrow \val_L^i))$.
		
		If $\neg\set_B^i\in \ann_{l,B}(\complete{\alpha}{\restr{l}{\xi}})$, then $u_i\notin\domain(\complete{\alpha}{\restr{l}{\xi}})$.
		We know $u_i\notin\domain(\alpha)$ otherwise it would be in $\domain(\complete{\alpha}{\restr{l}{\xi}})$.
		Therefore $\neg\set_L^i$ is in $\ann_{l,L}(\alpha)$. 
		And so if $L\rightarrow \ann_{l,L}(\alpha)$ then $L\rightarrow \neg\set_L^i$,
		therefore $B\wedge L \wedge \diff^m_L \rightarrow \neg\set_B^i$.
		We now look at all the cases of $c_i\in \ann_{l,B}(\complete{\alpha}{\restr{l}{\xi}})$ and show they can be satisfied with our strategy in $B$:
		
		If $\set_B^i\in \ann_{l,B}(\complete{\alpha}{\restr{l}{\xi}})$,
		then $u_i\in\domain(\complete{\alpha}{\restr{l}{\xi}})$
		$u_i\notin\domain(\xi)$ because $\domain(\xi)$ only extends up to $m$
		hence $u_i\in \domain({\alpha})$
		and $\set_L^i\in \ann_{l,L}(\alpha)$.
		And so if $L\rightarrow \ann_{l,L}(\alpha)$ then $L\rightarrow \set_L^i$, therefore $B\wedge L \wedge \diff^m_L \rightarrow \set_B^i$.
		
		If $\set_B^i\wedge u_i\in \ann_{l,B}(\complete{\alpha}{\restr{l}{\xi}})$
		then $u_i\in\complete{\alpha}{\restr{l}{\xi}}$.
		We know $u_i\notin\domain(\xi)$ because $\domain(\xi)$ only extends up to $m$. 
		Hence $u_i\in\alpha$
		and $\set_L^i\wedge u_i\in \ann_{l,L}(\alpha)$.
		And so if $L\rightarrow \ann_{l,L}(\alpha)$ then $L\rightarrow \set_L^i\wedge u_i$, therefore $B\wedge L \wedge \diff^m_L \rightarrow \set_B^i\wedge u_i$.
		
		If $\set_B^i\wedge \neg u_i\in \ann_{l,B}(\complete{\alpha}{\restr{l}{\xi}})$
		then $\neg u_i\in\complete{\alpha}{\restr{l}{\xi}}$,
		$u_i\notin\domain(\xi)$ because $\domain(\xi)$ only extends up to $m$. 
		Hence $\neg u_i\in\alpha$
		and $\set_L^i\wedge \neg u_i\in \ann_{l,L}(\alpha)$
		And so if $L\rightarrow \ann_{l,L}(\alpha)$ then $L\rightarrow \set_L^i\wedge \neg\val_L^i$, therefore $B\wedge L \wedge \diff^m_L \rightarrow \set_B^i\wedge \neg u_i$.

		\noindent\textbf{Suppose $j < i\leq m$.}
		
		We know $\diff^j_L\rightarrow \diff^{i-1}_L$ from Lemma~\ref{lem:irm:impl}, we will use that to get that when $\diff^j_L\wedge \set^i_L$ then $(\val^i_B, \set^i_B)=(\val^i_L,\set^i_L)$ which allows us to then show $(\set_B^i\rightarrow (u_i\leftrightarrow \val_B^i))\rightarrow (\set_L^i\rightarrow (u_i\leftrightarrow \val_L^i))$.
		When $\diff^{i-1}_L$ for $u_i\notin \domain(\xi)$ we refute $\neg \diff_{L}^{i-1}\wedge(\diff_R^{i-1} \vee\neg \set_L^{i})$,
		$\neg \diff_{L}^{i-1}\wedge(\diff_R^{i-1} \vee(\set_L^{i} \wedge (\val_L^{i}\leftrightarrow \val_\tau^{i})))$ , 
		$\neg \diff_L^{i-1}\wedge\diff_R^{i-1}\wedge \neg \set_R^i$ and
		$\neg \diff_{L}^{i-1}\wedge \set_R^{i}\wedge(\diff_R^{i-1} \vee\set_L^{i})$.
		When $\diff^{i-1}_L$ for $u_i\in \domain(\xi)$ when $\set^i_{L}$ is true we refute $\diff^{i-1}_{L}\wedge \neg \set^i_L$ and $\neg \diff^{i-1}_{L}\wedge (\diff^{i-1}_R\vee\neg \set^i_L)$.

		if $\neg\set_B^i\in \ann_{l,B}(\complete{\alpha}{\restr{l}{\xi}})$ 
		then  $u_i\notin 	\domain(\complete{\alpha}{\restr{l}{\xi}}) $,  	
		also $u_i\notin\domain(\alpha)$ and  $u_i\notin\domain(\xi)$
		so $\neg\set_L^i\in \ann_{l,L}(\alpha)$
		And so if $L\rightarrow \ann_{l,L}(\alpha)$ then $L\rightarrow \neg\set_L^i$
		when $\diff^{i-1}_L$ and $u_i\notin\domain(\xi)$, $(\val_B^i,\set_B^i)=(\val_L^i,\set_L^i)$ and so $B \wedge L \wedge \diff^{i-1}_L\wedge \set_B^i$.

		If $\set_B^i\in \ann_{l,B}(\complete{\alpha}{\restr{l}{\xi}})$ then 
		$*/u_i	\in \domain(\complete{\alpha}{\restr{l}{\xi}}) $
		so either $*/u_i\in \alpha$ or 
		$u_i\notin \domain(\alpha)$ and $*/u_i\in \xi$.
		If $*/u_i\in \alpha$ then
		$\set_L^i\in \ann_{l,L}(\alpha)$ and 
		$L\rightarrow \set_L^i$
		so when $\diff^{i-1}_L\wedge \set_L^i$ no matter which domain $u_i$ is in
		$(\val_B^i,\set_B^i)=(\val_L^i,\set_L^i)$
		$B \wedge L \wedge \diff^{i-1}_L\wedge \set_B^i$.
		If $u_i\notin \domain(\alpha)$ and $*/u_i\in \xi$.
		$\neg \set_L^i\in \ann_{l,L}(\alpha)$
		so
		$L\rightarrow \neg \set_L^i$.
		$u_i\in \domain(\xi)$ means that when $\diff^{i-1}_{L}$ and $\neg \set_L^i$
		$(\val_B^i,\set_B^i)=(0,1)$
		so
		$B \wedge L \wedge \diff^{i-1}_L\wedge \set_B^i$
		
		If  $\set_B^i\wedge u_i\in \ann_{l,B}(\complete{\alpha}{\restr{l}{\xi}})$
		then	$1/u_i	\in(\complete{\alpha}{\restr{l}{\xi}}) $
		and it can only be that $1/u_i\in \alpha$ as $\xi$ can only add $*/u_i$.
		So $\set_L^i\wedge u_i\in \ann_{l,L}(\alpha)$ and
		$L\rightarrow \set_L^i$.
		so when $\diff^{i-1}_L\wedge \set_L^i$ no matter which domain $u_i$ is in
		$(\val_B^i,\set_B^i)=(\val_L^i,\set_L^i)$.
		$B \wedge L \wedge \diff^{i-1}_L\wedge \set_B^i\wedge u_i$.
		
		Likewise, 	if  $\set_B^i\wedge \neg u_i\in \con_{l,B}(\complete{\alpha}{\restr{l}{\xi}})$
		then	$0/u_i	\in(\complete{\alpha}{\restr{l}{\xi}}) $
		and it can only be that $0/u_i\in \alpha$ as $\xi$ can only add $*/u_i$.
		So $\set_L^i\wedge u_i\in \ann_{l,L}(\alpha)$ and
		$L\rightarrow \set_L^i$.
		So when $\diff^{i-1}_L\wedge \set_L^i$ no matter which domain $u_i$ is in
		$(\val_B^i,\set_B^i)=(\val_L^i,\set_L^i)$.
		$B \wedge L \wedge \diff^{i-1}_L\wedge \set_B^i\wedge \neg u_i$.

		\noindent \textbf{Suppose $i=j$.}
		
		$\neg \diff^{j-1}_L$ by definition of $j$. $\neg \diff^{j-1}_R$ is also true as $\diff^{j-1}_R$ contradicts $\equ_{R=\tau\vee\xi}^{j-1}$ which is necessary for $\diff^{j}_L$. 
		With $\neg \diff^{j-1}_R$, ($\val_B^j, \set_B^j$) can only be defined as ($\val_R^j, \set_R^j$) in a small selection of circumstances.
		That is when:
		$\neg\set_L^j$ and $u_i\notin\domain(\tau\sqcup\sigma\sqcup\xi)$,
		$\set_L^j\wedge \val_L^j$ and $1/u_j\in\tau$,
		$\set_L^j\wedge \neg\val_L^j$ and $0/u_j\in\tau$,
		$\set_L^j\wedge\set_R^j $ and $*/u_j\in\sigma$,
		$\neg\set_L^j $ and $*/u_j\in\xi$.
		All but the latter contradict $\diff_{L}^j\wedge \diff_{L}^{j-1}$, but we can ignore whenever $\set_L^j$ is false.
		So $\diff_{L}^j\wedge\neg \diff^{j-1}_L\wedge\set_L^j \rightarrow \set_B^j $ this means that $(\set_B^j\rightarrow(u_i\leftrightarrow \val_B^j))\rightarrow (\set_L^j\rightarrow(u_i\leftrightarrow \val_L^j))$.

		If $\neg\set_B^j\in \ann_{l,B}(\complete{\alpha}{\restr{l}{\xi}})$ then 
		$u_j\notin\domain(\complete{\alpha}{\restr{l}{\xi}})$ and so
		$u_j\notin\domain(\alpha)$
		$u_j\notin\domain(\xi)$. 
		So $\neg\set_L^j\in \ann_{l,L}(\alpha)$
		and $L\rightarrow \neg\set_L^j$.
		Since $\diff^{j}_L$ is true then it can only be that $u_j\in\domain(\tau)$ or $u_j\in\domain(\sigma)$.
		If $u_j\in\domain(\tau)$ then 
		$\neg \diff^{j-1}_L\wedge ( \diff^{j-1}_R \vee (\set^j_L\wedge (\val^j_L \leftrightarrow \val^j_\tau )))$ is contradicted so $(\val_B^j,\set_B^j)=(\val_L^j,\set_L^j)$
		and $B\wedge L \wedge \diff^{j}_L \wedge \neg \diff^{j-1}_L \rightarrow \neg\set_B^i$.
		If $u_j\in\domain(\sigma)$ then 
		$\neg \diff^{j-1}_L\wedge \diff_{R}^{j-1}\wedge \neg \set^j_R$ and $\neg \diff^{j-1}_L\wedge \set^j_R\wedge (\diff_{R}^{j-1}\vee\set^j_L)$ are contradicted so $(\val_B^j,\set_B^j)=(\val_L^j,\set_L^j)$
		and $B\wedge L \wedge \diff^{j}_L \wedge \neg \diff^{j-1}_L \rightarrow \neg\set_B^i$.
		If $u_j\notin\domain(\tau\sqcup\sigma\sqcup\xi)	$
		$\diff^{j}_L$ is false in this case so we can ignore it.
		$(\val_B^j,\set_B^j)=(\val_L^j,\set_L^j)$ means that 
		$B \wedge L \wedge \diff^{j}_L \wedge \neg \diff^{j-1}_L\wedge \neg \set_B^j$.

		If $\set_B^j\in \ann_{l,B}(\complete{\alpha}{\restr{l}{\xi}})$,
		$u_j\in\domain(\complete{\alpha}{\restr{l}{\xi}})$.
		Either $*/u_j\in \alpha$ or  $u_j\not\in \domain(\alpha)$ and $*/u_j\in \xi$.
		If $*/u_j\in \alpha$,
		then $\set^j_{L}\in \ann_{l,L}(\alpha)$
		and $L\rightarrow \set_L^j$.
		If $u_j\notin\domain(\tau\sqcup\sigma\sqcup\xi)$,
		$\neg \diff_{L}^{j-1}\wedge(\diff_R^{j-1} \vee\neg \set_L^{j})$ is falsified so
		$(\val_B^j,\set_B^j)=(\val_L^j,\set_L^j)$
		and $B\wedge L \wedge \diff^{j}_L \wedge \neg \diff^{j-1}_L \rightarrow \set_B^i$.
		If $u_j\in\domain(\tau)$,
		$\diff_{L}^{j}\wedge\neg \diff_{L}^{j-1}\wedge \set_L^j$ means that $\val_L^j\oplus\val_\tau^j$. As a result
		$\neg \diff_{L}^{j-1}\wedge(\diff_R^{j-1} \vee(\set_L^{j} \wedge (\val_L^{j}\leftrightarrow \val_\tau^{j})))$ is falsified so
		$(\val_B^j,\set_B^j)=(\val_L^j,\set_L^j)$
		and $B\wedge L \wedge \diff^{j}_L \wedge \neg \diff^{j-1}_L \rightarrow \set_B^i$.
		If $u_j\in\domain(\sigma)$,
		$\set_L^{j}$ contradicts $\diff^j_{L}\wedge\neg \diff^{j-1}_{L}$, so this scenario does not occur.
		If $u_j\in\domain(\xi)$
		$\diff_L^{j-1}\wedge\neg \set_L^j$ is falsified by $\neg \diff_{L}^{j-1} $.
		$\neg \diff_L^{j-1}\wedge(\diff_R^{j-1}\vee \neg\set_L^{j})$ is falsified by 	$\set_L^{j}$  so
		$(\val_B^j,\set_B^j)=(\val_L^j,\set_L^j)$ and $B\wedge L \wedge \diff^{j}_L \wedge \neg \diff^{j-1}_L \rightarrow \set_B^i$.
		If $u_j\not\in \domain(\alpha)$ and $*/u_j\in \xi$
		then $\neg\set^j_{L}\in \ann_{l,L}(\alpha)$
		and $L\rightarrow \neg \set_L^j$.
		However this cannot happen when $\diff^j_{L}\wedge\neg \diff^{j-1}_{L}$.

		If $\set_B^j\wedge \val_B^j\in \ann_{l,B}(\complete{\alpha}{\restr{l}{\xi}})$,
		$1/u_j\in(\complete{\alpha}{\restr{l}{\xi}})$.
		As instantiate is only done by $*$ then	$1/u_j\in\alpha$.
		So it follows $\set_L^j\wedge \val_L^j\in \ann_{l,L}(\alpha)$.
		If $u_j\notin\domain(\tau\sqcup\sigma\sqcup\xi)$,
		$\neg \diff_{L}^{j-1}\wedge(\diff_R^{j-1} \vee\neg \set_L^{j})$ is falsified so
		$(\val_B^j,\set_B^j)=(\val_L^j,\set_L^j)$
		and $B\wedge L \wedge \diff^{j}_L \wedge \neg \diff^{j-1}_L \rightarrow \set_B^i\wedge\val_B^i$ .
		If $u_j\in\domain(\tau)$,
		$\diff_{L}^{j}\wedge\neg \diff_{L}^{j-1}\wedge \set_L^j\wedge \val_L^j$ means that $\neg\val_\tau^j$ and so.
		$\neg \diff_{L}^{j-1}\wedge(\diff_R^{j-1} \vee(\set_L^{j} \wedge (\val_L^{j}\leftrightarrow \val_\tau^{j})))$ is falsified so
		$(\val_B^j,\set_B^j)=(\val_L^j,\set_L^j)$
		and $B\wedge L \wedge \diff^{j}_L \wedge \neg \diff^{j-1}_L \rightarrow \set_B^i\wedge\val_B^i$.
		If $u_j\in\domain(\sigma)$,
		$\set_L^{j}$ contradicts $\diff^j_{L}\wedge\neg \diff^{j-1}_{L}$, so this scenario does not occur.
		If $u_j\in\domain(\xi)$,
		$\diff_L^{j-1}\wedge\neg \set_L^j$ is falsified by $\neg \diff_{L}^{j-1} $.
		$\neg \diff_L^{j-1}\wedge(\diff_R^{j-1}\vee \neg\set_L^{j})$ is falsified by 	$\set_L^{j}$  so
		$(\val_B^j,\set_B^j)=(\val_L^j,\set_L^j)$ and $B\wedge L \wedge \diff^{j}_L \wedge \neg \diff^{j-1}_L \rightarrow \set_B^i\wedge u_i$.
		
		If $\set_B^j\wedge \neg \val_B^j\in \ann_{l,B}(\complete{\alpha}{\restr{l}{\xi}})$,
		$0/u_j\in(\complete{\alpha}{\restr{l}{\xi}})$.
		As instantiate is only done by $*$ then	$0/u_j\in\alpha$.
		So it follows $\set_L^j\wedge \neg \val_L^j\in \ann_{l,L}(\alpha)$.
		If $u_j\notin\domain(\tau\sqcup\sigma\sqcup\xi)$,
		$\neg \diff_{L}^{j-1}\wedge(\diff_R^{j-1} \vee\neg \set_L^{j})$ is falsified so
		$(\val_B^j,\set_B^j)=(\val_L^j,\set_L^j)$
		and $B\wedge L \wedge \diff^{j}_L \wedge \neg \diff^{j-1}_L \rightarrow \set_B^i\wedge\neg \val_B^i$ .
		If $u_j\in\domain(\tau)$,
		$\diff_{L}^{j}\wedge\neg \diff_{L}^{j-1}\wedge \set_L^j\wedge \neg\val_L^j$ means that $\val_\tau^j$ and so
		$\neg \diff_{L}^{j-1}\wedge(\diff_R^{j-1} \vee(\set_L^{j} \wedge (\val_L^{j}\leftrightarrow \val_\tau^{j})))$ is falsified so
		$(\val_B^j,\set_B^j)=(\val_L^j,\set_L^j)$
		and $B\wedge L \wedge \diff^{j}_L \wedge \neg \diff^{j-1}_L \rightarrow \set_B^i\wedge\neg \val_B^i$.
		If $u_j\in\domain(\sigma)$,
		$\set_L^{j}$ contradicts $\diff^j_{L}\wedge\neg \diff^{j-1}_{L}$, so this scenario does not occur.
		If $u_j\in\domain(\xi)$
		$\diff_L^{j-1}\wedge\neg \set_L^j$ is falsified by $\neg \diff_{L}^{j-1} $.
		$\neg \diff_L^{j-1}\wedge(\diff_R^{j-1}\vee \neg\set_L^{j})$ is falsified by 	$\set_L^{j}$  so
		$(\val_B^j,\set_B^j)=(\val_L^j,\set_L^j)$ and $B\wedge L \wedge \diff^{j}_L \wedge \neg \diff^{j-1}_L \rightarrow \set_B^i\wedge\neg u_i$.
		
		\noindent \textbf{Suppose $i<j$.}
		
		In this case $\neg \diff^{i}_L,\neg \diff^{i-1}_L, \neg \diff^{i}_R, \neg \diff^{i-1}_R$ are all true. We can see from Lemma~\ref{lem:irm:nLnR} that $\set_L^i\rightarrow\set_B^i$ in all cases. We observe all the cases when $\set_L^i$ is true and $\val_B^i$ is not defined as $\val_L^i$.
		For $u_i\in\domain(\tau)$, this happens if $(\val_L^i\leftrightarrow\val_\tau^i )$, but then also $(\val_R^i\leftrightarrow\val_\tau^i )$ if $\neg \diff_{R}^i$ so $\val_B^i=\val_R^i=\val_L^i$.
		For $u_i\in\domain(\sigma)$, if $\neg \diff_{L}^{i-1}\wedge \set_R^{i}\wedge(\diff_R^{i-1} \vee\set_L^{i})$ then $\val_B^i=\val_R^i $, but this cannot happen if $\neg \diff^{i}_R\wedge\neg \diff^{i-1}_R $. So in all cases of  $\neg \diff^{i}_L,\neg \diff^{i-1}_L, \neg \diff^{i}_R, \neg \diff^{i-1}_R, \set_L^i$ we have $\val_B^i=\val_L^i$.
		This means that $\set_B^i\rightarrow(u_i\leftrightarrow \val_B^i)\rightarrow \set_L^i\rightarrow(u_i\leftrightarrow \val_L^i)$.
		
		If $\neg\set_B^i\in \ann_{l,B}(\complete{\alpha}{\restr{l}{\xi}})$ then 
		$u_j\notin\domain(\complete{\alpha}{\restr{l}{\xi}})$ and so
		$u_j\notin\domain(\alpha)$
		$u_j\notin\domain(\xi)$. 
		So $\neg\set_L^i\in \ann_{l,L}(\alpha)$
		and $L\rightarrow \neg\set_L^i$.
		$\neg \diff^{i}_L,\neg \diff^{i-1}_L$ means that 
		$u_i\notin\domain(\tau\sqcup\sigma\sqcup\xi) $
		From Lemma~\ref{lem:irm:nLnR} we know $\neg \diff^{i}_L\wedge\neg \diff^{i}_R\rightarrow\neg \set^i_B$. So $B\wedge L \wedge \diff^{j}_L \wedge \neg \diff^{i}_L \rightarrow \neg\set_B^i$.
		
		If $\set_B^i\in \ann_{l,B}(\complete{\alpha}{\restr{l}{\xi}})$.
		Either $*/u_i\in \alpha$ or  $u_i\not\in \domain(\alpha)$ and $*/u_i\in \xi$
		If $*/u_i\in \alpha$,
		then $\set^i_{L}\in \ann_{l,L}(\alpha)$
		and $L\rightarrow \set_L^i$.
		By Lemma~\ref{lem:irm:nLnR}, $u_i$ must be in $\domain(\tau)$
		or $\domain(\sigma)$.
		In either case $\set_B^i$ is true. 
		So $B\wedge L \wedge \diff^{j}_L \wedge \neg \diff^{i}_L \rightarrow \set_B^i$.
		If $u_i\not\in \domain(\alpha)$ and $*/u_i\in \xi$,
		then $\neg \set^i_{L}\in \ann_{l,L}(\alpha)$
		and $L\rightarrow \neg \set_L^i$.
		By Lemma~\ref{lem:irm:nLnR}, $\set_B^i$ is true.
		So $B\wedge L \wedge \diff^{j}_L \wedge \neg \diff^{i}_L \rightarrow \set_B^i$.
		
		If $\set_B^i\wedge\val_B^i\in \ann_{l,B}(\complete{\alpha}{\restr{l}{\xi}})$
		then $1/u_i\in \complete{\alpha}{\restr{l}{\xi}}$, so it must be that 
		$1/u_i\in \alpha$. And so
		$\set^i_{L}\wedge \val_L^i\in \ann_{y,L}(\alpha)$
		By Lemma~\ref{lem:irm:nLnR}, $u_i$ must be in $\domain(\tau)$
		or $\domain(\sigma)$.
		In either case ($\val_B^i,\set_B^i$)=($\val_L^i,\set_L^i$).
		So $B\wedge L \wedge \diff^{j}_L \wedge \neg \diff^{i}_L \rightarrow \set_B^i\wedge \val_L^i$, because $L\rightarrow \set_L^i\wedge\val_L^i$.
		
		Likewise, if $\set_B^i\wedge\neg\val_B^i\in \ann_{l,B}(\complete{\alpha}{\restr{l}{\xi}})$
		then $0/u_i\in \complete{\alpha}{\restr{l}{\xi}}$, so it must be that 
		$0/u_i\in \alpha$. And so
		$\set^i_{L}\wedge \neg\val_L^i\in \ann_{l,L}(\alpha)$.
		By Lemma~\ref{lem:irm:nLnR}, $u_i$ must be in $\domain(\tau)$
		or $\domain(\sigma)$.
		In either case ($\val_B^i,\set_B^i$)=($\val_L^i,\set_L^i$).
		So $B\wedge L \wedge \diff^{j}_L \wedge \neg \diff^{i}_L \rightarrow \set_B^i\wedge \neg\val_L^i$, because $L\rightarrow \set_L^i\wedge\neg \val_L^i$.
		
		\noindent\textbf{In all $\diff_L^m$ cases} $(\set_B^i\rightarrow(u_i\leftrightarrow \val_B^i))\rightarrow (\set_L^i\rightarrow(u_i\leftrightarrow \val_L^i))$ so then 
		$B\wedge \diff_L^m \rightarrow L$.
		We also have $B\wedge \diff_L^m\wedge L \rightarrow \ann_{l,B}(\complete{\alpha}{\restr{l}{\xi}})$. 
		We also get $B\wedge \diff_L^m\wedge L \rightarrow \con_{B}(l)$, from  $L \rightarrow l$ so we can get $B\wedge \diff_L^m\wedge L \rightarrow \ann_{B}(\instantiate(\xi,l^{\alpha}))$, this can be put in a disjunction $B\wedge \diff_L^m\wedge L \rightarrow \con_{B}(\instantiate(\xi, C_1))$, when $L\rightarrow \con_{L}(C_1)$ instead of $L \rightarrow \con_{L}(l^{\alpha})$. This is simplified to $B\wedge \diff_L^m \rightarrow \con_{B}(\instantiate(\xi, C_1))$ as $B\wedge \diff_L^m \rightarrow L$.

		Now we argue that $(R\rightarrow \con_R(l^\alpha))$ implies $(B\wedge \neg \diff_L\wedge \diff_R^m)\rightarrow \con_R(l^\complete{\alpha}{\restr{l}{\sigma}})$.
		
		\noindent\textbf{Suppose $i> m$. }
		
		$\diff_L\wedge \diff_R^m$ satisfies $\neg \diff_{L}^{m}\wedge(\diff_R^{m} \vee \neg x)$ so ($\val_B^i,\set_B^i$)=($\val_R^i,\set_R^i$) in all cases. 
		This means that $(\set_B^i\rightarrow (\val_B^i\leftrightarrow u_i))\rightarrow (\set_R^i\rightarrow (\val_R^i\leftrightarrow u_i))$.
		
		If $\neg \set_B^i\in \con_{l,B}(\complete{\alpha}{\restr{l}{\sigma}})$ 
		then 
		$u_i\notin\domain(\alpha)$ and
		$u_i\notin\domain(\sigma)$
		then 
		$\neg \set_R^i\in \con_{l,R}(\alpha)$
		so
		$R\rightarrow \neg \set_R^i$. and so
		$B\wedge R \wedge \neg \set_R^i \wedge \neg \diff_{L}^{m}\wedge \diff_R^{m} \rightarrow \neg \set_B^i $.
		
		If $\set_B^i\in \con_{l,B}(\complete{\alpha}{\restr{l}{\sigma}})$ 
		then 
		$u_i\in\domain(\complete{\alpha}{\restr{l}{\sigma}})$.
		Which means either $u_i\in\domain(\alpha)$ or $u_i\notin\domain(\alpha)$ and $u_i\in \sigma$. But $u_i\notin \sigma$ because $i>m$.
		Since $u_i\in\domain(\alpha)$,
		$\set_R^i\in \con_{l,R}(\alpha)$ and so
		$B\wedge R \wedge \set_R^i \wedge \neg \diff_{L}^{m}\wedge \diff_R^{m} \rightarrow \set_B^i $.
		
		If $\set_B^i\wedge\val_B^i\in \con_{l,B}(\complete{\alpha}{\restr{l}{\sigma}})$ 
		then 
		$1/u_i\in\complete{\alpha}{\sigma}$.
		Which means $1/u_i\in\alpha$,
		$\set_R^i\wedge u_i\in \con_{l,R}(\alpha)$ and so
		$B\wedge R \wedge \set_R^i\wedge \val_R^i \wedge \neg \diff_{L}^{m}\wedge \diff_R^{m} \rightarrow \set_B^i\wedge u_i $.
		If $\set_B^i\wedge\val_B^i\in \con_{l,B}(\complete{\alpha}{\restr{l}{\sigma}})$ 
		then 
		$0/u_i\in\complete{\alpha}{\sigma}$.
		Which means $0/u_i\in\alpha$ 
		$\set_R^i\wedge \neg u_i\in \con_{l,R}(\alpha)$ and so
		$B\wedge R \wedge \set_R^i\wedge \neg\val_R^i \wedge \neg \diff_{L}^{m}\wedge \diff_R^{m} \rightarrow \set_B^i\wedge \neg u_i $.

		\noindent\textbf{Suppose $j< i\leq m$. }	
		
		In this case $\neg \diff_L^{i-1}$, $\neg \diff_L^{i}$, $\diff_{R}^{i-1}$ and $\diff_{R}^{i}$ are all true. If $\set_R^i$ is true then
	
		$\neg \diff_{L}^{i-1}\wedge(\diff_R^{i-1} \vee\neg \set_L^{i})$,
		$\neg \diff_{L}^{i-1}\wedge(\diff_R^{i-1} \vee(\set_L^{i} \wedge (\val_L^{i}\leftrightarrow \val_\tau^{i})))$,
		$\neg \diff_{L}^{i-1}\wedge \set_R^{i}\wedge(\diff_R^{i-1} \vee\set_L^{i})$ and
		$\neg \diff_{L}^{i-1}\wedge(\diff_R^{i-1} \vee\neg\set_L^{i})$
		are all satisfied. So ($\val_B^i,\set_B^i$)=($\val_R^i,\set_R^i$) whenever $\set_R^i$ is true. 
		This means that $(\set_B^i\rightarrow (\val_B^i\leftrightarrow u_i))\rightarrow (\set_R^i\rightarrow (\val_R^i\leftrightarrow u_i))$.

		If $\neg \set_B^i\in \con_{l,B}(\complete{\alpha}{\restr{l}{\sigma}})$ 
		then 
		$u_i\notin\domain(\alpha)$ and
		$u_i\notin\domain(\sigma)$
		then 
		$\neg \set_R^i\in \con_{l,R}(\alpha)$
		so
		$R\rightarrow \neg \set_R^i$.
		When $\neg \diff_L^{i-1}$ and $\diff_{R}^{i-1}$ and $u_i\notin\domain(\sigma)$ then ($\val_B^i,\set_B^i$)=($\val_R^i,\set_R^i$), so 
		$B \wedge\neg \diff_L^{j}\wedge\neg \diff_L^{i}\wedge\diff_{R}^{j}\wedge \diff_{R}^{i}\wedge R\rightarrow \neg \set_B^i$.

		If $\set_B^i\in \con_{l,B}(\complete{\alpha}{\restr{l}{\sigma}})$ 
		then $*/u_i\in \complete{\alpha}{\sigma}$.
		So either $*/u_i\in \alpha$ or $*/u_i\in \sigma$ and $u_i\notin\domain(\alpha)$.
		If $*/u_i\in \alpha$ then
		$\set_R^i\in \con_{l,R}(\alpha)$
		and when $\set_R^i$ is true then ($\val_B^i,\set_B^i$)=($\val_R^i,\set_R^i$)
		so $R\rightarrow \set_R^i$ implies
		$B \wedge\neg \diff_L^{j}\wedge\neg \diff_L^{i}\wedge\diff_{R}^{j}\wedge \diff_{R}^{i}\wedge R\rightarrow \set_B^i$.
		If $*/u_i\in \sigma$ and $u_i\notin\domain(\alpha)$,
		$\neg \set_R^i\in \con_{l,R}(\alpha)$ and
		$\neg \diff_L^{i-1}\wedge\diff_R^{i-1}\wedge \neg \set_R^i$ is satisfied
		so ($\val_B^i,\set_B^i$)=($0,1$) therefore
		$B \wedge\neg \diff_L^{j}\wedge\neg \diff_L^{i}\wedge\diff_{R}^{j}\wedge \diff_{R}^{i}\wedge R\rightarrow \set_B^i$.
		
		If $\set_B^i\wedge \val_B^i\in \con_{l,B}(\complete{\alpha}{\restr{l}{\sigma}})$ 
		then $1/u_i\in \complete{\alpha}{\sigma}$.
		and it must be that $1/u_i\in \alpha$
		and so 	$\set_R^i\wedge \val_R^i\in \con_{l,R}(\alpha)$
		and when $\set_R^i$ is true then ($\val_B^i,\set_B^i$)=($\val_R^i,\set_R^i$)
		so $R\rightarrow \set_R^i\wedge u_i$ implies
		$B \wedge\neg \diff_L^{j}\wedge\neg \diff_L^{i}\wedge\diff_{R}^{j}\wedge \diff_{R}^{i}\wedge R\rightarrow \set_B^i\wedge u_i$
		
		If $\set_B^i\wedge \neg\val_B^i\in \con_{l,B}(\complete{\alpha}{\restr{l}{\sigma}})$ 
		then $0/u_i\in \complete{\alpha}{\sigma}$.
		and it must be that $0/u_i\in \alpha$
		and so 	$\set_R^i\wedge \val_R^i\in \con_{l,R}(\alpha)$
		and when $\set_R^i$ is true then ($\val_B^i,\set_B^i$)=($\val_R^i,\set_R^i$)
		so $R\rightarrow \set_R^i\wedge \neg u_i$ implies
		$B \wedge\neg \diff_L^{j}\wedge\neg \diff_L^{i}\wedge\diff_{R}^{j}\wedge \diff_{R}^{i}\wedge R\rightarrow \set_B^i\wedge \neg u_i$.

		\noindent\textbf{Suppose $i=j$. }	
		
		In this case $\neg \diff_L^{j-1}$, $\neg \diff_L^{j}$, $\neg \diff_{R}^{j-1}$ and $\diff_{R}^{j}$. If $\set_R^j$ then either ($\val_B^i,\set_B^i$)=($\val_R^i,\set_R^i$) or ($\val_B^i,\set_B^i$)=($\val_L^i,\set_L^i$). We will argue that ($\val_B^i,\set_B^i$)=($\val_L^i,\set_L^i$) is not chosen because of $\neg \diff_L^{j}$ and $\equ_R$
		$\neg \diff_{L}^{i-1}\wedge(\diff_R^{i-1} \vee\neg \set_L^{i})$ cannot be falsified because $\set_L^i$ being true would contradict $\neg \diff_L^{j}$. Likewise  $\neg \diff_{L}^{i-1}\wedge(\diff_R^{i-1} \vee(\set_L^{i} \wedge (\val_L^{i}\leftrightarrow \val_\tau^{i}))) $ cannot be falsified as $(\set_L^{i} \wedge (\val_L^{i}\leftrightarrow \val_\tau^{i})) $ being false would contradict $\neg \diff_L^{j}$.
		If $u_i\in \domain(\sigma)$ then $\neg \diff_L^{i-1}\wedge\diff_R^{i-1}\wedge \neg \set_R^i$ is false and $\neg \diff_{L}^{i-1}\wedge \set_R^{i}\wedge(\diff_R^{i-1} \vee\set_L^{i})$ is true. 
		Likewise if $u_i\in \domain(\xi )$ then $\diff_L^{i-1}\wedge \neg \set_L^i$ is false and $\neg \diff_{L}^{i-1}\wedge(\diff_R^{i-1} \vee\neg\set_L^{i})$ is true.
		The result is that $(\set_B^i\rightarrow (\val_B^i\leftrightarrow u_i))\rightarrow (\set_R^i\rightarrow (\val_R^i\leftrightarrow u_i))$.

		If $\neg\set_B^j\in \con_{l,B}(\complete{\alpha}{\restr{l}{\sigma}})$
		then $u_i\notin\domain(\complete{\alpha}{\sigma})$,
		which means $u_i\notin\domain(\alpha)$ and
		$u_i\notin\domain(\sigma)$.
		So $\neg\set_R^j\in \con_{l,R}(\alpha)$
		and thus $R\rightarrow \neg\set_R^j$.
		If $u_j\in \domain(\tau)$
		we argue that $\neg \diff_{L}^{j-1}\wedge(\diff_R^{j-1} \vee(\set_L^{j} \wedge (\val_L^{j}\leftrightarrow \val_\tau^{ij})))$ is satisfied because of $\neg \diff_{L}^{i}$. 
		Hence 
		($\val_B^i,\set_B^i$)=($\val_R^i,\set_R^i$) and so 
		$B\wedge \neg \diff_L^{j-1}\wedge\neg \diff_L^{j}\wedge\neg \diff_{R}^{j-1}\wedge\diff_{R}^{j}\wedge L \rightarrow \neg \set_B^j$.
		
		If $u_j\in \domain(\xi)$,
		we argue that $\neg \diff_{L}^{i-1}\wedge(\diff_R^{i-1} \vee\neg\set_L^{i})$ is satisfied because of $\neg \diff_{L}^{i}$ which insists on $\neg\set_L^{i}$. 
		Hence 
		($\val_B^i,\set_B^i$)=($\val_R^i,\set_R^i$) and so 
		$B\wedge \neg \diff_L^{j-1}\wedge\neg \diff_L^{j}\wedge\neg \diff_{R}^{j-1}\wedge\diff_{R}^{j}\wedge L \rightarrow \neg \set_B^j$.
		
			If $u_j\notin \domain(\tau\sqcup\sigma\sqcup\xi)$
		We argue that $\neg \diff_{L}^{i-1}\wedge(\diff_R^{i-1} \vee\neg \set_L^{i})$ is satisfied because of $\neg \diff_{L}^{i}$. 
		Hence 
		($\val_B^i,\set_B^i$)=($\val_R^i,\set_R^i$) and so 
		$B\wedge \neg \diff_L^{j-1}\wedge\neg \diff_L^{j}\wedge\neg \diff_{R}^{j-1}\wedge\diff_{R}^{j}\wedge L \rightarrow \neg \set_B^j$.

		If $\set_B^j\in \con_{l,B}(\complete{\alpha}{\restr{l}{\sigma}})$,
		then $*/u_j\in(\complete{\alpha}{\sigma})$.
		So either $*/u_j\in\alpha$ or $*/u_j\notin\alpha$ and $*/u_j\in \sigma$.
		If $*/u_j\in\alpha$ 
		then $\set_R^j\in \con_{l,R}(\alpha)$
		and $R\rightarrow \set_R^j$.
		When $\set_R^j$ is true we know 
		($\val_B^j,\set_B^j$)=($\val_R^j,\set_R^j$) 
		and so 
		$B\wedge \neg \diff_L^{j-1}\wedge\neg \diff_L^{j}\wedge\neg \diff_{R}^{j-1}\wedge\diff_{R}^{j}\wedge L \rightarrow \set_B^j$.
		If $*/u_j\notin\alpha$ and $*/u_j\in \sigma$
		then $\neg\set_R^j\in \con_{l,R}(\alpha)$
		and thus $R\rightarrow \neg\set_R^j$
		$\neg \diff_{L}^{i-1}\wedge \set_R^{j}\wedge(\diff_R^{j-1} \vee\set_L^{j})$
		is falsified. 
		So ($\val_B^j,\set_B^j$)=($\val_L^j,\set_L^i$).
		But because $\neg\diff_{L}^j$ we know that $\set_L^i$ is true therefore 
		$B\wedge \neg \diff_L^{j-1}\wedge\neg \diff_L^{j}\wedge\neg \diff_{R}^{j-1}\wedge\diff_{R}^{j}\wedge L \rightarrow \set_B^j$
		
		If $\set_B^j\wedge \val_B^j\in \con_{l,B}(\complete{\alpha}{\restr{l}{\sigma}})$,
		so $1/u_j\in(\complete{\alpha}{\sigma})$.
		So it must be that $1/u_j\in\alpha$.
		And so $\set_R^j\wedge\val_R^j\in \con_{l,R}(\alpha)$
		and thus $R\rightarrow \neg\set_R^j$
		since $\set_R^j$ is true we know that 
		($\val_B^j,\set_B^j$)=($\val_R^j,\set_R^j$) 
		and so 
		$B\wedge \neg \diff_L^{j-1}\wedge\neg \diff_L^{j}\wedge\neg \diff_{R}^{j-1}\wedge\diff_{R}^{j}\wedge L \rightarrow \set_B^j\wedge \val_B^j$
		

		If $\set_B^j\wedge \neg\val_B^j\in \con_{l,B}(\complete{\alpha}{\restr{l}{\sigma}})$,
		so $0/u_j\in(\complete{\alpha}{\sigma})$.
		So it must be that $0/u_j\in\alpha$
		And so $\set_R^j\wedge\neg\val_R^j\in \con_{l,R}(\alpha)$
		and thus $R\rightarrow \neg\set_R^j$
		since $\set_R^j$ is true we know that 
		($\val_B^j,\set_B^j$)=($\val_R^j,\set_R^j$) 
		and so 
		$B\wedge \neg \diff_L^{j-1}\wedge\neg \diff_L^{j}\wedge\neg \diff_{R}^{j-1}\wedge\diff_{R}^{j}\wedge L \rightarrow \set_B^j\wedge \val_B^j$

		\noindent\textbf{Suppose $i<j$.}	
		
		In this case $\neg \diff^{i}_L,\neg \diff^{i-1}_L, \neg \diff^{i}_R, \neg \diff^{i-1}_R$ are all true. We can see from Lemma~\ref{lem:irm:nLnR} that $\set_R^i\rightarrow\set_B^i$ in all cases. We observe all the cases when $\set_R^i$ is true and $\val_B^i$ is not defined as $\val_R^i$ and show they cannot happen.
		
		For $u_i\notin\domain(\tau\sqcup\sigma\sqcup\xi)$, if  $\neg \diff_{L}^{i-1}\wedge(\diff_R^{i-1} \vee\neg \set_L^{i})$ is false then $\set_L^i$ must be true, but this conflicts with $\neg \diff^{i}_L,\neg \diff^{i-1}_L$.	
		For $u_i\in\domain(\tau)$ if $\neg \diff_{L}^{i-1}\wedge(\diff_R^{i-1} \vee(\set_L^{i} \wedge (\val_L^{i}\leftrightarrow \val_\tau^{i})))$ is false then  $\set_L^{i} \rightarrow (\val_L^{i}\oplus \val_\tau^{i})$ is contradicting $\neg \diff^{i}_L,\neg \diff^{i-1}_L$.
		For $u_i \in \domain(\sigma)$ if $\neg \diff_{L}^{i-1}\wedge \set_R^{i}\wedge(\diff_R^{i-1} \vee\set_L^{i})$ is false the then $\set_L^{i}$ is false contradicting $\neg \diff^{i}_L,\neg \diff^{i-1}_L$.
		For $u_i \in \domain(\xi)$
		if
		$\neg \diff_{L}^{i-1}\wedge(\diff_R^{i-1} \vee\neg\set_L^{i})$
		is false then $\set_L^{i}$ is true but in $\domain(\xi)$ this contradicts 
		$\neg \diff^{i}_L,\neg \diff^{i-1}_L$.
		Therefore $(\set_B^i\rightarrow (\val_B^i\leftrightarrow u_i))\rightarrow (\set_R^i\rightarrow (\val_R^i\leftrightarrow u_i))$.
		
		If $\neg\set_B^i\in \con_{l,B}(\complete{\alpha}{\restr{l}{\sigma}})$ then 
		$u_j\notin\domain(\complete{\alpha}{\sigma})$ and so
		$u_j\notin\domain(\alpha)$ and
		$u_j\notin\domain(\sigma)$. 
		So $\neg\set_R^i\in \con_{l,R}(\alpha)$
		and $R\rightarrow \neg\set_R^i$.
		$\neg \diff^{i}_R,\neg \diff^{i-1}_R$ means that 
		$u_i\notin\domain(\tau\sqcup\sigma\sqcup\xi) $.
		From Lemma~\ref{lem:irm:nLnR} we know $\neg \diff^{i}_L\wedge\neg \diff^{i}_R\rightarrow\neg \set^i_B$. So $B\wedge R \wedge \neg \diff^{j}_R \wedge \neg \diff^{i}_L \rightarrow \neg\set_B^i$.		
		If $\set_B^i\in \con_{l,B}(\complete{\alpha}{\restr{l}{\sigma}})$.
		Either $*/u_i\in \alpha$ or  $u_i\not\in \domain(\alpha)$ and $*/u_i\in \sigma$.
		If $*/u_i\in \alpha$,
		then $\set^i_{R}\in \con_{l,R}(\alpha)$
		and $R\rightarrow \set_R^i$.
		By Lemma~\ref{lem:irm:nLnR}, $u_i$ must be in $\domain(\tau)$
		or $\domain(\xi)$.		
		In either case $\set_B^i$ is true. 
		So $B\wedge R \wedge \diff^{j}_R \wedge \neg \diff^{i}_L \rightarrow \set_B^i$
		If $u_i\not\in \domain(\alpha)$ and $*/u_i\in \sigma$,
		then $\neg \set^i_{R}\in \con_{l,R}(\alpha)$
		and $R\rightarrow \neg \set_R^i$
		By Lemma~\ref{lem:irm:nLnR}, $\set_R^i$ is true.
		So $B\wedge R \wedge \diff^{j}_R \wedge \neg \diff^{i}_L \rightarrow \set_B^i$.
		
		If $\set_B^i\wedge\val_B^i\in \con_{l,B}(\complete{\alpha}{\restr{l}{\sigma}})$
		then $1/u_i\in \complete{\alpha}{\sigma}$, so it must be that 
		$1/u_i\in \alpha$. And so
		$\set^i_{R}\wedge \val_R^i\in \con_{l,R}(\alpha)$.
		By Lemma~\ref{lem:irm:nLnR}, $u_i$ must be in $\domain(\tau)$
		or $\domain(\xi)$.
		In either case ($\val_B^i,\set_B^i$)=($\val_R^i,\set_R^i$).
		So $B\wedge R \wedge \diff^{j}_R \wedge \neg \diff^{i}_R \rightarrow \set_B^i\wedge \val_B^i$, because $R\rightarrow \set_R^i\wedge\val_R^i$.
	\end{sloppypar}		
		Likewise, if $\set_B^i\wedge\neg u_i\in \con_{l,B}(\complete{\alpha}{\restr{l}{\sigma}})$
		then $0/u_i\in \complete{\alpha}{\sigma}$, so it must be that 
		$0/u_i\in \alpha$. And so
		$\set^i_{R}\wedge \val_R^i\in \con_{l,R}(\alpha)$.
		By Lemma~\ref{lem:irm:nLnR}, $u_i$ must be in $\domain(\tau)$
		or $\domain(\xi)$.
		In either case ($\val_B^i,\set_B^i$)=($\val_R^i,\set_R^i$).
		So $B\wedge R \wedge \diff^{j}_R \wedge \neg \diff^{i}_R \rightarrow \set_B^i\wedge \neg u_i$, because $R\rightarrow \set_R^i\wedge\neg u_i$.
		With that we conclude all cases in $R$ and argue similarly to $L$.
\end{proof}

\renewcommand{\thesection}{\Alph{section}}
\setcounter{section}{1}
\subsection{Local Strategy Extraction for Simulation of \lquprc}\label{app:lquprc}

\subsubsection{Policy Variables}

For $u_i^*\notin C_1 \cup C_2$, $i\leq m$

\noindent$(\val^i_B, \set^i_B)=
\begin{cases} 
(\val^i_R, \set^i_R)& \text{if } \neg \diff_{L}^{i-1}\wedge(\diff_R^{i-1} \vee\neg \set_L^{i})\\ 
(\val^i_L, \set^i_L) & \text{otherwise. }
\end{cases}$

{For $u_i^*\in C_1$}, $i\leq m$

\noindent$(\val^i_B, \set^i_B)= 
\begin{cases}
(0,1)  & \text{if }\neg \diff_L^{i-1}\wedge\diff_R^{i-1}\wedge \neg \set_R^i\\
(\val^i_R, \set^i_R)& \text{if } \neg \diff_{L}^{i-1}\wedge \set_R^{i}\wedge(\diff_R^{i-1} \vee\set_L^{i})\\
(\val^i_L, \set^i_L) & \text{otherwise. }
\end{cases}$

{For $u_i^*\in C_2$}, $i\leq m$

\noindent$(\val^i_B, \set^i_B)= 
\begin{cases}
(0,1)  & \text{if }\diff_L^{i-1}\wedge \neg \set_L^i\\
(\val^i_R, \set^i_R)& \text{if } \neg \diff_{L}^{i-1}\wedge(\diff_R^{i-1} \vee\neg\set_L^{i})\\
(\val^i_L, \set^i_L) & \text{otherwise. }
\end{cases}$

For $u_i\in \domain(U^*)$, $i>m$

\noindent$(\val^i_B, \set^i_B)= 
\begin{cases}
(\val^i_R, \set^i_R)&\text{if } \set_R^i\wedge \neg \diff_{L}^{m}\wedge(\diff_R^{m} \vee \neg x)  \\
(0, 1)&\text{if } u_i\in U_2\text{ and }\neg \set_R^i\wedge \neg \diff_{L}^{m}\wedge(\diff_R^{m} \vee \neg x)\\
(1, 1)&\text{if } \neg u_i\in U_2\text{ and }\neg \set_R^i\wedge \neg \diff_{L}^{m}\wedge(\diff_R^{m} \vee \neg x)\\
(0, 1)&\text{if } u_i^*\in U_2\text{ and }\neg \set_R^i\wedge \neg \diff_{L}^{m}\wedge(\diff_R^{m} \vee \neg x)\\
(\val^i_L, \set^i_L) &\text{if } \set_L^i\wedge \diff_{L}^{m}\vee(\neg \diff_R^{m} \wedge x) \\
(0, 1)&\text{if } u_i\in U_1\text{ and }\neg \set_L^i\wedge (\diff_{L}^{m}\vee(\neg \diff_R^{m} \wedge x))\\
(1, 1)&\text{if } \neg u_i\in U_1\text{ and }\neg \set_L^i\wedge (\diff_{L}^{m}\vee(\neg \diff_R^{m} \wedge x))\\
(0, 1)&\text{if } u_i^*\in U_1\text{ and }\neg \set_L^i\wedge( \diff_{L}^{m}\vee(\neg \diff_R^{m} \wedge x))\\
\end{cases}$

For $u_i\notin \domain(U)$, $i>m$

\noindent$(\val^i_B, \set^i_B)= 
\begin{cases}
(0,1) &\text{if } u^*\in V_2 \text{ and } \neg \set_L^i\wedge (\diff_L^m \vee (\neg \diff_R^{m} \wedge x)) \\
(\val^i_L, \set^i_L) &\text{if } u^*\in V_2 \text{ and } \set_L^i\wedge (\diff_L^m \vee (\neg \diff_R^{m} \wedge x)) \\
(\val^i_R, \set^i_R) &\text{if } u^*\in V_2 \text{ and } \neg \diff_L^m \wedge (\diff_R^{m} \vee \neg x) \\
(0,1) &\text{if } u^*\in V_1 \text{ and } \neg \set_R^i\wedge (\neg \diff_L^m \wedge (\diff_R^{m} \vee \neg x))\\
(\val^i_R, \set^i_R) &\text{if } u^*\in V_1 \text{ and } \set_R^i\wedge (\neg \diff_L^m \wedge (\diff_R^{m} \vee \neg x) )\\
(\val^i_L, \set^i_L) &\text{if } u^*\in V_1 \text{ and } \diff_L^m \vee (\neg \diff_R^{m} \wedge x) \\
(\val^i_R, \set^i_R)& \text{if } u^*\notin V_1\cup V_2 \text{ and } \neg \diff_{L}^{m}\wedge(\diff_R^{m} \vee \neg x)\\
(\val^i_L, \set^i_L) & \text{if } u^*\notin V_1\cup V_2 \text{ and } \diff_L^m \vee (\neg \diff_R^{m} \wedge x)
\end{cases}$
\renewcommand{\thesection}{\arabic{section}}
\setcounter{section}{6}
\setcounter{thm}{5}
\begin{lem}
	The following propositions are true and have short Extended Frege proofs, given $(L\rightarrow \con_L(C_1\cup  U_1\vee \neg x))$ and $(R\rightarrow \con_R(C_2\cup  U_2\vee x))$
	\begin{itemize}
		\item $B\wedge \diff_L^m\rightarrow L$
		\item $B\wedge \neg \diff_L^m\wedge \diff_R^m\rightarrow R$
		\item $B\wedge \diff_L^m\rightarrow \con_B(C_1\vee V_2 \vee U)$
		\item $B\wedge \neg \diff_L^m\wedge \diff_R^m\rightarrow \con_B(C_2\vee V_1 \vee U)$
	\end{itemize}
	
\end{lem}

\begin{proof}
	\begin{sloppypar}
		We break $B\wedge \diff_L^m\rightarrow L$ into individual parts $\set_B^i \rightarrow (u_i\leftrightarrow \val^i_B)\wedge \diff_L^m  \rightarrow (\set_L^i\rightarrow (u_i\leftrightarrow \val_L^i))$ which we join by conjunction. We can do similarly for 
		$B\wedge \neg \diff_L^m\wedge \diff_R^m\rightarrow R$.
		
	For $B\wedge \diff_L^m\rightarrow \con_B(C_1\vee V_2 \vee U^*)$ we
	first derive
	$(L\rightarrow \con_{L}(C_1\cup U_1 \vee \neg x)) \rightarrow (B\wedge L \wedge \diff_L^m \rightarrow \con_B(C_1\vee V_2 \vee U^*) ) $, you can cut out $L$ using $B\wedge \diff_L^m\rightarrow L$. Removing $(L\rightarrow \con_{L}(C_1\cup U_1 \vee \neg x))$, uses the premise $(L\rightarrow \con_L(C_1\cup  U_1\vee \neg x))$. 
	
	To derive 	$(L\rightarrow \con_{L}(C_1\cup U_1 \vee \neg x)) \rightarrow (B\wedge L \wedge \diff_L^m \rightarrow \con_B(C_1\vee V_2 \vee U^*) ) $ we
	break this by non-starred literals $l\in C_1\cup U_1$ so we will show that $(L\rightarrow \con_{L,C_1\cup U_1 \vee \neg x}(l))\rightarrow(B\wedge \diff_L^m\rightarrow \con_{B,V_2\cup C_1\cup U^*}(l)) $.		$\diff^m_{L}\rightarrow \neg \ann_{x, L}(V_1)$ is used to remove the $x$ literal.
		
		For $p\in\{1,2\}$ let $W_p=\{u^* \mid u^*\in U_p \}$.
		For each $i$, either $\set_B^i$ or $\neg \set_B^i$ appears in $\ann_{l,B}(V_1\cup V_2 \cup U^*)$, so we treat $\ann_{l,B}(V_1\cup V_2 \cup U^*)$ as a set containing these subformulas. We show that if $c_i\in \ann_{l,B}(V_1\cup V_2\cup U^*)$ when $c_i= \set_B^i$ or $c_i= \neg\set_B^i$ then $L\rightarrow \ann_{l,L}(V_1\cup W_1)\rightarrow B \wedge \diff_{L}^m \rightarrow c_i$ and we also have $(L\rightarrow l)\rightarrow (B \wedge \diff_{L}^m \rightarrow l)$.
		
		For existential $l$, we can put these all together to get $(L\rightarrow \con_{L,C_1\cup U_1}(l))\rightarrow(B\wedge L \wedge \diff_L^m\rightarrow \con_{B,V_2\cup C_1\cup U^*}(l)) $. 

		For universal literals $u_k$ we also need to show $\neg \set^k_B$ is preserved when $u_k$ is not merged.
		For universal literals $u_k$ that are merged $\con_{B,V_2\cup C_1\cup U^*}(u_k^*)) = \bot$ so we show that the strategy for $B$ causes a contradiction between $B$ and $L\rightarrow u_k$.
		We do similarly for $B\wedge \neg \diff_L^m\wedge \diff_R^m\rightarrow \con_B(C_2\vee V_1 \vee U^*)$.

			\noindent\textbf{The $\diff_L^m$ cases.}
		If $\diff_L^m$ is true then there is some $j$ such that $\diff^j_L\wedge \neg \diff^j_L\wedge \neg \diff^j_R$ via Lemmas~\ref{lem:lquprc:chain}~and~\ref{lem:lquprc:impl}. We use the disjunction $\diff_L^m \rightarrow \bigvee_{j=1}^m \diff_L^j \wedge \neg \diff_L^{j-1}$ to join all the cases of $j$ together.

		\noindent\textbf{Suppose $i>m$.}
		
		$\diff_L^i$ satisfies $\diff_L^m \vee (\neg \diff_R^{m}\wedge x )$ so whenever $\diff^m_L$ is true and $\set_L^i$ is true, $(\val_B^i,\set_B^i)=(\val_L^i,\set_L^i)$, 
		therefore $(\set_B^i\rightarrow (u_i\leftrightarrow \val_B^i))\rightarrow (\set_L^i\rightarrow (u_i\leftrightarrow \val_L^i))$.
		
		If $\set^i_B\in \ann_{x,B}(V_1\cup V_2\cup U^*)$, 
		either $u^*_i \in V_1$, $u^*_i \in V_2$ or $u^*_i\in U^*$. If $u^*_i \in V_2$ then every possibility we have $\set_B^i$ be true. If $u^*_i \in V_1$ we know $\set_L^i$ will be true since it is assumed to be implied by $L$, hence $\set_B^i=\set_L^i$ suffices.
		If $u_i^*\in  U^*$ every case $\set^i_B$ is true when $\ann_{x,L}(V_1\cup W_1)$ is affirmed by $L$.
		
		If $\neg\set^i_B\in \ann_{x,B}(V_1\cup V_2\cup U^*)$ then $u_i^*\notin  V_1\cup V_2\cup U^*$, this means that $u_i^*\notin W_1$, so whenever $\ann_{x,L}(V_1\cup W_1)$ is true, $\neg \set^i_L$. But then $\neg \set^i_B$ must be true because of $\diff_L^m$.
		
		\noindent\textbf{Suppose $j<i\leq m $.}
		
		We know $\diff^j_L\rightarrow \diff^{i-1}_L$ from Lemma~\ref{lem:lquprc:impl}, we will use that to get that when $\diff^j_L\wedge \set^i_L$ then $(\val^i_B, \set^i_B)=(\val^i_L,\set^i_L)$ which allows us to then show $(\set_B^i\rightarrow (u_i\leftrightarrow \val_B^i))\rightarrow (\set_L^i\rightarrow (u_i\leftrightarrow \val_L^i))$.
		
		Suppose $\neg \set_B^i\in \ann_{x,B}(V_1\cup V_2\cup U^*)$, then $u_i^*\notin  C_1\cup C_2$ so $(\val^i_B, \set^i_B)=(\val^i_L,\set^i_L)$. But since $\set^i_L$ will be false because $u_i^*\notin  C_1 $, $\set^i_B$ will be false.

		Now suppose $ \set_B^i\in \ann_{x,B}(V_1\cup V_2\cup U^*)$, either $u_i^*\in C_1$, 
		in which case $(\val^i_B, \set^i_B)=(\val^i_L,\set^i_L)$,
		or $u_i^*\in C_2$,
		in which case $(\val^i_B, \set^i_B)=(\val^i_L,\set^i_L)$ or $\neg \set_L^i$, but here we know $\set_B^i$ will be forced to be true, regardless.

		\noindent\textbf{Suppose $i=j$.}
		
		$\diff^i_L$,  $\neg \diff^{i-1}_L$ and $\neg \diff^{i-1}_R$ are all true.
		If $\set^i_L\in \ann_{x,L}(V_1\cup W_1) $ then $\neg \set^i_L$, and if $\neg \set^i_L\in \ann_{x,L}(V_1\cup W_1) $ then $\set^i_L$. 
		
		If $\set^i_L\in \ann_{x,L}(V_1\cup W_1) $ and $\neg \set^i_L$ then $u_i^*\in C_1$ and so $(\val^i_B, \set^i_B)=(\val^i_L,\set^i_L)$. So if $\ann_{x,L}(V_1\cup W_1)$ is satisfied by $L$ the term $\set^i_B\in \ann_{x,L}(V_1\cup V_2 \cup U^*)$ is satisfied by $B$.
		If $\neg \set^i_L\in \ann_{x,L}(V_1\cup W_1) $ and $ \set^i_L$ then if $u_i^*\in C_2$, we know $\set^i_B\in \ann_{x,L}(V_1\cup V_2 \cup U^*)$, since $\set_L^i$ is true then $\set_B^i$ is true.

		If $u_i^*\notin C_1 \cup C_2$ then $\neg \set^i_B\in \ann_{x,B}(V_1\cup V_2 \cup U^*)$, but then $(\val^i_B, \set^i_B)=(\val^i_L,\set^i_L)$. So if $\ann_{x,L}(V_1\cup W_1)$ is satisfied by $L$ the term $\set^i_B\in \ann_{x,L}(V_1\cup V_2 \cup U^*)$ is satisfied by $B$.
		
		\noindent\textbf{Suppose $i<j$.}
		
		If $\neg \set^i_B\in \ann_{x,B}(V_1\cup V_2 \cup U^*)$ then $u^*\notin C_1\cup C_2$ and so by Lemma~\ref{lem:lquprc:nLnR} $\neg \set^i_B$ is true.
		If $\set^i_B\in \ann_{x,B}(V_1\cup V_2 \cup U^*)$ then $u^*\in C_1\cup C_2$ and so by Lemma~\ref{lem:lquprc:nLnR}, $\set^i_B$ is true.
		
		We can put this all together to show in \eFrege that 
		$B\wedge \diff^m_{L}\rightarrow L$,
		$L\rightarrow \con_{L,C_1\vee U_1\vee \neg x}(l)\rightarrow B \wedge L \wedge \diff^m_{L} \rightarrow \con_{B,C_2\vee V_2\vee U^*}(l)$, for existential literal $l$.
		Note that $\diff_{L}$ means that $\con_{R,C_2\cup U_2\vee x,R}(\neg x)$ is not satisfied by $L$ to begin with. 
		
		\noindent\textbf{Additional universal consideration.}
		
		If $l=u_k$,  then when $l$ does not become merged we also have to show that $\neg \set^k_B $ is preserved when $\con_{L,C_1\cup U_1\vee x}(l)$ and $\diff_L^m$.  Note that if $\diff^k_L$ then the annotation is contradicted. 
		If $u_k\in C_1\vee C_2$ or $\neg u_k\in C_1\vee C_2$, for $i\leq m$ then $\neg \set^i_B$ is desired, but $\set^i_B$ will only happen when forced by $\set^i_R$ being true, but this would mean $\diff_R^k$ and $\neg \diff_L^k$, which would contradict $\diff_L^m$.
		If $u_k\in C_1\vee C_2$ or $\neg u_k\in C_1\vee C_2$ for $i> m$ then $\diff_L^m$ will contradict an annotation.
		$u_k \in U_1$ then the literal will not appear as such in $\con_B(C_1\cup C_2 \cup U^*)$ because it will now only count as a starred literal.
		
		We have to show any universal literal $l= u_k$ or $l= \neg u_k$ that does become merged, can be removed from the disjunction. In essence we need to prove 
		$(L\rightarrow \con_{L,C_1\cup U_1}(l))\rightarrow(B\wedge L \wedge \diff_L^m\rightarrow (\bot)) $. The essential part is that $\con_{L,C_1\cup U_1}(l)$ contains $l$ but also contains  $\set_L^k$ which in turn guarantees $\set_B^k$ and forces $\val_B^k$ to be the opposite value of $l$.

		\noindent\textbf{The $\diff_R^m\wedge \neg \diff_L^m$ cases.}
		
			If $\diff_R^m$ is true then there is some $j$ such that $\diff^j_R\wedge \neg \diff^j_R\wedge \neg \diff^j_L$ via Lemmas~\ref{lem:lquprc:chain}~and~\ref{lem:lquprc:impl}. We use the disjunction $\diff_R^m \rightarrow \bigvee_{j=1}^m \diff_R^j \wedge \neg \diff_R^{j-1}$ to join all the cases of $j$ together.
			
		\noindent\textbf{Suppose $i>m$.}
		
		$\diff^m_R\wedge\neg\diff_L^m$ satisfies $ \neg \diff_L^m \wedge (\diff_R^{m}\vee \neg x )$ so whenever $\diff^m_R\wedge\neg\diff_L^m$  is true and $\set_R^i$ is true $(\val_B^i,\set_B^i)=(\val_R^i,\set_R^i)$, therefore $(\set_B^i\rightarrow (u_i\leftrightarrow \val_B^i))\rightarrow (\set_R^i\rightarrow (u_i\leftrightarrow \val_R^i))$.
		
		If $\set^i_B\in \ann_{x,B}(V_1\cup V_2\cup U^*)$, then $u_i^*\in  V_1$, $u_i^*\in  V_2$ or $u_i^*\in  U^*$. In every case $\set^i_B$ is true when $\ann_{x,R}(V_2\cup W_2)$ is affirmed by $R$ and $\diff^m_R\wedge\neg\diff_L^m$ is true.
		
		If $\neg\set^i_B\in \ann_{x,B}(V_1\cup V_2\cup U^*)$ then $u_i^*\notin  U$, this means that $u_i^*\notin W_2$, so whenever $\ann_{x,R}(V_2\cup W_2)$ is true, $\neg \set^i_R$. But then $\neg \set^i_B$ must be true because of $\diff^m_R\wedge\neg\diff_L^m$.
		
		\noindent\textbf{Suppose $j<i\leq m $.}
		
		We know $\diff^j_R\rightarrow \diff^{i-1}_R$ and $\neg \diff^m_R\rightarrow \neg \diff^{i-1}_R$ from Lemma~\ref{lem:lquprc:impl},
		we will use that to get that when $\diff^j_R\wedge \neg \diff^m_L \wedge \set^i_R$ then $(\val^i_B, \set^i_B)=(\val^i_R,\set^i_R)$ which allows us to then show $(\set_B^i\rightarrow (u_i\leftrightarrow \val_B^i))\rightarrow (\set_R^i\rightarrow (u_i\leftrightarrow \val_R^i))$.
		
		Suppose $\neg \set_B^i\in \ann_{x,B}(V_1\cup V_2\cup U^*)$, then $u_i^*\notin  C_1\cup C_2$ so $(\val^i_B, \set^i_B)=(\val^i_R,\set^i_R)$. But since $\set^i_R$ will be false because $u_i^*\notin  C_2 $, $\set^i_B$ will be false.
		
		Now suppose $ \set_B^i\in \ann_{x,B}(V_1\cup V_2\cup U^*)$, either $u_i\in C_2$ in which case $(\val^i_B, \set^i_B)=(\val^i_R,\set^i_R)$, but since $u_i\in C_2$ $\val^i_R$ must be true, or $u_i\in C_1$ in which case $(\val^i_B, \set^i_B)=(\val^i_R,\set^i_R)$ or $\neg \set_R^i$, but here we know $\set_B^i$ will be forced to be true.

		\noindent\textbf{Suppose $i=j$.}
		
		$\diff^i_R$  $\neg \diff^{i-1}_R$, $\neg \diff^{i}_L$ and $\neg \diff^{i-1}_L$ are all true.
		If $\set^i_R\in \ann_{x,R}(V_2\cup W_2) $ then $\neg \set^i_R$, and if $\neg \set^i_R\in \ann_{x,R}(V_2\cup W_2) $ then $\set^i_R$. 
		If $\set^i_R\in \ann_{x,R}(V_2\cup W_2) $ and $\neg \set^i_R$ then $u_i^*\in C_2$ and $u_i\notin C_1$. 
		$\neg \diff^{i}_L$ and $\neg \diff^{i-1}_L$ means that $\neg \set_L^i$, so then $(\val^i_B, \set^i_B)=(\val^i_R,\set^i_R)$.
		So if $\ann_{x,L}(V_1\cup W_1)$ is satisfied by $R$ the term $\set^i_B\in \ann_{x,L}(V_1\cup V_2 \cup U^*)$ is satisfied by $B$.
		
		If $\neg \set^i_R\in \ann_{x,R}(V_R\cup W_R) $ and $ \set^i_R$ then if $u_i^*\in C_1$, we know $\set^i_B\in \ann_{x,L}(V_1\cup V_2 \cup U^*)$, $\neg \diff^{i}_L$ and $\neg \diff^{i-1}_L$ means that $\set_L^i$ is true,	
		since $\set_R^i$ is also true then $\set_B^i$ is true.
		If $u_i^*\notin C_1 \cup C_2$ then $\neg \set^i_B\in \ann_{x,B}(V_1\cup V_2 \cup U^*)$, 
		$\neg \diff^{i}_L$ and $\neg \diff^{i-1}_L$ means that $\set_L^i$ is true,
		so then $(\val^i_B, \set^i_B)=(\val^i_R,\set^i_R)$. So if $\ann_{x,R}(V_2\cup W_2)$ is satisfied by $R$ the term $\set^i_B\in \ann_{x,B}(V_1\cup V_2 \cup U^*)$ is satisfied by $B$.

		\noindent\textbf{Suppose $i<j$.}
		
		If $\neg \set^i_B\in \ann_{x,B}(V_1\cup V_2 \cup U^*)$ then $u^*\notin C_1\cup C_2$ and so by Lemma~\ref{lem:lquprc:nLnR} $\neg \set^i_B$ is true.
		If $\set^i_B\in \ann_{x,B}(V_1\cup V_2 \cup U^*)$ then $u^*\in C_1\cup C_2$ and so by Lemma~\ref{lem:lquprc:nLnR}, $\set^i_B$ is true.
		
		We can put this all together to show in \eFrege that 
		$B\wedge \diff^m_{R}\wedge \neg \diff^m_{L}\rightarrow R$,
		$R\rightarrow \con_{R,C_2\vee U_2\vee x}(l)\rightarrow B \wedge R\wedge \diff^m_{R}\wedge \neg \diff^m_{L} \rightarrow \con_{B,C_2\vee V_2\vee U^*}(l)$, for existential literal $l$.
		Note that $\diff_{R}$ means that $\con_{R,C_2\cup U_2\vee x,R}(x)$ is not satisfied by $R$ to begin with. 
		
		\noindent\textbf{Additional universal consideration.}
		
		If $l=u_k$ then we also have to show that $\neg \set^k_B $ is preserved when $\con_{R,C_2\cup U_2\vee x,L}(y)$ and $\diff_R^m\wedge \neg \diff_L^m$. Note that if $\diff^k_R$ then the annotation is contradicted. 
		If $u_k\in C_1\vee C_2$ or $\neg u_k\in C_1\vee C_2$, for $i\leq m$ then $\neg \set^i_B$ is desired, but $\set^i_B$ will only happen when forced by $\set^i_L$ being true, but this would mean $\diff_L^k$ contradicting $\neg \diff_L^m$
		If $u_k\in C_1\vee C_2$ or $\neg u_k\in C_1\vee C_2$ for $i> m$ then $\diff_L^m$ will contradict an annotation.
		$u_k \in U_1$ then the literal will not appear as such in $\con_B(C_2\vee V_2\vee U^*)$ because it will now only count as a starred literal.
		
\end{sloppypar}		
		We have to show any universal literal $l= u_k$ or $l= \neg u_k$ that does become merged, can be removed from the disjunction. In essence we need to prove 
		$(R\rightarrow \con_{R,C_2\cup U_2}(l))\rightarrow(B\wedge R \wedge \neg \diff_L^m\wedge \diff_R^m\rightarrow (\bot)) $. The essential part is that $\con_{L,C_1\cup U_1}(l)$ contains $l$ but also contains  $\set_L^k$ which in turn guarantees $\set_B^k$ and forces $\val_B^k$ to be the opposite value of $l$. 
\end{proof}

\end{document}